\documentclass{lmcs}

\usepackage[english]{babel}

\PassOptionsToPackage{olditem,oldenum}{paralist}

\usepackage[alg,fig,nofnttls,noknwtls]{fmocdmac}

\usetikzlibrary{positioning}
\usepackage{cleveref}
\usepackage{todonotes}

\usepackage{array}
\usepackage{tabularx}

\newcolumntype{Y}{>{\raggedright\arraybackslash}X}

\AtEndPreamble {

}

\definecolor{revisionblue}{rgb}{0,0.361,0.686}

\theoremstyle{plain}
\newtheorem{theorem}[thm]{Theorem}
\newtheorem{lemma}[thm]{Lemma}
\newtheorem{proposition}[thm]{Proposition}
\newtheorem{corollary}[thm]{Corollary}
\newtheorem{claim}[thm]{Claim}

\theoremstyle{definition}
\newtheorem{problem}{Problem}
\newtheorem{example}{Example}
\newtheorem{definition}[thm]{Definition}

\theoremstyle{remark}

\theoremstyle{plain}

\crefname{definition}{Def.}{Defs.}
\crefname{problem}{Prob.}{Probs.}
\crefname{construction}{Con.}{Cons.}
\crefname{proposition}{Prop.}{Props.}
\crefname{lemma}{Lem.}{Lemm.}
\crefname{claim}{Clm.}{Clms.}
\crefname{theorem}{Thm.}{Thms.}
\crefname{corollary}{Cor.}{Cors.}
\crefname{conjecture}{Conj.}{Conjs.}
\crefname{remark}{Rem.}{Rems.}
\crefname{example}{Exm.}{Exms.}

\crefname{section}{Sec.}{Secs.}
\crefname{figure}{Fig.}{Figs.}

\cmdtxtabr{CQ}\cmdtxtabr{cq}
\cmdtxtabr{BCQ}\cmdtxtabr{bcq}
\cmdtxtabr{PF}\cmdtxtabr{pf}
\cmdtxtabr{BPF}\cmdtxtabr{bpf}
\cmdtxtabr{PFQ}\cmdtxtabr{pfq}
\cmdtxtabr{BPFQ}\cmdtxtabr{bpfq}
\cmdtxtabr{JoF}\cmdtxtabr{jof}
\cmdtxtabr{BJoF}\cmdtxtabr{bjof}
\cmdtxtabr{JoFQ}\cmdtxtabr{jofq}
\cmdtxtabr{BJoFQ}\cmdtxtabr{bjofq}
\cmdtxtabr{JU}\cmdtxtabr{ju}
\cmdtxtabr{BJU}\cmdtxtabr{bju}
\cmdtxtabr{JUQ}\cmdtxtabr{juq}
\cmdtxtabr{BJUQ}\cmdtxtabr{bjuq}
\cmdtxtabr{RJU}\cmdtxtabr{rju}
\cmdtxtabr{BRJU}\cmdtxtabr{brju}
\cmdtxtabr{RJUQ}\cmdtxtabr{rjuq}
\cmdtxtabr{BRJUQ}\cmdtxtabr{brjuq}
\cmdtxtabr{IJU}\cmdtxtabr{iju}
\cmdtxtabr{BIJU}\cmdtxtabr{biju}
\cmdtxtabr{IJUQ}\cmdtxtabr{ijuq}
\cmdtxtabr{BIJUQ}\cmdtxtabr{bijuq}

\cmdtxtabr{UCQ}\cmdtxtabr{ucq}

\cmdtxtabr{CP}
\cmdtxtabr{BCP}

\cmdtxtabr{MPI}
\cmdtxtabr{GMPI}
\cmdtxtabr{PIOM}

\cmdmthsetext{Exp}[Ex][e]
\seqoflet{expVec}{mthvec}
\newcommand{\expVec}{\eexpVec}

\cmdmthsetext{Unk}[Un][u]
\seqoflet{unkVec}{mthvec}
\newcommand{\unkVec}{\uunkVec}

\cmdmthsetext{Sol}[Sl][\xi]
\seqoflet{solVec}{mthvec}
\newcommand{\solVec}{\xisolVec}

\seqoflet{varVec}{mthvec}

\cmdmthsetext{Atm}[][a]
\cmdmthsetext{Qry}[Qr][q]
\cmdmthsetext{PrbTup}[PT][\vec*{t}]

\cmdmthsetext{DB}[DB][I]
\cmdmthset{IF}\cmdmthfun{if}
\cmdmthset{NIF}\cmdmthfun{nif}

\cmdtxtabr{poset}
\cmdtxtabr{cmsem}

\cmdtxtabr{MUC}
\cmdtxtabr{muc}

\DeclareRobustCommand{\bb}{\bSym\bSym}
\DeclareRobustCommand{\bs}{\bSym\sSym}
\DeclareRobustCommand{\bbinc}{\mthlrel{\sqsubseteq}[\bb]}
\DeclareRobustCommand{\bsinc}{\mthlrel{\sqsubseteq}[\bs]}

\cmdmthset{CP}
\cmdmthset{BCP}
\cmdmthset{JoFCP}[JoF\text{-}\,CP]
\cmdmthset{JoFBCP}[JoF\text{-}\,BCP]

\cmdmthfun{img}
\cmdmthfun{nimg}

\usrmth{comp}{}{argfun}
\usrmth{Comp}{}{argfun}

\usrmth{mce}{}{argfun}

\usrmth{Prof}{}{argfun}
\usrmth{Emb}{}{argfun}

\usrmth{subcls}{}{argfun}

\cmdmthset{Aut}
\cmdmthset{Iso}
\cmdmthsetext{Hom}[][h]

\usrmth{size}{}{argfun}

\usrmth{mgu}{}{argfun}
\usrmth{mgt}{}{argfun}
\usrmth{dmgu}{}{argfun}
\usrmth{dmgt}{}{argfun}

% ...

\cmdmthset{JOFQ}[JoFQ]\cmdmthset{BJOFQ}[BJoFQ]
\cmdmthset{CP}\cmdmthset{BCP}
\cmdmthset{JOFCP}[JoFCP]\cmdmthset{BJOFCP}[BJoFCP]

\usrmth{atm}{}{argfun}

\newcommandx{\MDI}
  {\txtabr[RmUpMd]{MDI}}

\newcommandx{\MRI}
  {\txtabr[RmUpMd]{MRI}}

\newcommand{\polsym}{P}
\newcommandx{\polSym}[3][1=, 2=, 3=]
  {\mthsym{\polsym#3}[#1][#2]}

\newcommand{\polelm}{P}
\newcommandx{\polElm}[3][1=, 2=, 3=]
  {\mthelm{\polelm#3}[#1][#2]}

\newcommand{\monsym}{M}
\newcommandx{\monSym}[3][1=, 2=, 3=]
  {\mthsym{\monsym#3}[#1][#2]}

\newcommand{\monelm}{M}
\newcommandx{\monElm}[3][1=, 2=, 3=]
  {\mthelm{\monelm#3}[#1][#2]}

\newcommand{\homFun}
  {\mthfun{h}}

\newcommand{\subFun}
  {\mthfun{\sigma}}

\newcommand{\adom}
  {\mthfun{adom}}

\newcommand{\bagFun}
  {\mthfun{\mu}}

\newcommand{\bodyFun}
  {\mthfun{body}}

% ...

\AfterEndPreamble{

\newcommand{\algnmpisol}{
\begin{algorithm}
  \caption{\label{alg:nmpisol} $\exists\forall$-Alternating $n$-MPI Solver.}
  \Function{$\sol(\MFun(\unkVec), \PFun(\unkVec), \PStr)$} {
    \nl $\phi \gets \text{facet complexity of the } \tuple* {\MFun(\unkVec)
        \!>\! \PFun[+](\unkVec)} {\PStr} [+]\text{-system}\hspace{-1.75em}$ \;
    \nl $\dsolVec \gets \ExsGuess \,\From \,\set{ \dsolVec \in \SetN[+][n] }{
        \size!{\dsolVec} \leq 6 \phi n^{3} }$
        \label{alg:nmpisol(exs:sol)} \;
    \nl $i \gets \AllGuess \,\From \,\set{ i \in \num{m} }{ \alphaElm[i] > 0 }$
        \label{alg:nmpisol(unv:sol)} \;
    \nl \lIf{$(\expVec[0] - \expVec[i])^{\intercal}\! \cdot \dsolVec \leq 0$} {
        \!\Return $\Ff$ }
    \nl $(\unkElm[1], \unkElm[2]) \gets \AllGuess \,\From \,\set{ (\unkElm[1],
        \unkElm[2]) \in \unkVec \times \unkVec }{ \unkElm[1] \vartriangleleft
        \unkElm[2] }$
        \label{alg:nmpisol(unv:str)} \;
    \nl \lIf{$\dsolElm[{\unkElm[1]}] \geq \dsolElm[{\unkElm[2]}] $} {
          \!\Return $\Ff$
        }
    \nl \Return $\Tt$ \;
  }
\end{algorithm}
}

}

\hypersetup {
  pdftitle  = {Attacking Diophantus: Special Cases of Bag Containment},
  pdfauthor = {G. Konstantinidis, X. Li, F. Mogavero}
}


\begin{document}

  \title{Attacking Diophantus: \\ Special Cases of Bag Containment}

  \thanks{This work is based on~\cite{KM19} and~\cite{KM25}, which appeared in
  PODS'19 and ICDT'25, respectively.}

  \author[G.~Konstantinidis]
    {George Konstantinidis\lmcsorcid{0000-0002-3962-9303}}[a]
  \author[X.~Li]{Xinzhuo Li\lmcsorcid{0009-0006-9982-0296}}[a]
  \author[F.~Mogavero]
    {Fabio Mogavero\lmcsorcid{0000-0002-5140-5783}}[b]

  \address{University of Southampton, Southampton, United Kingdom}
  \email{g.konstantinidis@soton.ac.uk, xinzhuo.li@soton.ac.uk}
  \address{Universit\`a degli Studi di Napoli Federico II, Napoli, Italy}
  \email{fabio.mogavero@unina.it}



\begin{abstract}

\emph{Query containment} is a fundamental decision problem in database theory:
given two queries, determine whether, over all database instances, every answer
produced by the first is also produced by the second.
For \emph{conjunctive queries} under \emph{set semantics}, the problem has been
well understood since the 1970s through the classical homomorphism-based
characterisation.
Under \emph{bag semantics}, the interpretation of relational database underlying
real systems, containment instead becomes a quantitative comparison of answer
multiplicities.
Despite decades of work, the decidability of \emph{bag containment} for
conjunctive queries remains open.
This open frontier is also quite fragile: for classes only slightly more
expressive than conjunctive queries, bag containment is undecidable, with the
corresponding negative results ultimately relying on reductions from variants of
\emph{Hilbert's 10th problem}.

This work develops a unified framework for bag containment of conjunctive
queries that subsumes two previously studied decidable cases:
\emph{projection-free} and \emph{join-on-free} containee queries.
The framework yields decidability for a substantially broader class of queries,
called \emph{join-uniform queries}, while leaving the containing query
arbitrary.
This contrasts with techniques that obtain decidability by imposing restrictions
on the containing query.

The proposed approach identifies the tractable classes, based on the internal
\emph{unification structure} of the query whose multiplicities have to be
bounded.
Specifically, it reduces containment to a controlled Diophantine problem.
Starting from the containee query, one builds a canonical model generated by all
its possible \emph{unifications}, over which the relevant multiplicities admit a
finite arithmetic characterisation.
Containment is then proved equivalent to the non-existence of solutions of a
corresponding \emph{Diophantine inequality system}.
Although these problems are undecidable in general, we show that the systems
arising from join-uniform containment form a decidable subclass.
Thus, the standard source of undecidability for bag containment is turned into
the core of the decision procedure.

\end{abstract}


  \maketitle



\section{Introduction}
\label{sec:int}

The \emph{query containment problem} under \emph{bag-semantics} has persisted as
a long-standing open problem in database theory for more than thirty
years~\cite{CV93}.
It asks whether the answer of one query (\aka \emph{containee}) is always
contained in the answer of another query (\aka \emph{containing}), across all
possible database instances.
The problem has significant applications in several areas of computer science,
such as query optimisation~\cite{CM77,KA13,KNR16}, data
integration~\cite{LMSS95,KA11}, knowledge representation and
reasoning~\cite{BL04}, and data privacy~\cite{MG06,GH18,KHC21}.
The simplest and most-studied form of the problem concerns \emph{conjunctive
queries} ({\CQ}s), which constitute the core of every structured query language,
under \emph{set-semantics}, where duplicated tuples within database relations
and query answers are not considered.
Under bag-semantics, instead, relations and answers are \emph{bags}, \ie,
multisets, allowing for multiple copies of the same tuple, a possibility that
had already been incidentally considered in~\cite{IW91}.
This semantics constitutes the default behaviour on most relational database
management systems, motivating the identification and study of the associated
decision problems~\cite{CV93,IR95}.
Two variants will be relevant here.
Under \emph{bag-set semantics}, database relations are sets while query answers
retain their multiplicities; under \emph{bag-bag semantics}, multiplicities are
allowed already in the database relations.
We refer to the corresponding containment problems as \emph{bag-set} and
\emph{bag-bag containment}, respectively.
Despite numerous attempts and significant strides~\cite{Kol13}, the general
containment problem for {\CQ}s under bag-semantics remains open to date.
Recent work has also reinforced the relevance of the notion of multiplicity
beyond the classical containment problem.
One such direction concerns probabilistic query evaluation under bag-semantics,
where Boolean queries give rise to probability distributions over answer
multiplicities, leading to expectation and bounded-count probability
variants~\cite{GLS23}.
In parallel, a growing line of work on Datalog over semirings and provenance
studies recursive query evaluation over algebraic domains in which
multiplicities, weights, or annotations are first-class
objects~\cite{KNPSW22,ZDKRT24,FKR25}.
Since ordinary bag-semantics corresponds to the semiring of natural numbers,
this line provides a broader algorithmic context for quantitative query
semantics.
Finally, recent SQL-equivalence tools explicitly reason about bag-semantics,
symbolic or bounded equivalence, integrity constraints, and features of
practical SQL~\cite{ZANHW22,WPC24,HZWW24}.
All these developments make clear that reasoning about multiplicities is both
theoretically and operationally relevant.

The containment problem for {\CQ}s under set-semantics was shown to be
\NPTime-complete by Chandra and Merlin~\cite{CM77}, who proved its equivalence
to the \emph{homomorphism detection problem} between the two queries.
This, in turn, corresponds to the evaluation of the containing query on the
\emph{canonical instance} of the containee one.
In a series of articles~\cite{GM96,Coh06,Coh09,ADG10,Chi12,Chi14,Chi16},
starting with~\cite{CV93,IR95}, partial positive decidability and complexity
results were presented for the bag-semantics version of the problem.
These works, most of them relying on homomorphism manipulations, either offered
robust yet incomplete criteria for containment/non-containment, or focused on
providing complete decision procedures for significant subclasses of {\CQ}s.
By contrast, interesting extensions of the bag-containment problem were proven
to be undecidable.
Ioannidis and Ramakrishnan~\cite{IR95} first established the undecidability for
unions of {\CQ}s, by reducing from the validity problem of \emph{Diophantine
inequalities}, a question tightly related to \emph{Hilbert's 10th
Problem}~\cite{Dav73,Rob73,Mat93}.
By reducing from another variant of the same problem, Jayram~\etal~\cite{JKV06}
demonstrated undecidability for {\CQ}s with inequalities.
Marcinkowski and Orda~\cite{MO24} further refined the negative results by
showing undecidability of several generalisations of the problem.
Recently, Marcinkowski and Ostropolski-Nalewaja investigated the mixed \CQ/\UCQ
frontier under bag-semantics~\cite{MO25}, by introducing a technique that
produces a conjunctive query approximating the behaviour of a given union of
conjunctive queries.
They then use this technique to analyse containment questions in which one query
is a \CQ and the other one is a \UCQ.
Their results further emphasise how close the known undecidability mechanisms
for unions of conjunctive queries come to the still-open conjunctive-query case.

In 2011, Kopparty and Rossman~\cite{KR11} initiated a new line of research that
addresses the equivalent \emph{homomorphism domination problem} by means of more
sophisticated techniques grounded in \emph{information theory}.
They proved decidability for particular classes of {\CQ}s enjoying specific
graph-theoretic structures, known as \emph{chordal} and \emph{series-parallel}.
More recently, Khamis~\etal~\cite{KKNS20,KKNS21} expanded upon this idea,
proving that the bag-containment problem of arbitrary containee queries into
acyclic containing ones is equivalent to an unresolved question regarding the
validity of certain \emph{information inequalities}.
They also showed decidability for chordal containing queries having a
\emph{simple junction tree}.
This line of work is orthogonal to the Diophantine one followed here: it obtains
decidability by exploiting structural restrictions on the containing query,
whereas our approach restricts the containee and leaves the containing query
arbitrary.

In~\cite{KM19}, a natural fragment of {\CQ}s was studied, namely
\emph{projection-free queries} ({\PFQ}s), \ie, queries without existentially
quantified variables.
It was proved that the bag-containment of a \PFQ into an arbitrary \CQ can be
decided by a $\QAE$-alternating polynomial-time algorithm, placing the problem
in the second level of the polynomial hierarchy, more precisely in $\ULH[2][P]$,
while the problem is also \NPTimeH.
That work first reduced the containment problem to solving a special case of
Diophantine inequalities, known as \emph{monomial-polynomial inequalities}
(\MPI), and then presented a procedure for the latter, based on solving a linear
system.
Notably, this approach reversed the usual role of Diophantine inequalities in
bag-semantics: instead of serving as a source of undecidability, they became a
tool for proving decidability.
A subsequent step was taken in~\cite{KM25}, where \emph{join-on-free queries}
({\JoFQ}s), which allow existential variables while forbidding joins on them,
were considered for the containee query.
The proof retained the same general philosophy as~\cite{KM19}, but required a
more delicate homomorphism-counting analysis.
Canonical witnesses were replaced by \emph{multicanonical instances}, the notion
of atom images by the refined notion of \emph{net-images}, and the containment
problem was reduced to a constrained monomial-polynomial inequality.
The same kind of algorithmic approach, applied over an exponentially larger
search space, placed bag-set and bag-bag containment of {\JoFQ}s into arbitrary
{\CQ}s in \CoNExpTime.

The present article shows that these two cases are not isolated.
What makes the method work is not the specific projection-freeness, nor the
absence of existential joins, but the finiteness and order-theoretic regularity
of the unification closure generated by the containee query.
To this aim, we introduce and study the class of \emph{join-uniform queries}
({\JUQ}s), in which every atom of an existentially connected component contains
all existential variables of that component.
This class subsumes both {\PFQ}s and {\JoFQ}s, while allowing genuine joins on
existential variables.
We focus on bag-set containment of a boolean \JUQ into an arbitrary boolean \CQ.
Unlike for {\JoFQ}s, the standard copy-attribute reduction from bag-bag to
bag-set containment does not preserve join-uniformity.
Therefore, our result does not settle bag-bag containment for {\JUQ}s, whereas
the corresponding bag-bag results for {\PFQ}s and {\JoFQ}s follow by the
standard reduction.

The first part of our approach turns the database problem into a finite
arithmetic one.
We develop this reduction in four steps.
We first identify a finite family of structured instances that suffices for
testing non-containment, then derive an exact way of counting homomorphisms on
them, organise the resulting numerical parameters through the order induced by
the containee query, and finally obtain the desired monomial-polynomial
characterisation.

We start in~\cref{sec:commul;sub:mulins} from the main difficulty created by
joins and self-joins: different components of the containee may have overlapping
images in a database.
To describe all relevant overlaps finitely, we close its components under the
unifications that may arise within or across them.
For {\JUQ}s, the resulting \emph{minimal unification closure} is finite and
inherits a natural homomorphic order.
Its elements represent the finitely many query shapes that can occur when
components overlap.
Using these shapes, we build a controlled family of \emph{multicanonical
instances} and prove that it suffices to search for counterexamples to
containment within this family.

In~\cref{sec:commul;sub:cnthom}, we determine how to count the homomorphisms of
an arbitrary containing query into such instances.
The key point is that we do not need to remember the particular copy of a
closure shape in which every part of the query lands.
It is enough to record just the corresponding closure shapes and then count how
many copies of each shape are available.
This separates the combinatorial ways in which the containing query can be
mapped from the actual multiplicities occurring in the instance.
Summing over the finitely many possible mapping patterns yields an exact
expression for its multiplicity in terms of finitely many counting parameters.

In~\cref{sec:commul;sub:poschr}, we study the dependencies among these
parameters.
A copy associated with one closure element may also contribute homomorphisms
from other elements above it and the number of such contributions depends only
on the homomorphisms between the corresponding closure shapes.
The homomorphic order therefore comes equipped with numerical weights that
describe how local contributions accumulate along the closure.
This gives an invertible relation between the numbers of genuinely new copies
introduced at each level and the total homomorphism counts observed there.
The latter are particularly convenient because the multiplicity of the containee
is simply a monomial in these total counts.
At this stage, the relevant database multiplicities are reduced to a finite
collection of numerical counts satisfying the appropriate arithmetic relations.

Finally, in~\cref{sec:commul;sub:polchr}, we make this arithmetic description
explicit.
We group together the finitely many mapping patterns of the containing query
that contribute in the same way and express their total contribution as a
polynomial.
Using the relation between the two forms of counting, this polynomial can be
written in the same variables used by the monomial for the containee.
The values assigned to these variables cannot be arbitrary: after recovering
numbers of new copies at the different closure levels, those numbers must still
be non-negative integers.
We call assignments satisfying this consistency requirement \emph{natural}.
If, in addition, the counting parameters themselves are natural numbers, we call
such assignments \emph{Diophantine natural}.
The closure order, the weights controlling the interaction between its levels,
and the fixed contributions determined by the ground closure elements are
collected into a \emph{multiplicity poset}, whose arithmetic constraints express
precisely these requirements.
We thus obtain a finite monomial-polynomial inequality whose Diophantine natural
solutions correspond exactly to witnesses of non-containment.
The inverse passage from total homomorphism counts to the underlying numbers of
new copies may introduce negative coefficients into the polynomial, even though
its value still represents an actual homomorphism count on genuine counting
assignments.
The resulting polynomials nevertheless retain a crucial property: on Diophantine
natural assignments, the whole polynomial dominates every monomial having a
positive coefficient, once that coefficient is ignored.
We call this property \emph{strong non-negativeness}.
It is this additional structure, together with the arithmetic constraints on the
counting parameters, that the second part of our approach exploits.

The second part of our approach leaves the database interpretation behind and
asks when a generic constrained monomial-polynomial inequality has a Diophantine
natural solution (see~\cref{prb:nmpi}).
The central idea is to compare the rates at which its monomials grow.
We develop this idea in five steps.
The first three isolate the pure growth problem, the fourth restores the
additional arithmetic constraints, and the last turns the resulting
characterisation into a decision procedure.

We start in~\cref{sec:dphprb;sub:sol1gmpi} with the simplest case of an
inequality over a single unknown.
Here, the asymptotic comparison between the two sides is essentially determined
by their degrees: a monomial of larger degree eventually dominates one of
smaller degree.
This simple observation provides the basic growth principle used throughout the
analysis.

In~\cref{sec:dphprb;sub:nexppar1mpi}, we extend this idea to several unknowns.
We represent their values as powers of a common base and use the corresponding
exponents as parameters.
The growth rate of each monomial then becomes a linear expression in these
parameters.
Comparing the growth rates of the monomials therefore gives a homogeneous system
of linear inequalities, enriched with constraints expressing the order inherited
from the closure.

In~\cref{sec:dphprb;sub:incsolnmpi}, we show that this linearisation captures
exactly the relevant positive solutions whose coordinates increase strictly
along this order.
Solutions of the original inequality of this form yield feasible exponent
parameters, while solutions of the linear system can be turned into arbitrarily
large Diophantine solutions satisfying the same strict order.
Thus, once only growth is taken into account, the non-linear problem reduces to
linear feasibility.

In~\cref{sec:dphprb;sub:gensolnmpi}, we restore the arithmetic conditions that
were temporarily set aside.
In general, the values of the unknowns must correspond to valid counting
parameters, the polynomial may contain negative coefficients, and some
coordinates may be zero.
Strong non-negativeness allows us to reconcile the negative coefficients with
the previous linear characterisation, while preserving the additional counting
constraints.
Solutions with zero coordinates are handled by restricting the problem to the
part of the order on which the solution is positive.
This yields a linear characterisation of Diophantine natural solvability over an
appropriate restricted set of unknowns.
Unlike in the join-on-free setting, the argument works directly over an
arbitrary finite multiplicity poset and does not require a meet-semilattice.

Finally, in~\cref{sec:dphprb;sub:decprcnmpi}, we turn this characterisation into
an effective procedure.
Polynomial bounds on integer solutions of linear systems allow us to search for
a bounded witness and verify the corresponding constraints within
\QEA-alternating polynomial time.
This completes the solution of the abstract Diophantine problem independently of
the database construction from which it originated.

Combining the two parts yields a 2\CoNExpTime upper bound for bag-set
containment of {\JUQ}s into arbitrary {\CQ}s.
The same framework recovers the sharper \CoNExpTime upper bounds for bag-set and
bag-bag containment of {\JoFQ}s into arbitrary {\CQ}s, as well as the \ULH[2][P]
upper bound for bag-bag containment of {\PFQ}s.
All these problems are also \NPTimeH.

The article is organised as follows.
\Cref{sec:prl} introduces the required preliminaries, while
\cref{sec:spccasbagcon} presents the containment problems and the three classes
of containee queries considered here.
\Cref{sec:commul} develops the unification closure, multicanonical instances,
and the homomorphism-counting characterisation.
\Cref{sec:dphprb} isolates and solves the resulting Diophantine problem.
Finally, \cref{sec:deccom} turns this solution into the containment algorithms
and their complexity bounds.
We conclude with a discussion of the consequences and open directions.




\section{Preliminaries}
\label{sec:prl}



\subsection{Queries \& Homomorphisms}
\label{sec:prl;sub:qryhom}

We consider infinite sets of \emph{constants} $\ConSet$, \emph{values} $\ValSet$ and \emph{variables} $\VarSet$ with $\VarSet \subset \ValSet$. The set of terms is defined as $\TerSet = \ConSet \cup \ValSet$.
Intuitively, values which are not variables (i.e., elements of $\ValSet \setminus \VarSet$)  will be used in database instances in a way similar to labelled nulls~\cite{fagin2005data}: at times acting as constants, in the sense that they form tuples, but at times treated as variables in the sense that we do not care about their particular value and we will use homomorphisms that can map a value in $\ValSet$ to another term.

We use the usual notion of \emph{relation names}, and atoms of the form $\RRel(\vec*{\tElm})$ where $\RRel$ is a relation name and $\vec*{\tElm}$ is a tuple of terms.
%
\emph{Ground atoms} or
\emph{facts} atoms which do not contain variables. 
Relations are sets of facts, and so is a \emph{database} or an \emph{instance}.
For every syntactic expression $\sElm$, we write $\VarSet(\sElm)$, $\ValSet(\sElm)$,
$\ConSet(\sElm)$, and $\TerSet(\sElm)$ for the sets of variables, values, constants, and
terms occurring in $\sElm$ respectively.

Given two sets of atoms $\SSet[1]$ and $\SSet[2]$, a \emph{homomorphism} from
$\SSet[1]$ to $\SSet[2]$ is mapping $\hFun \colon \TerSet(\SSet[1]) \to
\TerSet(\SSet[2])$ such that: (i) for all $\conElm \in \ConSet \cap
\TerSet(\SSet[1])$ it holds that $\hFun(\conElm) = \conElm$, and (ii)
$\hFun(\SSet[1]) \subseteq \SSet[2]$.
A homomorphism from $\SSet[1]$ to $\SSet[2]$ is denoted $\SSet[2]
\preccurlyeq^{\hFun} \SSet[1]$, or simply $\SSet[2] \preccurlyeq \SSet[1]$ when
we do not care about $\hFun$.
When $\SSet[2] \preccurlyeq \SSet[1]$ but $\SSet[1] \not\preccurlyeq \SSet[2]$,
we write $\SSet[2] \prec \SSet[1]$.
%
%
%
%
%
We denote by $\HomSet(\SSet[1], \SSet[2])$ the set of all such homomorphisms.
For an atom set $\SSet$, an \emph{automorphism} of $\SSet$ is a homomorphism $\homFun \in
  \HomSet(\SSet, \SSet)$ such that
  $\homFun(\SSet) = \SSet$. Trivially, the  identity mapping $id_S$ is an automorphism. An automorphism is non-trivial if it is not the identity.

  An \emph{isomorphism} from $\SSet[1]$ to $\SSet[2]$ is a homomorphism $\hFun$ from $\SSet[1]$ to $\SSet[2]$ (i.e., $\SSet[2] \preccurlyeq^{\hFun} \SSet[1]$) such that $\hFun$ has an inverse $\hFun^{-1}$, which is a homomorphism from $\SSet[2]$ to $\SSet[1]$, that is, $\SSet[1] \preccurlyeq^{\hFun^{-1}} \SSet[2]$ (note that automorphisms are isomorphisms). Note that, it must be that $\hFun(\SSet[1])=\SSet[2]$ and $\hFun^{-1}(\SSet[2])=\SSet[1]$. Also, it must hold that $\ConSet(\SSet[1]) = \ConSet(\SSet[2]) = \CSet$. We will denote the existence of an isomorphism by writing $\SSet[1] \approx \SSet[2]$.

A \emph{conjunctive query} (\CQ) $\qryElm(\vec*{\varElm})$ is a first-order
formula of the form
\[
  \exists \vec*{\yElm_1}, ...,\vec*{\yElm_n}
  \bigwedge_{i=1}^{n} \RRel[i](\vec*{\cElm_i}, \vec*{\varElm_i}, \vec*{\yElm_i})\textit{, where }\vec*{\varElm_i}, \vec*{\yElm_i}\textit{, are variables, }\vec*{\cElm_i}\textit{ constants, and }  \vec*{\varElm} = \bigcup^{i=n}_{i=1} \vec*{\varElm_i},
\]
We also use the datalog notation
\[
  \qryElm(\vec*{\varElm}) \leftarrow
  \RRel[1](\vec*{\cElm_1}, \vec*{\varElm_1}, \vec*{\yElm_1}), \ldots,
  \RRel[n](\vec*{\cElm_n}, \vec*{\varElm_n}, \vec*{\yElm_n}).
\]
The set of atoms in $\qryElm$, called the \emph{body}, is denoted by $\bodyFun(\qryElm)$. The free variables $\vec*{\varElm}$ are also called \emph{distinguished}; the rest are \emph{existential}. When it is clear from context we will refer to the body of a query set as simply the ``query''. For example, we will talk about homomorphisms from a query to another set of atoms, or another query (meaning homomorphisms from the body set o atoms).
When the distinguished variables $\vec*{\varElm} = \emptyset$, the query is Boolean (\BCQ).
We can turn any non-Boolean query $
  \qryElm(\vec*{\varElm}) \leftarrow
  \RRel[1](\vec*{\cElm_1}, \vec*{\varElm_1}, \vec*{\yElm_1}), \ldots,
  \RRel[n](\vec*{\cElm_n}, \vec*{\varElm_n}, \vec*{\yElm_n})
$ to a BCQ by \emph{freezing} it, that is, by replacing the distinguished variables $\vec*{\varElm}$ with a tuple $\vec*{\hat{\varElm}}$ of ``fresh'' constants (not existing in the query):
 $ \qryElm() \leftarrow
  \RRel[1](\vec*{\cElm_1}, \vec*{\hat{\varElm_1}}, \vec*{\yElm_1}), \ldots,
  \RRel[n](\vec*{\cElm_n}, \vec*{\hat{\varElm_n}}, \vec*{\yElm_n}).$

%
%

A \emph{set database instance} $\dbElm$ is a set of facts.
The \emph{answer under set semantics} of a \CQ $\qryElm(\vec*{\varElm})$ over
$\dbElm$, denoted by $\qryElm(\dbElm)$, is the set of tuples
$\vec*{\conElm}$ such that there exists a homomorphism
$\hFun \in \HomSet(\bodyFun(\qryElm), \dbElm)$ with
$\hFun(\vec*{\varElm}) = \vec*{\conElm}$.

\begin{definition}[Components]
  \label{def:components}
  Let $\SSet$ be a finite set of atoms or facts. 
  The \emph{component graph} of $\SSet$ is the undirected graph whose
  vertices are the elements of $\SSet$ and in which two distinct vertices
  $\atmElm[1],\atmElm[2] \in \SSet$ are adjacent if they share a non-constant value, that is:
  \[
    \bigl(\TerSet(\atmElm[1]) \cap \TerSet(\atmElm[2])\bigr)
    \setminus \ConSet
    \neq
    \emptyset .
  \]
  The \emph{components} of $\SSet$ are the maximal subsets
  of $\SSet$ induced by the connected components of this graph.
  \end{definition}

  \begin{example}
\label{exm:0}
The components of the set $\{R_1(a,y,c), R_2(a,y,d), R_3(c), R_4(a,d)\}$ where $y$ is a variable and $a,c,d$ are constants, are three: $\{R_1(a,y,c), R_2(a,y,d)\}$ , $\{R_3(c)\}$ and $\{R_4(a,d)\}$.
\end{example}

The components of a query are the components of its body set of atoms with the caveat that we treat distinguished variables as constants, that is, we also disconnect components that share distinguished variables.

\begin{definition}[Components of a CQ]
  \label{def:componentsQuery}
  Let a CQ $\qryElm(\vec*{\varElm})$ with $\SSet$ its body set of atoms and $\vec*{\varElm}$ its tuple of distinguished variables. 
  The \emph{component graph} of $\qryElm$ is the undirected graph whose
  vertices are the elements of $\SSet$ and in which two distinct vertices
  $\atmElm[1],\atmElm[2] \in \SSet$ are adjacent if they share an existential variable, that is:
  \[
    \bigl(\TerSet(\atmElm[1]) \cap \TerSet(\atmElm[2])\bigr)
    \setminus (\ConSet \cup \vec*{\varElm})
    \neq
    \emptyset .
  \]
  The \emph{components} of $q$ are the maximal subsets
  of $\SSet$ induced by the connected components of this graph.
  \end{definition}

  For a \BCQ $\qryElm$, we write $\comp{\qryElm}$ for the components of
  $\bodyFun(\qryElm)$, 
  identifying each component
  with the subquery whose body is that component. 
  Given an instance $\dbElm$, we write $\comp{\dbElm}$ for the components of $\dbElm$; facts in component in $\dbElm$ are connected only through
  values in $\ValSet$ (in fact, in $\ValSet \setminus \VarSet$ since instances do not contain variables).

\begin{definition}[Unification]
\label{def:unif}
Given atoms $\atmElm[1]$, $\atmElm[2]$, 
a mapping
$\uFun \colon \TerSet(\atmElm[1]) \cup \TerSet(\atmElm[2]) \to
\TerSet(\atmElm[1]) \cup \TerSet(\atmElm[2])$ is a \emph{unifier} for
$\atmElm[1]$ and $\atmElm[2]$ if (i) for all $\conElm \in \ConSet$, $\uFun(\conElm) = \conElm$, and (ii)
$\uFun(\atmElm[1]) = \uFun(\atmElm[2])$.
The \emph{most general unifier} for $\atmElm[1]$ and $\atmElm[2]$ is denoted by
$\mgu{\atmElm[1], \atmElm[2]}$.
Whenever $\atmElm[1],\atmElm[2] \in \rqryElm$, we use the same notation for the
extension of this mapping to $\TerSet(\rqryElm)$ that fixes every term in
$\TerSet(\rqryElm) \setminus (\TerSet(\atmElm[1]) \cup \TerSet(\atmElm[2]))$.
We lift unification to sets of atoms and write $\mgu{\SSet}$ for the atom
obtained by unifying the atoms in $\SSet$, whenever the unification exists.

For atoms $\atmElm[i]\in\rqryElm[i]$ ($i=1,2$), take copies of the whole
queries with disjoint variable sets, keeping constants unchanged, and use
the same notation for the copies and their selected atoms. This also
applies when the two original queries are the same. We call the selected atoms
\emph{disjointly unifiable} if they are unifiable in these copies. In
this case, let $\uFun$ be a most general unifier of the selected atoms.
For $i=1,2$, let $\subFun[i]$ agree with $\uFun$ on the terms of
$\atmElm[i]$ and leave the other terms of $\rqryElm[i]$ unchanged.
We call $(\subFun[1],\subFun[2])$ a \emph{disjoint most general unifier}
and write $(\subFun[1],\subFun[2])\defeq\dmgu{\atmElm[1],\atmElm[2]}$.
In particular, $\subFun[1](\atmElm[1])=\subFun[2](\atmElm[2])$.
The resulting union $\subFun[1](\rqryElm[1])\cup\subFun[2](\rqryElm[2])$
is determined up to isomorphism.

We also use the standard fact that a
finite simultaneous unification of atoms can be realised, up to isomorphism, by
a sequence of pairwise most-general unifications.
\end{definition}

\subsection{Posets \& Convolutions}
\label{sec:prl;sub:poscon}


A \emph{partially ordered set} (\emph{\poset}, for short) is a structure $\PStr
\defeq \tuple {\PSet} {\preccurlyeq}$ consisting of a set $\PSet$ endowed with a
\emph{partial order} $\preccurlyeq \;\subseteq \PSet \times \PSet$, \ie, a
binary relation that is reflexive, antisymmetric, and transitive.
The symbol $\prec$ denotes the \emph{strict order} induced by $\preccurlyeq$,
while $\vartriangleleft$ is the corresponding \emph{successor relation}, \ie,
$\sElm \vartriangleleft \pElm$ \iff $\sElm \prec \pElm$ and there is no $\rElm
\in \PSet$ such that $\sElm \prec \rElm \prec \pElm$, for all $\pElm, \sElm \in
\PSet$.
The poset $\PStr$ is said to be \emph{finite} if its underlying set $\PSet$ is
finite.
It is \emph{locally finite} if every \emph{interval} $\numcc{\sElm}{\pElm}
\defeq \set{ \rElm \in \PSet }{ \sElm \preccurlyeq \rElm \preccurlyeq \pElm }$
is finite, where $\pElm, \sElm \in \PSet$.
Clearly, every finite poset is also locally finite.

A function $\nFun \colon \PSet \to \SetR$ defined on the elements of $\PStr$ is
\emph{increasing} (\resp, \emph{non-decreasing}) if $\nFun(\rElm) <
\nFun(\pElm)$ (\resp, $\nFun(\rElm) \leq \nFun(\pElm)$), for all $\pElm, \rElm
\in \PSet$ with $\rElm \vartriangleleft \pElm$.

Let $\fFun, \mFun \colon {\preccurlyeq} \to \SetR$ be two functions defined on
the set of all comparable pairs of elements in $\PStr$.
If $\PStr$ is locally finite, the \emph{convolution product} $(\mFun * \fFun)
\colon {\preccurlyeq} \to \SetR$ of these two functions is definable, for every
$\pElm, \sElm \in \PSet$ with $\sElm \preccurlyeq \pElm$, by
\[
  (\mFun * \fFun)(\sElm, \pElm)
\defeq
  \sum_{\sElm \preccurlyeq \rElm \preccurlyeq \pElm} \mFun(\sElm, \rElm)
  \fFun(\rElm, \pElm).
\]
The \emph{Kronecker delta}, namely, the function $\deltaFun \colon
{\preccurlyeq} \to \SetR$, where $\deltaFun(\pElm, \pElm) \defeq 1$ and
$\deltaFun(\sElm, \pElm) \defeq 0$, for all $\pElm, \sElm \in \PSet$ with $\sElm
\prec \pElm$, is the identity element \wrt this operation, \ie, $(\deltaFun *
\fFun) = (\fFun * \deltaFun) = \fFun$.
In addition, every function $\mFun$ such that $\mFun(\pElm, \pElm) \neq 0$, for
all $\pElm \in \PSet$, admits a unique \emph{convolution inverse}, \ie, a
function $\mFun^{-1} \colon {\preccurlyeq} \to \SetR$ satisfying $(\mFun *
\mFun^{-1}) = (\mFun^{-1} * \mFun) = \deltaFun$.
This inverse can be computed recursively, for all $\pElm, \sElm \in \PSet$ with
$\sElm \prec \pElm$, by setting
\[
  \mFun^{-1}(\pElm, \pElm)
\defeq
  \frac{1}{\mFun(\pElm, \pElm)}
\quad\text{ and }\quad
  \mFun^{-1}(\sElm, \pElm)
\defeq
  - \frac{1}{\mFun(\pElm, \pElm)} \sum_{\sElm \preccurlyeq \rElm \prec \pElm}
  \mFun^{-1}(\sElm, \rElm) \, \mFun(\rElm, \pElm).
\]
Clearly, if $\eFun \defeq \mFun * \fFun$ and $\mFun$ admits a convolution
inverse, then $\fFun = \mFun^{-1} * \eFun$, since the convolution product is
associative.
Moreover, if $\mFun(\sElm, \pElm), \fFun(\sElm, \pElm) \in \SetN$, for all
$\pElm, \sElm \in \PSet$ with $\sElm \preccurlyeq \pElm$, then also
$\eFun(\sElm, \pElm) \in \SetN$, for all $\pElm, \sElm \in \PSet$ with $\sElm
\preccurlyeq \pElm$.
Observe that, if $\mFun(\sElm, \pElm) \equiv 0 \pmod{\mFun(\sElm, \sElm)}$, for
all $\pElm, \sElm \in \PSet$ with $\sElm \preccurlyeq \pElm$, it follows that
$\mFun(\pElm, \pElm) \mFun^{-1}(\sElm, \pElm) \in \SetZ$, for all $\pElm, \sElm
\in \PSet$ with $\sElm \preccurlyeq \pElm$.

Assume $\PStr$ to be a finite poset.
Given $\nFun \colon \PSet \to \SetR$ and $\mFun \colon {\preccurlyeq} \to
\SetR$, consider the \emph{convolution action} $(\mFun \star \nFun) \colon \PSet
\to \SetR$ of $\mFun$ on $\nFun$ defined, for every $\pElm \in \PSet$, by
\[
  (\mFun \star \nFun)(\pElm)
\defeq
  \sum_{\sElm \preccurlyeq \pElm} \mFun(\sElm, \pElm) \, \nFun(\sElm).
\]
Clearly, $\nFun = \deltaFun \star \nFun$.
Moreover, if $\gFun \defeq \mFun \star \nFun$ and $\mFun$ admits a convolution
inverse, then $\nFun$ can be recovered from $\gFun$ by $\nFun = \mFun^{-1} \star
\gFun$.
Indeed, for every $\pElm \in \PSet$, it holds that
\[
  (\mFun^{-1} \star \gFun)(\pElm)
=
  \sum_{\sElm \preccurlyeq \pElm} \mFun^{-1}(\sElm, \pElm) \gFun(\sElm)
=
  \sum_{\sElm \preccurlyeq \pElm} \mFun^{-1}(\sElm, \pElm) \sum_{\tElm
  \preccurlyeq \sElm} \mFun(\tElm, \sElm) \nFun(\tElm)
=
  \sum_{\tElm \preccurlyeq \sElm \preccurlyeq \pElm} \mFun^{-1}(\sElm,
  \pElm) \mFun(\tElm, \sElm) \nFun(\tElm)
\]
and, by exchanging the order of summation, one obtains
\[
  \sum_{\tElm \preccurlyeq \pElm} \! \left( \sum_{\tElm \preccurlyeq \sElm
  \preccurlyeq \pElm} \mFun(\tElm, \sElm) \mFun^{-1}(\sElm, \pElm) \right)
  \! \nFun(\tElm)
=
  \sum_{\tElm \preccurlyeq \pElm} (\mFun * \mFun^{-1})(\tElm, \pElm)
  \nFun(\tElm)
=
  \sum_{\tElm \preccurlyeq \pElm} \deltaFun(\tElm, \pElm) \, \nFun(\tElm)
=
  \nFun(\pElm).
\]
When $\mFun(\sElm, \rElm) \leq \mFun(\sElm, \pElm)$ and $\mFun(\sElm, \rElm),
\nFun(\sElm) \geq 0$ (\resp,  $\mFun(\sElm, \rElm), \nFun(\sElm) > 0$), for all
$\pElm, \rElm, \sElm \in \PSet$ with $\sElm \preccurlyeq \rElm \preccurlyeq
\pElm$, the result $\gFun$ of their convolution action is non-decreasing (\resp,
increasing).
Indeed, if $\rElm \vartriangleleft \pElm$, it holds that
\[
  \gFun(\rElm)
=
  \sum_{\sElm \preccurlyeq \rElm} \mFun(\sElm, \rElm) \nFun(\sElm)
\leq
  \sum_{\sElm \preccurlyeq \rElm} \mFun(\sElm, \pElm) \nFun(\sElm)
\leq (\resp, <)
  \sum_{\sElm \preccurlyeq \pElm} \mFun(\sElm, \pElm) \nFun(\sElm)
=
  \gFun(\pElm) \,.
\]
Finally, under the assumption that $\mFun(\sElm, \pElm) > 0$ and $\mFun(\pElm,
\pElm) \mFun^{-1}(\sElm, \pElm) \in \SetZ$, for all $\pElm, \sElm \in \PSet$
with $\sElm \preccurlyeq \pElm$, which, as we shall see, is a property enjoyed
by our context, one has $\nFun(\pElm) \in \SetN$, for all $\pElm \in \PSet$,
\iff the following conditions hold true, for all $\pElm \in \PSet$, where
$\lFun(\pElm) \defeq - \sum_{\sElm \prec \pElm} \mFun(\pElm, \pElm)
\mFun^{-1}(\sElm, \pElm) \gFun(\sElm)$:
\begin{inparaenum}[(1)]
\item
  $\gFun(\pElm) \equiv \lFun(\pElm) \pmod{\mFun(\pElm, \pElm)}$;
\item
  $\gFun(\pElm) \geq \lFun(\pElm)$.
\end{inparaenum}
Intuitively, the first condition guarantees the divisibility required for
$\nFun(\pElm)$ to be integral, whereas the second one ensures that the resulting
integer is non-negative.

For the general theory of incidence algebras, convolution inverses, and M\"obius
inversion on locally finite posets, we refer to~\cite{Rot64,Sta11}.





\section{Special Cases of Bag Containment}
\label{sec:spccasbagcon}
In this section we  introduce the bag-containment problems studied throughout the paper and the query classes that form the focus of our investigation. We first formalise the notions of bag-set and bag-bag containment, then revisit the previously known tractable classes of projection-free and join-on-free queries before introducing the more general class of join-uniform queries. The section concludes with representative examples that illustrate the progression from the earlier special cases to the broader framework developed in the remainder of the paper.



\subsection{The Bag-Set and Bag-Bag Containment Problems}
\label{sec:spccasbagcon;sub:cntprb}
A bag or \emph{multiplicity} over a set $\SSet$ is function $\bagFun \colon
\SSet \to \SetN$. Given an instance $\dbElm$, a \emph{bag} over $\dbElm$ is a function $\bagFun \colon \dbElm \to \SetN$ that, intuitively, attaches a multiplicity number or ``count'' on the database facts in $\dbElm$.
For a \CQ
\[
  \qryElm(\vec*{\varElm}) =
  \exists \vec*{\yElm} \bigwedge_{i=1}^{n}
  \RRel[i](\vec*{\conElm_i},\vec*{\varElm_i}, \vec*{\yElm_i}),
\]
the \emph{answer under bag-bag semantics of $\qryElm$ over $(\dbElm,\bagFun)$}
is the bag $\qryElm_{bb}^{\dbElm,\bagFun} \colon \qryElm(\dbElm) \to \SetN$
defined as follows.
For every $\vec*{\conElm} \in \qryElm(\dbElm)$, let
\[
  \mathcal{H}_{\vec*{\conElm}}
  \defeq
  \set{ \hFun \in \HomSet(\bodyFun(\qryElm), \dbElm) }{
    \hFun(\vec*{\varElm}) = \vec*{\conElm} }.
\]
Then
\[
  \qryElm_{bb}^{\dbElm,\bagFun}(\vec*{\conElm})
  \defeq
  \sum_{\hFun \in \mathcal{H}_{\vec*{\conElm}}}
  \prod_{i=1}^{n}
  \bagFun\!\bigl(\hFun(\RRel[i](\vec*{\conElm[i]},\vec*{\varElm[i]}, \vec*{\yElm[i]}))\bigr).
\]

The \emph{answer under bag-set semantics} of $\qryElm$ over $\dbElm$, denoted
by $\qryElm_{bs}^{\dbElm}$, is defined in the same way with
$\bagFun(\atmElm) = 1$ for every $\atmElm \in \dbElm$.

Given two {\CQ}s $\qryElm$ and $\pElm$, we say that:
\begin{itemize}[\textbullet]
\item
  $\qryElm$ is \emph{set contained} in $\pElm$, written
  $\qryElm \sqsubseteq_{s} \pElm$, if
  $\qryElm(\dbElm) \subseteq \pElm(\dbElm)$ for every set instance $\dbElm$;
\item
  $\qryElm$ is \emph{bag-bag contained} in $\pElm$, written
  $\qryElm \sqsubseteq_{bb} \pElm$, if
  $\qryElm_{bb}^{\dbElm,\bagFun}(\vec*{\conElm}) \leq
  \pElm_{bb}^{\dbElm,\bagFun}(\vec*{\conElm})$ for every set instance $\dbElm$,
  every bag $\bagFun$ over $\dbElm$, and every answer tuple $\vec*{\conElm}$;
\item
  $\qryElm$ is \emph{bag-set contained} in $\pElm$, written
  $\qryElm \sqsubseteq_{bs} \pElm$, if
  $\qryElm_{bs}^{\dbElm}(\vec*{\conElm}) \leq
  \pElm_{bs}^{\dbElm}(\vec*{\conElm})$ for every set instance $\dbElm$ and
  every answer tuple $\vec*{\conElm}$.
\end{itemize}

\noindent Recall that bag-bag containment implies bag-set containment, which implies set-containment~\cite{CV93}.

\begin{problem}[Containment Problem]
\label{prb:conprb}
  Given two classes $\QCls$ and $\PCls$ of {\CQ}s defined on the same vector of
  free variables $\vec*{\varElm}$, the containment problem under bag-set
  semantics $\CP[\bs](\QCls, \PCls)$ (\resp, bag-bag semantics $\CP[\bb](\QCls,
  \PCls)$) asks, for any pair $(\qryElm(\vec*{\varElm}),
  \pqryElm(\vec*{\varElm})) \in (\QCls, \PCls)$, to decide whether
  $\qryElm(\vec*{\varElm}) \bsinc \pqryElm(\vec*{\varElm})$ (\resp,
  $\qryElm(\vec*{\varElm}) \bbinc \pqryElm(\vec*{\varElm})$) holds true.
\end{problem}

In this paper we embark on a sequence of positive, decidability, results for certain instances the of bag containment problem. Our method is inspired
by the converse techniques of those often exploited in the proofs of undecidability which commonly reduce a diophantine problem into a containment one;  here, we typically describe an exponential reduction of our containment problem to
the solution of a Diophantine inequality system of a certain structure, and then proceed to prove decidability for the latter. In this quest, our containment problems follow a similar restriction pattern: we always consider the containing query an unrestricted general $\CQ$, while restricting the containee query. The class of containee queries that we consider in this paper captures results in \cite{KM19} and \cite{KM25} as we explain below. 




\subsection{A Tale of Three Queries}
\label{sec:spccasbagcon;sub:talthrqry}

The approach of \cite{KM19} was the first work to
use Diophantine inequalities as a positive tool, providing
a decidability result for a major class of containee queries and in doing so it showed decidability of an interesting restriction of the Diophantine problem. The class of queries studied in~\cite{KM19} is \emph{projection-free-queries}, that is, the class of queries that don't have projections (existential variables).

\begin{definition}[Projection-Free Query]
\label{def:pfq}
  A \CQ $\qryElm(\vec*{\varElm})$ with free variables $\vec*{\varElm} \subset
  \VarSet$ is \emph{projection-free} (\PFQ) if it is of the form $\bigwedge_{i =
  1}^{n} \RRel[i](\vec*{\conElm}[i], \vec*{\varElm}[i])$, where
  $\vec*{\conElm}[i] \subseteq \ConSet$ and $\vec*{\varElm}[i] \subseteq
  \vec*{\varElm}$, for all $1 \leq i \leq n$.
\end{definition}

\begin{example}
\label{exm:1}
    As an example consider the $\PFQ$
    $\qryElm[1](\vec*{\varElm[1]}, \vec*{\varElm[2]}) \leftarrow \RRel(\vec*{\varElm[1]}, \vec*{\varElm[2]}), \RRel(\vec*{\conElm[1]}, \vec*{\varElm[2]}), \RRel(\vec*{\varElm[1]}, \vec*{\conElm[2]})$.
\end{example}

Deciding bag-bag and bag-set containment of a $\PFQ$ into an arbitrary $\CQ$ is in $\ULH[2][P]$~\cite{KM19}.
Following along this direction, the approach in~\cite{KM25} considered a natural extension of $\PFQ$ by allowing existential variables in the containee query which however can not be joined terms. This resulted in the class of \emph{join-on-free} queries.

\begin{definition}[Join-on-Free Query]
\label{def:jofqry}
  A \CQ $\qryElm(\vec*{\varElm})$ with free variables $\vec*{\varElm} \subset
  \VarSet$ is \emph{join-on-free} (\JoFQ) if it is equivalently writable in the
  form $\bigwedge_{i = 1}^{n} \exists \vec*{\yvarElm}[i]
  \RRel[i](\vec*{\conElm}[i], \vec*{\varElm}[i], \vec*{\yvarElm}[i]
  )$, where
  $\vec*{\conElm}[i] \subseteq \ConSet$ and $\vec*{\varElm}[i] \subseteq
  \vec*{\varElm}$, for all $1 \leq i \leq n$.
\end{definition}

This class contains queries for which bag-containment has been discussed, but
remained unsolved, such as the the classic example from~\cite{CV93}:
\begin{example}\cite{CV93}
\label{exm:2}
%
Decide whether $\qryElm[2] \bbinc \qryElm[3]$ and $\qryElm[2] \bsinc \qryElm[3]$ for the following queries 
  $\qryElm[2](\varElm, \zvarElm)
\leftarrow
  \PRel(\varElm), \QRel(\uvarElm, \varElm), \QRel(\vvarElm, \zvarElm),
  \RRel(\zvarElm)$ and 
  $\qryElm[3](\varElm, \zvarElm)
\leftarrow
  \PRel(\varElm), \QRel(\uvarElm, \yvarElm), \QRel(\vvarElm, \yvarElm),
  \RRel(\zvarElm)$.
\end{example}

In this example, $\qryElm[2]$ is a \JoFQ since all the join terms appear in the head, while $\qryElm[3]$ is an arbitrary $\CQ$. A decision procedure for this problem is presented in~\cite{KM25} where the problem is proved to be in $\CoNExpTime$.

The idea for $\JoFQ$s came from an observation in the process of reducing the problem of bag-bag containment to bag-set containment. When faced with a bag-bag decision problem on two $\CQ$s one can reduce this to a bag-set decision problem by transforming the queries to include a new ``copy'' or ``count'' attribute in each atom~\cite{}. This attribute is an existential variable unique to each atom, and all existential variables added in this way can not be join terms (so changing a \PFQ or a \JoFQ query in this way results in a \JoFQ query). In fact, while the approaches in~\cite{KM19} and ~\cite{KM25} are similar in spirit, the former focuses on deciding bag-bag containment of $\PFQ$ (bag-bag containment always implies bag-set containment in $\CQ$s), while the latter first reduces $\JoFQ$ bag-bag containment to $\JoFQ$ bag-set containment and then focuses on the latter (since the class of $\JoFQ$ is closed under the bag-bag to bag-set reduction). The following example shows, as proven in~\cite{KM25}, that reducing $\PFQ$ bag-bag containment to bag-set containment results in obtaining a $\JoFQ$ bag-set containment problem. Solving bag-set containment for $\JoFQ$s essentially solves both problems on $\PFQ$s and $\JoFQ$s.


\begin{example}
\label{exm:3}
Consider the projection-free query
$\qryElm[4](x,y)\leftarrow R(x,y)$.  Consider the bag-bag containment problem of $\qryElm[4] \bbinc \qryElm[5]$ with
$\qryElm[5](x,y)\leftarrow R(x,y),R(x,z)$. Note that $\qryElm[4]$ is $\PFQ$ while $\qryElm[5]$ is $\JoFQ$.  Assume an instance $\dbElm$ with multiplicity function $\bagFun$, such that $m=\bagFun(R(a,b))$. The multiplicity of answer tuple  $(a,b)$ in $\qryElm[4]$ will be $m$, while for $\qryElm[5]$ it will be $m^2$.  
The classic reduction to transform bag-bag to bag-set containment transforms the query by extending the arities of each atom: $\qryElm[4]$ becomes $\qryElm^{\prime}_4(x,y) \leftarrow R(x,y,w)$ while $\qryElm[5]$ becomes $\qryElm^{\prime}_5(x,y) \leftarrow R(x,y,w),R(x,z,w^{\prime})$. It holds that $\qryElm[4](x,y) \bbinc \qryElm[5]$ iff $\qryElm^{\prime}_4 \bsinc \qryElm^{\prime}_5$.  Note that both $\qryElm^{\prime}_4$ and $\qryElm^{\prime}_5$ are $\JoFQ$.
\end{example}

In this paper, we present a unifying framework for the aforementioned results and extend the class of containee queries further: we allow existential joins under the condition that if two atoms share an existential join they need to share all existential variables between them; we call these \emph{join-uniform queries} and this class subsumes the classes of $\PFQ$ and $\JoFQ$. 

\begin{definition}[Join-Uniform Query]
\label{def:juq}
  A \CQ $\qryElm(\vec*{\varElm})$ with free variables $\vec*{\varElm} \subset
  \VarSet$ is \emph{join-uniform} (\JUQ) if it is equivalently writable in the
  form $\bigwedge_{i = 1}^{n} \exists \vec*{\yvarElm}[i] \bigwedge_{j =
  1}^{m_{i}} \RRel[i, j](\vec*{\conElm}[i, j], \vec*{\varElm}[i, j],
  \vec*{\yvarElm}[i])$, where $\vec*{\conElm}[i, j] \subseteq \ConSet$ and
  $\vec*{\varElm}[i, j] \subseteq \vec*{\varElm}$, for all $1 \leq i \leq n$ and
  $1 \leq j \leq m_{i}$.
\end{definition}

The $\JUQ$ definition essentially breaks the query into existentially connected components; each atom in a component needs to have all existential variables in that component. This class includes $\JoFQ$: $\qryElm[2]$ in Example~\ref{exm:2} has four components $\comp{\qryElm[2]}$=$\{\{\PRel(\varElm)\}$,$\{\QRel(\uvarElm, \varElm)\},\{ \QRel(\vvarElm, \zvarElm)\}$, $\{\RRel(\zvarElm)\}\}$ and, trivially,
every (atomic) component contains all existential variables of the component. 
Indeed, the components of any $\PFQ$ and $\JoFQ$ 
query are singletons, since components are defined by decomposing the query on its distinguished variables; every atom in a $\PFQ$ (which has only distinguished variables), and every atom in a $\JoFQ$ (where only distinguished variables are shared across atoms) are their own components. Hence both of these classes are trivially Join-Uniform Queries. However, $\JoFQ$s are much more powerful and can express interesting examples with differently looking components, such as the following.

\begin{example}
\label{exm:4}
Fix constants $c,d,e$ and let $\qryElm[6]$ have the four connected
components
\[
\begin{array}{rcl}
 C &:=& \{R(x,c),\ R(c,x)\},\\[1mm]
 D &:=& \{R(y,d),\ R(d,y)\},\\[1mm]
 W &:=& \{S(u,v,z),\ S(v,u,z)\},\\[1mm]
 U &:=& \{S(w,t,e),\ S(t,w,e)\}.
\end{array}
\]
Thus, consider the \JUQ query with the above components: 
$\qryElm[6]()$ $\leftarrow$ $R(x,c)$, $R(c,x)$, $R(y,d)$, $R(d,y)$, $S(u,v,z)$, $S(v,u,z)$, $S(w,t,e)$, $S(t,w,e)$. Abusing notation, we might write the query as $\qryElm[6]() \leftarrow C,D,W,U$.  The \JUQ condition is that every existential component variable
appears in every atom (in some position) of its component.

A containing query, under bag-bag and bag-set semantics, is $\qryElm[7] \leftarrow
C,D,W,U,R(s,c)$.  Intuitively, whenever $\qryElm[6]$ has a homomorphism
onto an instance, the component $C$ maps onto a fact of the form
$R(\alpha,c)$ with some multiplicity, so the atom $R(s,c)$ has at least one homomorphism.  

For a non-containment example consider the set
\[
\begin{array}{rcl}
 U_1 &:=& \{S(w_1,t_1,e),\ S(t_1,w_1,e)\},\\[1mm]
 U_2 &:=& \{S(w_2,t_2,e),\ S(t_2,w_2,e)\},
\end{array}
\]
and define
$p_1 \leftarrow U_1,U_2,R(c,d)$.  We will see an instance that provides a proof that $\qryElm[6] \not \bsinc p_1$ in Example~\ref{exm:netimgcnt}.


\end{example}

In the rest of this paper, we investigate bag containment for the class of boolean $\JUQ$s (or $\BJUQ$s). We focus on the bag-set containment problem which seems technically less complicated. Note, however, that in this case solving the bag-set containment problem does not necessarily solve the bag-bag containment: indeed the trivial reduction from bag-bag containment for $\JUQ$s, in the spirit of the one in example~\ref{exm:3}, does not result in bag-set of $\JUQ$; the extra attribute added to each atom in the reduction takes the class of queries outside of $\JUQ$. 

\begin{problem}
  \label{prb:bjuqconprb}
   Given a pair of {\BCQ}s $(\qryElm, \pqryElm)$, where $\qryElm$ is a \JUQ,
  the \emph{join-uniform
  containment problem under bag-set semantics} is the problem of
  deciding whether $\qryElm \bsinc \pElm$ holds true.
\end{problem}





\section{Computing Multiplicities}
\label{sec:commul}



\subsection{Multicanonical Instances}
\label{sec:commul;sub:mulins}

A \BJUQ is essentially a cross product of different components with the only common terms across components being constants.
Due to this, we can show that its bag-set multiplicity, \ie, the number of
homomorphisms it has to an instance, can be computed as the product of the
number of homomorphisms each different component in the query has to that instance.

A \BJUQ might still contain atoms of the same predicate, known as self-joins\footnote{Note, containment on
self-join-free queries is decidable~\cite{ADG10,IR95}}. Thus, two components in a \BJUQ might have overlapping images. Note that in this case there is at least two atoms, one in each corresponding component (or even in the same component), that are unifiable.
Still, we can compute a $\BJUQ$ multiplicity through a monomial function whose
parameters correspond to the number of tuples in the image of a component.
However, such simple closed expressions are not directly derivable for the
arbitrary {\BCQ}s $\pqryElm$ used as containing queries in \cref{prb:bjuqconprb}.
To address this, we need a notion of \emph{canonical instance}
(see~\cite{AHV95, KM19, KM25}) that precisely reflects the structure of the containee query.
In particular, in this context, the canonical instance has to reflect the closure of the containee
query under unification of its atoms, per the following definition.

\begin{definition}[Minimal Unification Closure]
\label{def:muc}
  A \emph{minimal unification closure} (\emph{\muc}) of a \BCQ $\qryElm$,
  denoted by $\denot{\qryElm} \subseteq \BCQ$, is a least set of {\BCQ}s modulo
  isomorphism for which the following three constraints are satisfied:
  \begin{enumerate}[1)]
  \item\label{def:muc(bas)}
    \textbf{(decomposition)}
    every component $\rqryElm \in \comp{\qryElm}$ of $\qryElm$ admits a \BCQ
    $\der{\rqryElm} \in \denot{\qryElm}$ such that $\der{\rqryElm} \approx
    \rqryElm$;
  \item\label{def:muc(fac)}
    \textbf{(factorisation)}
    for all {\BCQ}s $\rqryElm \in \denot{\qryElm}$ and unifiable atoms
    $\atmElm[1], \atmElm[2] \in \rqryElm$, there exists a \BCQ $\der{\rqryElm}
    \in \denot{\qryElm}$ such that $\der{\rqryElm} \approx \subFun(\rqryElm)$,
    where $\subFun \defeq \mgu{\atmElm[1], \atmElm[2]}$;
  \item\label{def:muc(mrg)}
    \textbf{(merging)}
    for all {\BCQ}s $\rqryElm[1], \rqryElm[2] \in \denot{\qryElm}$ and
    disjointly unifiable atoms $\atmElm[1] \in \rqryElm[1]$ and $\atmElm[2] \in
    \rqryElm[2]$, there exists a \BCQ $\der{\rqryElm} \in \denot{\qryElm}$ such
    that $\der{\rqryElm} \approx \subFun[1](\rqryElm[1]) \cup
    \subFun[2](\rqryElm[2])$, where $(\subFun[1], \subFun[2]) \defeq
    \dmgu{\atmElm[1], \atmElm[2]}$.
  \end{enumerate}
  $\denot{\qryElm}$ is partitionable into the sets $\denot{\qryElm}[\bot]$ and
  $\denot{\qryElm}[\top]$ of ground and non-ground {\BCQ}s, respectively.
\end{definition}

Due to the minimality (up to isomorphism) of the closure, for every $\BJUQ$ $\rqryElm$ there might be at most one query in the closure that is isomorphic to $\rqryElm$; we denote this as $\denot{\rqryElm}[\qryElm]$.
%
\begin{proposition}
\label{prp:muc}
  For all {\BCQ}s $\qryElm$ and $\rqryElm$, there is at most one \BCQ
  $\der{\rqryElm} \in \denot{\qryElm}$ such that $\der{\rqryElm} \approx
  \rqryElm$.
  Moreover, if $\rqryElm \in \comp{\qryElm}$, such a query $\der{\rqryElm}$,
  denoted by $\denot{\rqryElm}[\qryElm]$, surely exists.
  Finally, for every instance $\dbElm$, it holds that $\card{\HomSet(\qryElm,
  \dbElm)} = \prod_{\rqryElm \in \comp{\qryElm}}
  \card!{\HomSet(\denot{\rqryElm}[\qryElm], \dbElm)}$.
\end{proposition}

The intuition behind Def.~\ref{def:muc} is to capture the structure of the query images of $\qryElm$ on an arbitrary instance. Every homomorphism from a component of $\qryElm$ could be to a component in the instance that is isomorphic to the query component, or to a more ``specific'' component. 
Intuitively, unifying exhaustively the components of $\qryElm$ will give us possible structures of the images. Indeed, the first constraint of Def.~\ref{def:muc} demands that the components of $\qryElm$ are themselves (up to isomorphism) elements of $\denot{\qryElm}$ (as queries), the second constraint unifies atoms in a query in $\denot{\qryElm}$ in order to create more ``specific'' queries, and the last constraint considers exhaustive unification across components. Note that, in this process one might come up with ground queries, i.e., queries with no variables. In the case of atomic ground queries these are all single components; otherwise, a ground query with two or more atoms will be a cross-product of components. For technical reasons, \cref{def:muc} treats the two query classes (non-ground vs ground) differently placing them in $\denot{\qryElm}[\top]$ and
  $\denot{\qryElm}[\bot]$ respectively (as we will see below if $\qryElm$ is a \JUQ so are elements of $\denot{\qryElm}[\top]$).
Using \cref{def:muc} our objective is to create a canonical instance where all the images of the unified queries in \denot{\qryElm} are disjoint and one can compute the multiplicity of a query onto this instance as a polynomial over these disjoint layers of images.

\begin{example}[Minimal unification closures]
\label{exm:muc}


Consider the \JUQ query $\qryElm[6]$ of \cref{exm:4}.  The minimal unification closure $\denot{\qryElm[6]}$ has the following nine elements:

\begin{center}
\begin{tabularx}{\linewidth}{
  @{}
  >{\centering\arraybackslash}c |
  >{\centering\arraybackslash}c |
  >{\centering\arraybackslash}c |
  Y
  @{}
}
name & body & ground & origin \\ \hline

$C$   & $\{R(x,c),R(c,x)\}$       & no
      & (decomposition) initial component \\

$D$   & $\{R(y,d),R(d,y)\}$       & no
      & (decomposition) initial component \\

$C_0$ & $\{R(c,c)\}$              & yes
      & (factorisation) constraint 2 on $C$ \\

$D_0$ & $\{R(d,d)\}$              & yes
      & (factorisation) constraint 2 on $D$ \\

$E$   & $\{R(c,d),R(d,c)\}$       & yes
      & (merging) constraint 3 on $C$ and $D$ \\

$W$   & $\{S(u,v,z),S(v,u,z)\}$   & no
      & (decomposition) initial component \\

$U$   & $\{S(w,t,e),S(t,w,e)\}$   & no
      & (decomposition) initial component; also (merging)
        constraint 3 on $W$ and $U$ \\

$W_0$ & $\{S(t,t,z)\}$            & no
      & (factorisation) constraint 2 on $W$ \\

$U_0$ & $\{S(w,w,e)\}$            & no
      & (factorisation) constraint 2 on $U$
\end{tabularx}
\end{center}
\end{example}

The origin column shows how an element is obtained in $\denot{\qryElm}$, via another query $\rqryElm$ (or two queries) in $\denot{\qryElm}$ and one of the constrains in Def.~\ref{def:muc}. Note that, there could be more than one ways to obtain the same query in $\denot{\qryElm}$. Nevertheless, as we prove below, it is always true that for the query (or queries) $\rqryElm$ in $\denot{\qryElm}$, on which we apply the constraints of Def.~\ref{def:muc} obtaining a new query $\der{\rqryElm}$, it holds that $\der{\rqryElm} \preccurlyeq \rqryElm$, i.e., there is a homomorphism from $\rqryElm$ to $\der{\rqryElm}$.
For the above example the homomorphic relations are $C_0\preccurlyeq C$, $E\preccurlyeq C$, $D_0\preccurlyeq D$,
$E\preccurlyeq D$, $U_0\preccurlyeq U$, $U_0\preccurlyeq W_0$, $U\preccurlyeq W$, and $W_0\preccurlyeq W$.
%
This order is easier to read in the following figure where lower nodes are more specific, and the direction of the edges shows the direction of homomorphisms:

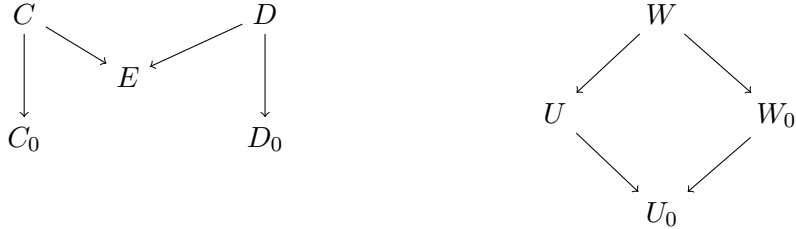
\[
\begin{tikzpicture}[baseline=(current bounding box.center),node distance=11mm]
\node (C) {$C$};
\node[right=26mm of C] (D) {$D$};
\node[below=of C] (C0) {$C_0$};
\node[below=of D] (D0) {$D_0$};
\node[below right=3mm and 8mm of C] (E) {$E$};

\draw[<-] (C0)--(C);
\draw[<-] (E)--(C);
\draw[<-] (E)--(D);
\draw[<-] (D0)--(D);

\node[right=46mm of D] (W) {$W$};
\node[below left=of W] (U) {$U$};
\node[below right=of W] (W0) {$W_0$};
\node[below right=of U] (U0) {$U_0$};

\draw[<-] (U)--(W);
\draw[<-] (W0)--(W);
\draw[<-] (U0)--(U);
\draw[<-] (U0)--(W0);
\end{tikzpicture}
\]

Because a $\BJUQ$ has a number of different components it can happen that there are two elements in $\denot{\qryElm}$, e.g., $\rqryElm[1]$, $\rqryElm[2]$, that are \emph{homomorphically equivalent} that is, it holds both that $\rqryElm[1] \preccurlyeq \rqryElm[2]$ and that $\rqryElm[2] \preccurlyeq \rqryElm[1]$. Note that, as the proof of the next theorem shows, two $\BJUQ$s 
are homomorphically equivalent whenever they are isomorphic (e.g., $\rqryElm[1] \approx \rqryElm[2]$).
In any case, as the next theorem states, we can always arrange the elements of $\denot{\qryElm}$ into a partial order using $\preccurlyeq$ (which is always either a  $\prec$ or a $\approx$). At the same time the theorem proves that for two queries in $\denot{\qryElm}$, if $\rqryElm[1] \preccurlyeq^{\hFun} \rqryElm[2]$ then homomorphism $\hFun$ must be a \emph{variable-onto} homomorphism which is a sufficient criterion for bag-set containment~\cite{CV93}.

\begin{theorem}[\BJUQ-\MUC Characterisation]
\label{thm:juqmucchr}
  For every \BJUQ $\qryElm$, the structure $\PStr = \tuple {\denot{\qryElm}}
  {\preccurlyeq}$ is a finite \poset.
  Moreover, every \BCQ in $\denot{\qryElm}[\top]$ is connected and join-uniform.
  Finally, $\rqryElm[1] \preccurlyeq \rqryElm[2]$ \iff $\rqryElm[1] \bsinc
  \rqryElm[2]$, for all {\BCQ}s $\rqryElm[1], \rqryElm[2] \in \denot{\qryElm}$.
\end{theorem}
\begin{proof}
  We first provide an invariant property of the closure that is used throughout
  the proof.
  Since $\denot{\qryElm}$ is the least set closed under the rules
  of~\cref{def:muc}, every element of $\denot{\qryElm}$ is obtained, up to
  $\approx$-equivalence, by starting from a component of $\qryElm$ and
  repeatedly applying either factorisation (Item~\ref{def:muc(fac)} of
  \cref{def:muc}, \ie, an intra-query unification, or merging
  (Item~\ref{def:muc(mrg)} of \cref{def:muc}, \ie, an inter-query disjoint
  unification.
  Along this construction, no new constants are introduced.   Moreover, every
 element $\rqryElm$ obtained in this way enjoys the
  property that all its variables $\varElm \in \var{\rqryElm}$ occur in each one
  of its atoms $\atmElm \in \rqryElm$.
  We now prove this by induction on the decreasing number of operations of
  factorisation and merging.
  \begin{enumerate}
  \item
    All non-ground initial components $\rqryElm \in \comp{\qryElm}$ enjoy the
    property due to the \BJUQ hypothesis on $\qryElm$.
  \item
    For a non-ground factorisation $\der{\rqryElm} = \subFun(\rqryElm)$ of a
    query $\rqryElm \in \denot{\qryElm}$, where $\atmElm[1], \atmElm[2] \in
    \rqryElm$ and $\subFun \defeq \mgu{\atmElm[1], \atmElm[2]}$, consider an
    arbitrary variable $\varElm \in \var{\der{\rqryElm}}$ and choose a variable
    $\yvarElm \in \var{\rqryElm}$ with $\subFun(\yvarElm) = \varElm$, whose
    existence is ensured by the fact that constants are necessarily preserved by
    the unification process.
    Clearly, by the inductive hypothesis on $\rqryElm$, the variable $\yvarElm$
    occurs in every atom of $\rqryElm$, thus, $\varElm$ occurs in every atom of
    $\der{\rqryElm}$, as well.
  \item
    For a non-ground merging $\subFun[1](\rqryElm[1]) \cup
    \subFun[2](\rqryElm[2])$ of two queries $\rqryElm[1], \rqryElm[2] \in
    \denot{\qryElm}$, where $\atmElm[1] \in \rqryElm[1]$, $\atmElm[2] \in
    \rqryElm[2]$, and $(\subFun[1], \subFun[2]) \defeq \dmgu{\atmElm[1],
    \atmElm[2]}$, consider an arbitrary variable $\varElm \in
    \var{\der{\rqryElm}}$.
    \Wlogx, suppose that $\varElm \in \var{\subFun[1](\rqryElm[1])}$.
    Then, there exists a variable $\yvarElm \in \var{\rqryElm[1]}$ with
    $\subFun[1](\yvarElm) = \varElm$, for the same reasons of the previous case.
    By the inductive hypothesis on $\rqryElm[1]$, the variable $\yvarElm$ occurs
    in every atom of $\rqryElm[1]$ and, in particular, in the selected atom
    $\atmElm[1]$.
    Thus, $\varElm$ occurs in every atom of $\subFun[1](\rqryElm[1])$ and, in
    particular, in the unified atom $\subFun[1](\atmElm[1])$, as well.
    Due to the fact that $\subFun[1](\atmElm[1]) = \subFun[2](\atmElm[2])$, the
    same variable $\varElm$ occurs in the unified copy of $\atmElm[2]$.
    Hence, there must be a variable $\zvarElm \in \var{\rqryElm[2]}$ with
    $\subFun[2](\zvarElm) = \varElm$.
    By the inductive hypothesis this time on $\rqryElm[2]$, the variable
    $\zvarElm$ occurs in every atom of $\rqryElm[2]$, thus, $\varElm$ occurs in
    every atom of $\subFun[2](\rqryElm[2])$, as well.
    Summing up, $\varElm$ occurs in every atom of $\der{\rqryElm}$.
  \end{enumerate}

  The invariant implies that each homomorphism between two elements of \denot{\qryElm} is surjective on the variables of its codomain, i.e., a variable-onto mapping~\cite{CV93}.
  The invariant also implies finiteness of \denot{\qryElm}.  No induction step introduces relation symbols or constants, and every non-ground element has all its variables in each atom.  Hence the number of variables of a non-ground element is bounded by the maximum arity of the relation symbols occurring in $\qryElm$.  Up to
  $\approx$, only finitely many non-ground atoms, and therefore only
  finitely many non-ground {\BCQ}s, can satisfy this invariant.  The ground
  elements are finite in number for the same reason: their constants belong to
  the finite set $\con{\qryElm}$, and their relation symbols already occur in
  $\qryElm$.  By \cref{def:muc}, the closure contains at most one representative
  of each $\approx$-class.  Thus $\denot{\qryElm}$ is finite.

 To prove a partial ordered set one needs to prove reflexivity, antisymmetry and transitivity of the relation $\preccurlyeq$. Reflexivity of $\preccurlyeq$ follows from the identity homomorphism, and
  transitivity follows by composition of 
  homomorphisms.  For antisymmetry, let
  $\rqryElm[1],\rqryElm[2]\in\denot{\qryElm}$ and assume
  $\rqryElm[1]\preccurlyeq^{\hFun}\rqryElm[2]$ and
  $\rqryElm[2]\preccurlyeq^{\hFun^{\prime}}\rqryElm[1]$ (we need to prove that  $\rqryElm[1]=\rqryElm[2]$).  
  If either $\rqryElm[1]$,$\rqryElm[2]$ is ground then its homomorphism to the other one is the identity, hence $\rqryElm[1] = \rqryElm[2]$.
%
%
  For non-ground queries since both $\hFun$ and $\hFun^{\prime}$ are variable-onto, it must be that $\card{\VarSet(\rqryElm[1])}\leq\card{\VarSet(\rqryElm[2])}$ and
  $\card{\VarSet(\rqryElm[2])}\leq\card{\VarSet(\rqryElm[1])}$, hence the two
  variable sets have the same cardinality.  The two homomorphisms are
  therefore bijections, and two-way bijective homomorphisms imply an isomorphism, that is  $\rqryElm[1] \approx \rqryElm[2]$.
  %
  %
%
  However, by \cref{def:muc} there cannot be two isomorphic queries in $\denot{\qryElm}$ that are different, therefore it must be that $\rqryElm[1]=\rqryElm[2]$, and this proves $\tuple{\denot{\qryElm}}{\preccurlyeq}$ is a finite \poset.

  The second assertion follows from the same invariant.  If
  $\rqryElm\in\denot{\qryElm}$ is non-ground, then all its atoms share every
  variable of $\rqryElm$, and thus they belong to one component.  Moreover,
  this normal form is precisely the Boolean join-uniform form for a connected
  Boolean query.  Hence every non-ground element of $\denot{\qryElm}$ is a
  connected \BJUQ.

 On the the equivalence with bag-set containment, let
  $\rqryElm[1]\preccurlyeq^{\hFun}\rqryElm[2]$. Since $\hFun$ is a variable-onto homomorphism then $\rqryElm[1]\bsinc \rqryElm[2]$ follows~\cite{CV93}. For the other direction if
  $\rqryElm[1]\bsinc \rqryElm[2]$, it is trivially true (due to set inclusion) that there is a homomorphism $\hFun$ from $\rqryElm[2]$ to $\rqryElm[1]$ thus $\rqryElm[1]\preccurlyeq^{\hFun}\rqryElm[2]$.


\end{proof}

Two queries in $\denot{\qryElm}$ may have overlapping homomorphic images
even when neither is below the other. To record these overlaps, we define
the image of a query as the set of its homomorphic images, each of which
is a set of facts. The classical set query image~\cite{AHV95} instead
takes the union of these sets.

\begin{definition}[Image]
\label{def:img}
  For each \BCQ $\qryElm$ and instance $\dbElm$, the \emph{image} of $\qryElm$
  on $\dbElm$ is the set of sets of atoms defined as follows: $\imgFun(\qryElm,
  \dbElm) \defeq \set{ \homFun(\qryElm) }{ \homFun \in \HomSet(\qryElm, \dbElm)
  }$.
  Each element of $\imgFun(\qryElm, \dbElm)$ is called an \emph{image element}.
\end{definition}

\begin{example}
\label{exm:img}
Consider again the query $\qryElm[6]()$ $\leftarrow$ $R(x,c)$, $R(c,x)$, $R(y,d)$, $R(d,y)$, $S(u,v,z)$, $S(v,u,z)$, $S(w,t,e)$, $S(t,w,e)$ of \cref{exm:4}, or its shorthand notation $\qryElm[6]() \leftarrow C,D,W,U$.
Let instance $\dbElm$ consisting of the union of the following sets of facts.
The non-ground copies are
\[
\begin{array}{rcll}
C_i&=&\{R(\alpha_i,c),R(c,\alpha_i)\},&i=1,2,\\
D_j&=&\{R(\beta_j,d),R(d,\beta_j)\},&j=1,2,3,\\
W_1&=&\{S(\rho,\sigma,\tau),S(\sigma,\rho,\tau)\},\\
U_i&=&\{S(\lambda_i,\mu_i,e),S(\mu_i,\lambda_i,e)\},&i=1,2,\\
W^0_1&=&\{S(\omega,\omega,\zeta)\},\\
U^0_1&=&\{S(\nu,\nu,e)\}.
\end{array}
\]
and the selected ground copy is $E^*=\{R(c,d),R(d,c)\}$. Here $\conElm,d,e$ are constants and the rest are values in $\ValSet \setminus \VarSet$.

The following table gives the image elements and the numbers of
homomorphisms.
\[
\begin{array}{c|c|c}
 r & \mathrm{img}(r,I) & |\mathrm{Hom}(r,I)|\\ \hline
C & C_1,C_2,E^*  & 3\\
D & D_1,D_2,D_3,E^*  & 4\\
E & E^*  & 1\\
C_0 & \varnothing &  0\\
D_0 & \varnothing &  0\\
U_0 & U^0_1 & 1\\
U & U_1,U_2,U^0_1 & 5\\
W_0 & W^0_1,U^0_1 & 2\\
W & W_1,U_1,U_2,W^0_1,U^0_1 & 8
\end{array}
\]
Note that $U$ has three image elements but five
homomorphisms, because each
$U_i$ admits two automorphisms, while $U^0_1$ admits one.  Similarly,
$W$ has five image elements but eight homomorphisms.
\end{example}

In the example above, $E^*$ is an image element of both $C$ and $D$,
although $C$ and $D$ are incomparable. The next lemma shows that if
image elements of two closure queries overlap, their union is an
image element of a closure query below both.

\begin{lemma}
\label{lem:img}
  Let $\qryElm$ be a \BJUQ and $\dbElm$ an instance.
  Then, for all {\BCQ}s $\rqryElm[1], \rqryElm[2] \in \denot{\qryElm}$ and image
  elements $\JSet[1] \in \imgFun(\rqryElm[1], \dbElm)$ and $\JSet[2] \in
  \imgFun(\rqryElm[2], \dbElm)$ satisfying $\JSet[1] \cap \JSet[2] \neq
  \emptyset$, there exists a \BCQ $\sqryElm \in \denot{\qryElm}$ such that
  $\sqryElm \preccurlyeq \rqryElm[1]$, $\sqryElm \preccurlyeq \rqryElm[2]$, and
  $\JSet[1] \cup \JSet[2] \in \imgFun(\sqryElm, \dbElm)$.
\end{lemma}
\begin{proof}
  For $i\in\{1,2\}$, choose
  $\homFun[i]\in\HomSet(\rqryElm[i],\dbElm)$ with
  $\homFun[i](\rqryElm[i])=\JSet[i]$. Take renamed-apart copies of the
  two queries, and use the same notation for the copies and the
  corresponding homomorphisms. Choose a fact
  $\atmElm[0]\in\JSet[1]\cap\JSet[2]$ and atoms
  $\atmElm[i]\in\rqryElm[i]$ such that
  $\homFun[i](\atmElm[i])=\atmElm[0]$ for $i\in\{1,2\}$.
  The selected atoms are disjointly unifiable because they map to the
  same fact. Let
  $(\subFun[1],\subFun[2])\defeq\dmgu{\atmElm[1],\atmElm[2]}$.
  Extend each $\subFun[i]$ to the whole query $\rqryElm[i]$ by leaving
  its other variables unchanged, and put
  $\sqryElm[0]\defeq
  \subFun[1](\rqryElm[1])\cup\subFun[2](\rqryElm[2])$.

  Since both selected atoms map to $\atmElm[0]$, the two homomorphisms
  assign the same value to any terms identified by the most general
  unifier. We can therefore define $\homFun[0]$ on $\sqryElm[0]$ by
  $\homFun[0](\subFun[i](\tElm))=\homFun[i](\tElm)$ for each term
  $\tElm$ of $\rqryElm[i]$, $i\in\{1,2\}$.
  This gives $\homFun[0]\in\HomSet(\sqryElm[0],\dbElm)$, and
  $\homFun[0](\sqryElm[0])=\JSet[1]\cup\JSet[2]$.
  By Item~\ref{def:muc(mrg)} of \cref{def:muc}, there is
  $\sqryElm\in\denot{\qryElm}$ with $\sqryElm\approx\sqryElm[0]$.
  Composing the renamings, $\subFun[i]$, and an isomorphism from
  $\sqryElm[0]$ to $\sqryElm$ shows
  $\sqryElm\preccurlyeq\rqryElm[i]$ for both $i$.
  Composing $\homFun[0]$ with the inverse isomorphism gives a homomorphism from
  $\sqryElm$ to $\dbElm$ with image $\JSet[1]\cup\JSet[2]$, so
  $\JSet[1]\cup\JSet[2]\in\imgFun(\sqryElm,\dbElm)$.
\end{proof}

We now define net images by excluding an image element of a closure
query whenever it is contained in an image element of a strictly lower
query in $\denot{\qryElm}$.

\begin{definition}[Net Image]
\label{def:netimg}
  For each \BCQ $\qryElm$, instance $\dbElm$, and \BCQ $\rqryElm \in
  \denot{\qryElm}$, the \emph{net image} of $\rqryElm$ on $\dbElm$ \wrt
  $\qryElm$ is the set of sets of atoms defined as follows:
  $\nimgFun[\qryElm](\rqryElm, \dbElm) \defeq \set!{ \JSet \in \imgFun(\rqryElm,
  \dbElm) }{ \forall \sqryElm \in \denot{\qryElm}, \allowbreak \sqryElm \prec
  \rqryElm \ldotp \forall \KSet \in \imgFun(\sqryElm, \dbElm) \ldotp \JSet
  \not\subseteq \KSet }$.
  Each element of $\nimgFun(\qryElm, \dbElm)$ is called a \emph{net-image
  element}.
\end{definition}

\begin{example}[Net images]
\label{exm:netimg}
Net images remove images that are already included in a strictly lower closure
element.
%
%
In our running example $\qryElm[6]$ and the instance presented above in~\cref{exm:img}, the element $E^*$ is an image of all $C$, $D$ and $E$, but is
net only for $E$, since $E$ does not have a strictly lesser per $\prec$ query which (has an image element that) includes $E^{*}$ .  The element $U^0_1$ is an image of $U$, $W_0$, and $W$, but
is net only for $U_0$.  The element $W^0_1$ is an image of $W$ but is excluded
from $\mathrm{nimg}(W,I)$ because it is already an image of the lower element
$W_0\prec W$.  In summary, the next table extends the previous example with net-images.
\[
\begin{array}{c|c|c|c}
 r & \mathrm{img}(r,I) & \mathrm{nimg}(r,I) & |\mathrm{Hom}(r,I)|\\ \hline
C & C_1,C_2,E^* & C_1,C_2 & 3\\
D & D_1,D_2,D_3,E^* & D_1,D_2,D_3 & 4\\
E & E^* & E^* & 1\\
C_0 & \varnothing & \varnothing & 0\\
D_0 & \varnothing & \varnothing & 0\\
U_0 & U^0_1 & U^0_1 & 1\\
U & U_1,U_2,U^0_1 & U_1,U_2 & 5\\
W_0 & W^0_1,U^0_1 & W^0_1 & 2\\
W & W_1,U_1,U_2,W^0_1,U^0_1 & W_1 & 8
\end{array}
\]
\end{example}

It should be evident that the image of a query $\rqryElm$ in $\denot{\qryElm}$ is
completely partitionable into the net-images of all queries $\der{\rqryElm}$ that
are related to $\rqryElm$ by the homomorphic ordering $\preccurlyeq$, as
prescribed by the following theorem. The intuition behind the first four parts of the theorem is that: (a) any image element of a query in $\denot{\qryElm}$ is the (subset of the) net-image element of exactly one query in $\denot{\qryElm}$, (b) the net image elements are the same size as their queries, (c) two different net image elements do not intersect (they can not share the same fact), (d) if a query $\rqryElm$ has a homomorphism into a net-image element of another query $\sqryElm$, you can obtain the first homomorphism by composing a homomorphism that maps $\rqryElm$ to $\sqryElm$ with one that maps $\sqryElm$ to its net image element. All these properties are then useful to prove the last one: the number of homomorphisms of a query in $\denot{\qryElm}$, on an arbitrary instance, can be obtained by summing, for all smaller per $\preccurlyeq$ queries $\sqryElm$, the product of (i) the number of homomorphisms from $\rqryElm$ to $\sqryElm$, times  (ii) the number of net image elements for $\sqryElm$.

\begin{lemma}
\label{lem:netimg}
  Let $\qryElm$ be a \BJUQ and $\dbElm$ an instance.
  Then, the following properties hold true:
  \begin{enumerate}[a)]
  \item\label{lem:netimg(cov)}
    for all {\BCQ}s $\rqryElm \in \denot{\qryElm}$ and image elements $\JSet \in
    \imgFun(\rqryElm, \dbElm)$, there exists a \BCQ $\sqryElm \in
    \denot{\qryElm}$ and a net-image element $\KSet \in
    \nimgFun[\qryElm](\sqryElm, \dbElm)$ such that $\sqryElm \preccurlyeq
    \rqryElm$ and $\JSet \subseteq \KSet$;
  \item\label{lem:netimg(bij)}
    for all {\BCQ}s $\rqryElm \in \denot{\qryElm}$, net-image elements $\JSet
    \in \nimgFun[\qryElm](\rqryElm, \dbElm)$, and homomorphisms $\homFun \in
    \HomSet(\rqryElm,  \dbElm)$, with $\JSet = \homFun(\rqryElm)$, the induced
    function $\iotaFun \colon \atmElm \in \rqryElm \mapsto \homFun(\atmElm) \in
    \JSet$ is a bijection;
  \item\label{lem:netimg(dis)}
    for all {\BCQ}s $\rqryElm[1], \rqryElm[2] \in \denot{\qryElm}$ and net-image
    elements $\JSet[1] \in \nimgFun[\qryElm](\rqryElm[1], \dbElm)$ and $\JSet[2]
    \in \nimgFun[\qryElm](\rqryElm[2], \dbElm)$, it holds that $\rqryElm[1] =
    \rqryElm[2]$ and $\JSet[1] = \JSet[2]$, whenever $\JSet[1] \cap \JSet[2]
    \neq \emptyset$;
  \item\label{lem:netimg(dec)}
    for all {\BCQ}s $\rqryElm, \sqryElm \in \denot{\qryElm}$ and homomorphisms
    $\homFun[\rqryElm] \in \HomSet(\rqryElm, \dbElm)$ and $\homFun[\sqryElm] \in
    \HomSet(\sqryElm, \dbElm)$, with $\homFun[\rqryElm](\rqryElm) \subseteq
    \homFun[\sqryElm](\sqryElm) \in \nimgFun[\qryElm](\sqryElm, \dbElm)$, there
    exists a unique homomorphism $\homFun \in \HomSet(\rqryElm, \sqryElm)$ such
    that $\homFun[\rqryElm] = \homFun[\sqryElm] \cmp \homFun$;
  \item\label{lem:netimg(hom)}
    $\card{\HomSet(\rqryElm, \dbElm)} = \sum_{\sqryElm \preccurlyeq \rqryElm}
    \card{\HomSet(\rqryElm, \sqryElm)} \cdot \card{\nimgFun[\qryElm](\sqryElm,
    \dbElm)}$, for all {\BCQ}s $\rqryElm \in \denot{\qryElm}$.
  \end{enumerate}
\end{lemma}
\begin{proof}
  We prove the five items in order.

  For Item~\ref{lem:netimg(cov)}, start with $\JSet[0]\defeq\JSet$ and
  $\sqryElm[0]\defeq\rqryElm$.
  If $\JSet[0]$ is already a net-image element of $\sqryElm[0]$, we are done.
  Otherwise, by \cref{def:netimg}, there are $\sqryElm[1]\prec \sqryElm[0]$ and
  $\JSet[1]\in\imgFun(\sqryElm[1],\dbElm)$ such that
  $\JSet[0]\subseteq\JSet[1]$.
  If $\JSet[1]$ is not net, repeat the argument.
  Since $\tuple{\denot{\qryElm}}{\preccurlyeq}$ is finite by
  \cref{thm:juqmucchr}, this descending process terminates at some
  $\sqryElm'\preccurlyeq\rqryElm$ and
  $\JSet'\in\nimgFun[\qryElm](\sqryElm',\dbElm)$ with $\JSet\subseteq\JSet'$.

  For Item~\ref{lem:netimg(bij)} of \cref{lem:netimg}, function $\iotaFun$
  intuitively maps atoms to atoms via using $\homFun$.
  Surjectivity of function $\iotaFun$ follows trivially from
  $\JSet=\homFun(\rqryElm)$.
  Suppose it is not injective.
  Then there are two atoms $\atmElm[1],\atmElm[2]\in\rqryElm$ with $\atmElm[1]
  \neq \atmElm[2]$ such that $\homFun(\atmElm[1])=\homFun(\atmElm[2])$.
  This means that (i) these two atoms are unifiable, (ii) $\homFun$ (or rather,
  the restriction of $\homFun$ on the terms of these atoms) is a unifier for
  these two atoms, and (iii) since $\mgu{\atmElm[1],\atmElm[2]}$ is more general
  there is a homomorphism $\homFun_m$ from
  $\mgu{\atmElm[1],\atmElm[2]}(\atmElm[1]) =
  \mgu{\atmElm[1],\atmElm[2]}(\atmElm[2]) $ to
  $\homFun(\atmElm[1])=\homFun(\atmElm[2])$.
  Let $\homFun[\sigma]$ be the extension of $\mgu{\atmElm[1],\atmElm[2]}$ that
  is the identity on all terms not in $\atmElm[1]$,$\atmElm[2]$.
  By Item~\ref{def:muc(fac)} of \cref{def:muc} there is
  $\der{\rqryElm}\in\denot{\qryElm}$, with $\der{\rqryElm}\preccurlyeq\rqryElm$
  and $\der{\rqryElm} \approx \homFun[\sigma](\rqryElm)$.
  Up to this isomorphism, we identify $\der{\rqryElm}$ with
  $\homFun[\sigma](\rqryElm)$ in what follows.
  Since $\atmElm[1] \neq \atmElm[2]$ and
  $\homFun[\sigma](\atmElm[1])=\homFun[\sigma](\atmElm[2])$, the query
  $\homFun[\sigma](\rqryElm)$ has strictly fewer atoms than $\rqryElm$. Hence,
  $\der{\rqryElm}\not\approx\rqryElm$ and, therefore,
  $\der{\rqryElm}\prec\rqryElm$.
  We claim that, there exists a homomorphism $\homFun_{\der{\rqryElm}}$ from
  $\der{\rqryElm}$ to $\JSet$, defined as follows.
  For all constants in $\der{\rqryElm}$, $\homFun_{\der{\rqryElm}}$ is the
  identity.
  For all terms $\tqryElm$ that are not in $\atmElm[1]$ or $\atmElm[2]$ (these
  are mapped to themselves by  $\homFun[\sigma]$, therefore they exist in
  $\der{\rqryElm}$), we define $\homFun_{\der{\rqryElm}}(\tqryElm) =
  \homFun(\tqryElm)$.
  For the rest of the terms $\tqryElm$, $\homFun_{\der{\rqryElm}}(\tqryElm) =
  \homFun_m(\tqryElm)$.
  Now consider an arbitrary atom $\batmElm$ in $\der{\rqryElm}$: if it is ground
  then obviously $\homFun_{\der{\rqryElm}}(\batmElm) = \homFun(\batmElm)$; if
  $\batmElm = \homFun_{\sigma}(\atmElm[1]) = \homFun_{\sigma}(\atmElm[2])$ then
  $\homFun_{\der{\rqryElm}}$ (which on this atom is actually $\homFun_m$) maps
  this to $\homFun(\atmElm[1])=\homFun(\atmElm[2])$.
  Clearly, for any other non-ground atom $\homFun_{\der{\rqryElm}}(\batmElm) =
  \homFun(\batmElm)$.
  Therefore, $\homFun_{\der{\rqryElm}}$ is a homomorphism from $\der{\rqryElm}$
  to $\JSet$ while $\der{\rqryElm}\prec\rqryElm$; thus $\JSet$ cannot be a
  net-image of $\rqryElm$ (contradiction).

  We now prove Item~\ref{lem:netimg(dis)} of \cref{lem:netimg}.  Let
  $\JSet[i]\in\nimgFun[\qryElm](\rqryElm[i],\dbElm)$ for $i\in\{1,2\}$, and
  assume $\JSet[1]\cap\JSet[2]\neq\emptyset$.  By \cref{lem:img}, there is
  $\rqryElm[0]\in\denot{\qryElm}$ such that
  $\rqryElm[0]\preccurlyeq\rqryElm[i]$ for both $i=1,2$, and
  $\JSet[1]\cup\JSet[2]\in\imgFun(\rqryElm[0],\dbElm)$.  Applying
  item~\ref{lem:netimg(cov)} of this lemma, to this image gives a
  net-image element $\JSet[0]$ of some
  $\rqryElm'\preccurlyeq\rqryElm[0]$ such that
  $\JSet[1]\cup\JSet[2]\subseteq\JSet[0]$.  Thus
  $\rqryElm'\preccurlyeq\rqryElm[i]$ and
  $\JSet[i]\subseteq\JSet[0]$ for each $i$.
  Since $\JSet[i]$ is net for $\rqryElm[i]$, this forces
  $\rqryElm'=\rqryElm[i]$; otherwise $\JSet[i]$ is contained in an image of a
  strictly smaller element $\rqryElm'$ which contradicts the net-image
  definition.
  Hence $\rqryElm[1]=\rqryElm[2]=\rqryElm'$.
  By item~\ref{lem:netimg(bij)} of this lemma,
  $\card{\JSet[0]}=\card{\rqryElm'}=\card{\JSet[i]}$.
  The inclusions $\JSet[i]\subseteq\JSet[0]$ are therefore equalities, and so
  $\JSet[1]=\JSet[2]$.

  We next prove Item~\ref{lem:netimg(dec)}.
  Let {\BCQ}s $\rqryElm,\sqryElm\in\denot{\qryElm}$ and homomorphisms
  $\homFun[\rqryElm]\in\HomSet(\rqryElm,\dbElm)$ and
  $\homFun[\sqryElm]\in\HomSet(\sqryElm,\dbElm)$, with
  $\homFun[\rqryElm](\rqryElm)\subseteq \homFun[\sqryElm](\sqryElm)\in
  \nimgFun[\qryElm](\sqryElm,\dbElm)$.
  By Item~\ref{lem:netimg(bij)}, $\homFun[\sqryElm]$ induces a bijection between
  the atoms of $\sqryElm$ and the facts of $\homFun[\sqryElm](\sqryElm)$.
  Hence, for every atom $\atmElm\in\rqryElm$, there is a unique atom
  $\rho(\atmElm)\in\sqryElm$ such that $\homFun[\rqryElm](\atmElm)=
  \homFun[\sqryElm](\rho(\atmElm))$.

  Consider renamed-apart copies of $\rqryElm$ and $\sqryElm$ and the finite
  family of atom equations $\atmElm=\rho(\atmElm)$, for all
  $\atmElm\in\rqryElm$.
  This family is simultaneously unifiable, since the equalities
  $\homFun[\rqryElm](\atmElm)= \homFun[\sqryElm](\rho(\atmElm))$, for all
  $\atmElm\in\rqryElm$, provide a common instance of all these equations.
  Let $\subFun$ be a most general simultaneous unifier and let $\subFun[1]$ and
  $\subFun[2]$ be its restrictions to the copies of $\rqryElm$ and $\sqryElm$,
  respectively.
  Put $\tqryElm\defeq \subFun[1](\rqryElm)\cup\subFun[2](\sqryElm)$.
  By construction, $\subFun[1](\atmElm)=\subFun[2](\rho(\atmElm))$ for every
  $\atmElm\in\rqryElm$, and therefore
  $\subFun[1](\rqryElm)\subseteq\subFun[2](\sqryElm)$.
  Hence, $\tqryElm=\subFun[2](\sqryElm)$.

  By the standard fact on simultaneous unification recalled in \cref{def:unif},
  this unification can be realised by one application of the merging rule of
  Item~\ref{def:muc(mrg)} of \cref{def:muc}, followed by finitely many
  applications of the factorisation rule of Item~\ref{def:muc(fac)}.
  Thus, there exists $\der{\sqryElm}\in\denot{\qryElm}$ such that
  $\der{\sqryElm}\approx\tqryElm$ and $\der{\sqryElm}\preccurlyeq\sqryElm$.

  Moreover, by the universal property of the most general unifier, there is a
  homomorphism $\homFun[0]\in\HomSet(\tqryElm,\dbElm)$ such that
  $\homFun[0]\cmp\subFun[1]=\homFun[\rqryElm]$ and
  $\homFun[0]\cmp\subFun[2]=\homFun[\sqryElm]$.
  In particular, $\homFun[0](\tqryElm)=\homFun[\sqryElm](\sqryElm)$.
  Therefore, $\homFun[\sqryElm](\sqryElm)$ is also an image element of
  $\der{\sqryElm}$, up to the isomorphism $\der{\sqryElm}\approx\tqryElm$.
  Since this image is net for $\sqryElm$ and
  $\der{\sqryElm}\preccurlyeq\sqryElm$, it cannot be that
  $\der{\sqryElm}\prec\sqryElm$.
  Hence, $\der{\sqryElm}=\sqryElm$, and consequently $\tqryElm\approx\sqryElm$.

  We claim that $\subFun[2]$ is an isomorphism from $\sqryElm$ to $\tqryElm$.
  It is surjective by $\tqryElm=\subFun[2](\sqryElm)$.
  It is also injective on atoms: if
  $\subFun[2](\atmElm[1])=\subFun[2](\atmElm[2])$, then
  $\homFun[\sqryElm](\atmElm[1])= \homFun[\sqryElm](\atmElm[2])$, and
  Item~\ref{lem:netimg(bij)} yields $\atmElm[1]=\atmElm[2]$.
  Finally, $\subFun[2]$ cannot identify two distinct terms of $\sqryElm$.
  Indeed, constants are fixed, while identifying two variables, or a variable
  with a constant, would strictly decrease the number of variables, contrary to
  $\tqryElm\approx\sqryElm$.
  Thus, $\subFun[2]$ is an isomorphism.

  We can therefore define $\homFun\defeq\subFun[2]^{-1}\cmp\subFun[1]$.
  This is a homomorphism from $\rqryElm$ to $\sqryElm$, and
  $\homFun[\sqryElm]\cmp\homFun = \homFun[0]\cmp\subFun[2]\cmp
  \subFun[2]^{-1}\cmp\subFun[1] = \homFun[0]\cmp\subFun[1] = \homFun[\rqryElm]$.

  It remains to prove uniqueness.
  Let $\homFun',\homFun''\in\HomSet(\rqryElm,\sqryElm)$ satisfy
  $\homFun[\rqryElm]=\homFun[\sqryElm]\cmp\homFun'
  =\homFun[\sqryElm]\cmp\homFun''$.
  For every atom $\atmElm\in\rqryElm$, we then have
  $\homFun[\sqryElm](\homFun'(\atmElm))= \homFun[\sqryElm](\homFun''(\atmElm))$.
  By Item~\ref{lem:netimg(bij)}, the induced map of $\homFun[\sqryElm]$ on atoms
  is injective, and therefore $\homFun'(\atmElm)=\homFun''(\atmElm)$.
  Since every term of $\rqryElm$ occurs in an atom and constants are fixed, this
  implies $\homFun'=\homFun''$.

  We finally prove Item~\ref{lem:netimg(hom)} of \cref{lem:netimg}.  Fix
  $\sqryElm\in\denot{\qryElm}$.  For every
  $\sqryElm'\preccurlyeq\sqryElm$ and every
  $\JSet'\in\nimgFun[\qryElm](\sqryElm',\dbElm)$, choose one homomorphism
  $\homFun_{\JSet'}\in\HomSet(\sqryElm',\dbElm)$ with
  $\homFun_{\JSet'}(\sqryElm')=\JSet'$.  By
  Item~\ref{lem:netimg(bij)} of \cref{lem:netimg}, $\homFun_{\JSet'}$ induces a
  bijection from the atoms of $\sqryElm'$ to the facts of $\JSet'$.  Hence
  composition with
  $\homFun_{\JSet'}$ maps every homomorphism
  $\homFun[0]\in\HomSet(\sqryElm,\sqryElm')$ to a homomorphism
  $\homFun_{\JSet'}\circ\homFun[0]\in\HomSet(\sqryElm,\dbElm)$ whose image is
  contained in $\JSet'$.  By Item~\ref{lem:netimg(dec)} of
  \cref{lem:netimg}, this composition map is injective and its image is exactly
  the set of homomorphisms from
  $\sqryElm$ to $\dbElm$ whose image is contained in $\JSet'$.
  Conversely, let $\homFun\in\HomSet(\sqryElm,\dbElm)$ and put
  $\JSet=\homFun(\sqryElm)$.  By Item~\ref{lem:netimg(cov)} of
  \cref{lem:netimg}, there is a net-image element $\JSet'$ of some
  $\sqryElm'\preccurlyeq\sqryElm$ containing $\JSet$.  By
  Item~\ref{lem:netimg(dec)} of \cref{lem:netimg}, there is a unique
  $\homFun[0]\in\HomSet(\sqryElm,\sqryElm')$ with
  $\homFun=\homFun_{\JSet'}\circ\homFun[0]$.  The strong disjointness property
  proved above gives uniqueness of $\sqryElm'$ and $\JSet'$.  Thus the
homomorphisms
  from $\sqryElm$ to $\dbElm$ are partitioned by the pairs
  $(\sqryElm',\JSet')$ with
  $\sqryElm'\preccurlyeq\sqryElm$ and
  $\JSet'\in\nimgFun[\qryElm](\sqryElm',\dbElm)$, and each pair contributes
  exactly $\card{\HomSet(\sqryElm,\sqryElm')}$ homomorphisms.  This is precisely
  the displayed equality.
\end{proof}


We are now able to define the notion of our multicanonical instance for a \BJUQ query. This will be an instance which reflects the structure of the query, and in particular the structure  of the unification closure \qryElm; intuitively, for every element of $\qryElm$ a multicanonical instance has a number of data components that are isomorphic to this element. We will later prove that we can ``simulate'' any multiplicity of a \BJUQ on an arbitrary instance by considering only multicanonical instances.

\begin{definition}[Multicanonical Instance]
\label{def:mulcandbs}
  A \emph{multicanonical instance} for a \BCQ $\qryElm$ is an instance $\dbElm$
  for which the following constraints are satisfied:
  \begin{enumerate}[1)]
  \item\label{def:mulcandbs(iso)}
    for each \BCQ $\rqryElm \in \denot{\qryElm}$ and net-image element $\JSet
    \in \nimgFun[\qryElm](\rqryElm, \dbElm)$, it holds that $\JSet \approx
    \rqryElm$;
  \item\label{def:mulcandbs(atm)}
    for all ground atoms $\atmElm \in \dbElm$, there is a \BCQ $\rqryElm \in
    \denot{\qryElm}$ with $\atmElm \in \rqryElm$ such that
    $\nimgFun[\qryElm](\rqryElm, \dbElm) \neq \emptyset$;
  \item\label{def:mulcandbs(cov)}
    for all non-ground components $\JSet \in \comp{\dbElm}$, there is a \BCQ
    $\rqryElm \in \denot{\qryElm}$ with $\JSet \in \nimgFun[\qryElm](\rqryElm,
    \dbElm)$.
  \end{enumerate}
\end{definition}

\begin{example}[Multicanonical instances]
\label{exm:mulcandbs}
The instance used for the \JoFQ table in \cref{exm:img} is multicanonical.  Each
net-image fact is isomorphic to the corresponding closure atom; the only ground
net-image facts are those of $E^{*}$; and different non-ground facts intersect only
at the query constants $c,d$.
%
\end{example}

%

In order to define multicanonical instances with a variable number of
net-images, we first need the notion of a ground selection for a query.
Intuitively, a ground selection consists exactly of the $\prec$-minimal ground
queries that can be constructed using only atoms from the selection itself.
More precisely, for a candidate set $\GSet$ of ground queries, we let
$\LambdaSet[\GSet]$ contain all ground closure elements whose atoms belong to
$\PhiSet[\GSet] \defeq \bigcup \GSet$; then $\GSet$ is a ground selection when
it consists exactly of the $\prec$-minimal elements of $\LambdaSet[\GSet]$.

\begin{definition}[Ground Selection]
\label{def:grnsel}
  A \emph{ground selection} for a \BJUQ $\qryElm$ is a set of ground {\BCQ}s
  $\GSet \subseteq \denot{\qryElm}[\bot]$ such that, for all {\BCQ}s $\rqryElm
  \in \denot{\qryElm}[\bot]$, it holds that $\rqryElm \in \GSet$ \iff $\rqryElm
  \in \LambdaSet[\GSet]$ and $\sqryElm \not\prec \rqryElm$, for every \BCQ
  $\sqryElm \in \LambdaSet[\GSet]$, where $\LambdaSet[\GSet] \defeq \set!{
  \sqryElm \in \denot{\qryElm}[\bot] }{ \sqryElm \subseteq \PhiSet[\GSet] }$ and
  $\PhiSet[\GSet] \defeq \bigcup \GSet$.
  Such a ground selection $\GSet$ is \emph{non-trivial} if
  $\denot{\rqryElm}[\qryElm] \in \LambdaSet[\GSet]$, for all $\rqryElm \in
  \comp[\bot]{\qryElm}$.
\end{definition}


\begin{example}[Two ground selections for the \JUQ example]
\label{exm:grnsel}
Consider our running case of \cref{exm:muc}.
The closure of $\qryElm[6]$ has three ground elements\footnote{Note that we often omit the heads of the queries -- elements of the closure -- in our examples.}: $C_0=\{R(c,c)\}$,
$D_0=\{R(d,d)\}$, and $E=\{R(c,d),R(d,c)\}$. Consider the set  $G_1=\{C_0,D_0\}$; this is a ground selection as there is no other query in $\LambdaSet[\GSet]$ other than $C_0,D_0$. If on the other hand, instead of $E$ our closure contained an element $E^{\prime} = \{ R(c,c), R(d,d) \}$ then $\{C_0,D_0\}$ could not be a ground selection as there is another query in the closure (query $E^{\prime}$) which consists solely of atoms in $\{C_0,D_0\}$ (so it ends up in $\LambdaSet[\{C_0,D_0\}]$) and for which both $E^{\prime} \prec C_0$ and $E^{\prime} \prec D_0$.
In this new scenario, if we omitted either $C_0$ or $D_0$ from the selection the latter would be a ground selection since even though we still have $E^{\prime} \prec C_0$ or  $E^{\prime} \prec D_0$ the query $E^{\prime}$ is not made anymore solely from atoms in the selection.
In the same modified example, another valid ground selection would be $\{E^{\prime}\}$; since here, although we have queries in the closure that are made of selection atoms (queries $C_0$ and $D_0$), $E^{\prime}$ does not map to any of them. For the original \cref{exm:muc}, $\GSet_2 = \{E\}$ is another ground selection since $\LambdaSet[\GSet]$ will contain just $E$.
\end{example}

We now prescribe the number of net-image elements for each closure element.
A ground element can have at most one, namely itself; those assigned $1$ must
form a ground selection.

\begin{definition}[Net-Image Counting]
\label{def:netimgcnt}
  A \emph{net-image counting} for a \BJUQ $\qryElm$ is a function $\nFun \colon
  \denot{\qryElm} \to \SetN$ for which $\nFun(\rqryElm) \leq 1$, for all {\BCQ}s
  $\rqryElm \in \denot{\qryElm}[\bot]$, and $\GSet[\nFun] \defeq \set!{ \rqryElm
  \in \denot{\qryElm}[\bot] }{ \nFun(\rqryElm) = 1 }$ is a ground selection.
  For such a function $\nFun$, the sets $\LambdaSet[\nFun] \defeq
  \LambdaSet[{\GSet[\nFun]}]$ and $\PhiSet[\nFun] \defeq
  \PhiSet[{\GSet[\nFun]}]$ are also introduced.
\end{definition}

\begin{example}[Two net-image countings]
\label{exm:netimgcnt}
The two ground selections above  $G_1=\{C_0,D_0\}$ and $G_1=\{E\}$ support the following net-image countings for the closure of \cref{exm:muc} :
\[
\begin{array}{c|ccccccccc}
 r & C_0&D_0&E&C&D&U_0&U&W_0&W\\ \hline
 n_1(r)&1&1&0&2&3&1&2&1&1\\
 n_2(r)&0&0&1&2&3&1&2&1&1
\end{array}
\]
The ground coordinates are $0$ or $1$.  Moreover $G_{n_1}=G_1$ and
$G_{n_2}=G_2$.  As we will see below we can create multicanonical instances with exactly these sizes of net-images. 
ground part.
\end{example}

The next lemma shows that for every instance we can define a net-image counting function whose values are the sizes of the net-images of the queries in the closure.

\begin{lemma}
\label{lem:arbins2netimgcnt}
  For every \BJUQ $\qryElm$ and instance $\dbElm$, the function $\nFun[\dbElm]
  \colon \denot{\qryElm} \to \SetN$ defined as $\nFun[\dbElm](\sqryElm) \defeq
  \card{\nimgFun[\qryElm](\sqryElm, \dbElm)}$, for each \BCQ $\sqryElm \in
  \denot{\qryElm}$, is a net-image counting for $\qryElm$.
\end{lemma}
\begin{proof}
  Let $\sqryElm\in\denot{\qryElm}$ be ground.  Since homomorphisms fix
  constants, $\sqryElm$ has at most one image in a set instance, namely
  $\sqryElm$ itself.  Hence $\nFun_{\dbElm}(\sqryElm)\leq 1$.

  Let $\nFun[\dbElm]
  \colon \denot{\qryElm} \to \SetN$ defined as $\nFun[\dbElm](\sqryElm) \defeq
  \card{\nimgFun[\qryElm](\sqryElm, \dbElm)}$, for each \BCQ $\sqryElm \in
  \denot{\qryElm}$, and let $\GSet[{\nFun_{\dbElm}}]$ and $\LambdaSet[{\nFun_{\dbElm}}]$ be the sets
  defined in \cref{def:netimgcnt}.  We have to prove that
  $\GSet[{\nFun_{\dbElm}}]$ is a ground selection.
  Suppose first that $\nFun_{\dbElm}(\sqryElm)=1$.  Then
  $\sqryElm\in\nimgFun[\qryElm](\sqryElm,\dbElm)$, and in particular
  $\sqryElm\subseteq\bigcup\GSet[{\nFun_{\dbElm}}]$, so
  $\sqryElm\in\LambdaSet[{\nFun_{\dbElm}}]$.  If there were
  $\tqryElm\in\LambdaSet[{\nFun_{\dbElm}}]$ with
  $\sqryElm\subsetneq\tqryElm$, then $\tqryElm$ would be ground as well, since
  every non-ground element of $\denot{\qryElm}$ has a variable in every atom by
  \cref{thm:juqmucchr}.  Moreover, every atom of $\tqryElm$ belongs to
  $\bigcup\GSet[{\nFun_{\dbElm}}]$.  Such an atom is ground, hence it belongs
  to some ground $\sqryElm[0]\in\denot{\qryElm}$ with
  $\nFun_{\dbElm}(\sqryElm[0])=1$, and so it belongs to $\dbElm$.  Thus
  $\tqryElm$ has the ground image $\tqryElm$ in $\dbElm$.  Since
  $\sqryElm\subsetneq\tqryElm$, we have
  $\tqryElm\prec_{\qryElm}\sqryElm$, contradicting that $\sqryElm$ is a
  net-image element of itself.

  Conversely, suppose that $\sqryElm\in\LambdaSet[{\nFun_{\dbElm}}]$ and that
  there is no $\tqryElm\in\LambdaSet[{\nFun_{\dbElm}}]$ with
  $\sqryElm\subsetneq\tqryElm$.  As above, every atom of $\sqryElm$ belongs to
  $\dbElm$, so $\sqryElm\in\imgFun(\sqryElm,\dbElm)$.  By
  Item~\ref{lem:netimg(cov)} of \cref{lem:netimg}, there are
  $\tqryElm\preccurlyeq\sqryElm$ and
  $\JSet\in\nimgFun[\qryElm](\tqryElm,\dbElm)$ such that
  $\sqryElm\subseteq\JSet$.  Since $\sqryElm$ is ground, $\tqryElm$ is ground as
  well; otherwise a non-ground element would be below a ground one.  Hence
  $\JSet=\tqryElm$ and $\nFun_{\dbElm}(\tqryElm)=1$.  In particular,
  $\tqryElm\subseteq\bigcup\GSet[{\nFun_{\dbElm}}]$, so
  $\tqryElm\in\LambdaSet[{\nFun_{\dbElm}}]$.  From
  $\sqryElm\subseteq\tqryElm$ and the maximality of $\sqryElm$ inside
  $\LambdaSet[{\nFun_{\dbElm}}]$, we obtain $\tqryElm=\sqryElm$.  Therefore
  $\sqryElm\in\nimgFun[\qryElm](\sqryElm,\dbElm)$, and
  $\nFun_{\dbElm}(\sqryElm)=1$.

\end{proof}

The preceding lemma obtains a net-image counting from any instance.
Conversely, the next lemma constructs, for every net-image counting, a
multicanonical instance with exactly the prescribed number of net-image
elements for each closure element.


\begin{lemma}
\label{lem:netimgcnt2mulcan}
  For every \BJUQ $\qryElm$ and net-image counting $\nFun$ for $\qryElm$, there
  exists a multicanonical instance $\dbElm[\nFun]$ for $\qryElm$ such that
  $\card{\nimgFun[\qryElm](\rqryElm, \dbElm[\nFun])} = \nFun(\rqryElm)$, for
  each \BCQ $\rqryElm \in \denot{\qryElm}$.
\end{lemma}
\begin{proof}

  Let $\GSet[\nFun]$ and $\LambdaSet[\nFun]$ be as in
  \cref{def:netimgcnt}.
  For every non-ground $\sqryElm\in\denot{\qryElm}$ and every
  $i\in\{1,\ldots,\nFun(\sqryElm)\}$, choose a fresh copy
  $\JSet^{*}_{\sqryElm,i}$ of $\sqryElm$, together with an isomorphism
  $\alpha_{\sqryElm,i}\colon\sqryElm\to\JSet^{*}_{\sqryElm,i}$ fixing the
  constants of $\con{\qryElm}$.  The fresh values are chosen so that different
  non-ground copies are disjoint outside $\con{\qryElm}$.  For every ground
  $\sqryElm$ with $\nFun(\sqryElm)=1$, put
  $\JSet^{*}_{\sqryElm,1}\defeq\sqryElm$ and let $\alpha_{\sqryElm,1}$ be the
  identity map.  Define
  \[
    \dbElm[\nFun]
    \defeq
    \bigcup_{\sqryElm\in\denot{\qryElm}[\top]}
    \bigcup_{1\leq i\leq\nFun(\sqryElm)}
    \JSet^{*}_{\sqryElm,i}
    \cup
    \bigcup_{\sqryElm\in\GSet[\nFun]}
    \sqryElm.
  \]

  We first prove that each constructed copy is a net-image element.  Let
  $\JSet^{*}_{\sqryElm,i}$ be one of the constructed copies.
  If $\sqryElm$ is non-ground, the isomorphism $\alpha_{\sqryElm,i}$ shows
  that $\JSet^{*}_{\sqryElm,i}$ is an image of $\sqryElm$.
  If it were contained in an image $\JSet'$ of some
  $\tqryElm\prec_{\qryElm}\sqryElm$, then $\JSet'$ would contain a fresh value
  outside $\con{\qryElm}$. Hence $\tqryElm$ is non-ground. By
  \cref{thm:juqmucchr}, every variable of $\tqryElm$ occurs in every atom,
  so this fresh value occurs in every fact of $\JSet'$.
  Thus $\JSet'$ lies in one non-ground component of $\dbElm[\nFun]$.
  The non-ground components of $\dbElm[\nFun]$ are
  exactly the fresh copies.  Thus
  $\JSet'\subseteq\JSet^{*}_{\sqryElm,i}$, and consequently
  $\JSet'=\JSet^{*}_{\sqryElm,i}$.  Composing a homomorphism from
  $\tqryElm$ onto $\JSet'$ with $\alpha_{\sqryElm,i}^{-1}$ gives
  $\sqryElm\preccurlyeq\tqryElm$, contradicting
  $\tqryElm\prec_{\qryElm}\sqryElm$.

  Assume now that $\sqryElm$ is ground.  Then $i=1$ and
  $\JSet^{*}_{\sqryElm,1}=\sqryElm$.  If this image were contained in an
  image of some $\tqryElm\prec_{\qryElm}\sqryElm$, then $\tqryElm$ must be
  ground as well, since a non-ground element cannot be strictly below a ground
  one.  Moreover, the image of $\tqryElm$ is just $\tqryElm$ itself, so
  $\sqryElm\subsetneq\tqryElm$.  Since non-ground copies contain fresh values
  in every fact, the ground facts of $\dbElm[\nFun]$ come only from selected
  ground copies.  Hence $\tqryElm\subseteq\bigcup\GSet[\nFun]$, \ie,
  $\tqryElm\in\LambdaSet[\nFun]$.  This contradicts the ground selection
  condition for $\GSet[\nFun]$.  Hence the ground copy is net.

  Conversely, let $\JSet[0]\in\nimgFun[\qryElm](\sqryElm,\dbElm[\nFun])$.
  Choose a fact $\atmElm\in\JSet[0]$.  This fact belongs to some constructed
  copy $\JSet^{*}_{\tqryElm,i}$, and we have just proved that this copy is a
  net-image element of $\tqryElm$.  By the disjointness of net-image elements,
  Item~\ref{lem:netimg(dis)} of \cref{lem:netimg}, we get
  $\sqryElm=\tqryElm$ and $\JSet[0]=\JSet^{*}_{\tqryElm,i}$.  Thus the
  net-image elements of each $\sqryElm$ are exactly the constructed copies
  $\JSet^{*}_{\sqryElm,i}$ with $1\leq i\leq\nFun(\sqryElm)$.  Hence
  $\card{\nimgFun[\qryElm](\sqryElm,\dbElm[\nFun])}=\nFun(\sqryElm)$.

  It remains to check that $\dbElm[\nFun]$ is multicanonical.  Item
  \ref{def:mulcandbs(iso)} of \cref{def:mulcandbs} follows from the description
  of the net-image elements above.  For Item~\ref{def:mulcandbs(atm)}, every
  ground atom of $\dbElm[\nFun]$ lies in one of the selected ground blocks,
  and that block is a net-image element.  Finally, every non-ground component of
  $\dbElm[\nFun]$ is one of the fresh non-ground copies: selected ground
  blocks contain no fresh values, and distinct fresh copies share only
  constants of $\con{\qryElm}$.  We have already shown
  that each such copy is a net-image element, so Item~\ref{def:mulcandbs(cov)}
  holds.

\end{proof}

For a \BJUQ $\qryElm$ and an arbitrary instance $\dbElm$, the preceding
lemmas produce a multicanonical instance $\dbElm[\qryElm][*]$ with the same
net-image counts. The next theorem shows that $\qryElm$ has the same number of
homomorphisms on both instances, while every \BCQ $\pqryElm$ with
$\con{\pqryElm}\subseteq\con{\qryElm}$ has at most as many homomorphisms on
$\dbElm[\qryElm][*]$ as on $\dbElm$.

\begin{theorem}[Instance Canonisation]
\label{thm:inscan}
  For all {\BJUQ}s $\qryElm$ and instances $\dbElm$, there exists a
  multicanonical instance $\dbElm[\qryElm][*]$ for which the following two
  properties hold true:
  \begin{enumerate}[a)]
  \item\label{thm:inscan(eql)}
    $\card{\HomSet(\qryElm, \dbElm)} = \card!{\HomSet(\qryElm,
    \dbElm[\qryElm][*])}$;
  \item\label{thm:inscan(inq)}
    $\card{\HomSet(\pqryElm, \dbElm)} \geq \card!{\HomSet(\pqryElm,
    \dbElm[\qryElm][*])}$, for all {\BCQ}s $\pqryElm$ such that $\con{\pqryElm}
    \subseteq \con{\qryElm}$.
  \end{enumerate}
\end{theorem}
\begin{proof}
  Let $\nFun_{\dbElm}\colon\denot{\qryElm}\to\SetN$ be defined by
  \[
    \nFun_{\dbElm}(\sqryElm)
    \defeq
    \card{\nimgFun[\qryElm](\sqryElm,\dbElm)}.
  \]
  For each ground $\sqryElm\in\denot{\qryElm}$,
  $\nFun_{\dbElm}(\sqryElm)\leq 1$.
  Moreover, the proof of \cref{lem:arbins2netimgcnt} shows that
  $\GSet[\nFun_{\dbElm}]$ consists exactly of the $\prec$-minimal elements
  of $\LambdaSet[\nFun_{\dbElm}]$.  Applying the construction in the proof
  of \cref{lem:netimgcnt2mulcan} to $\nFun_{\dbElm}$, we obtain a
  multicanonical instance $\dbElm[\qryElm][*]$ such that
  \[
    \card{\nimgFun[\qryElm](\sqryElm,\dbElm[\qryElm][*])}
    =
    \nFun_{\dbElm}(\sqryElm)
    =
    \card{\nimgFun[\qryElm](\sqryElm,\dbElm)}
  \]
  for every $\sqryElm\in\denot{\qryElm}$.  Hence, for every
  $\sqryElm\in\denot{\qryElm}$, choose a bijection
  \[
    \Gamma_{\sqryElm}\colon
    \nimgFun[\qryElm](\sqryElm,\dbElm[\qryElm][*])
    \to
    \nimgFun[\qryElm](\sqryElm,\dbElm).
  \]

  For every $\sqryElm\in\denot{\qryElm}$ and
  $\KSet\in\nimgFun[\qryElm](\sqryElm,\dbElm[\qryElm][*])$, put
  $\JSet_{\KSet}\defeq\Gamma_{\sqryElm}(\KSet)$, and choose a homomorphism
  $\homFun_{\sqryElm,\KSet}\in\HomSet(\sqryElm,\dbElm)$ with
  $\homFun_{\sqryElm,\KSet}(\sqryElm)=\JSet_{\KSet}$.  Since
  $\dbElm[\qryElm][*]$ is multicanonical, choose an isomorphism
  $\alpha_{\sqryElm,\KSet}\colon\sqryElm\to\KSet$ fixing the constants of
  $\con{\qryElm}$.
  Let $\piFun$ map $\adom(\dbElm[\qryElm][*])$ to $\adom(\dbElm)$ and fix
  every constant of $\con{\qryElm}$ in its domain.  On each net-image element
  $\KSet\in\nimgFun[\qryElm](\sqryElm,\dbElm[\qryElm][*])$, define
  \[
    \piFun(\alpha_{\sqryElm,\KSet}(t))
    \defeq
    \homFun_{\sqryElm,\KSet}(t)
  \]
  for every term $t$ of $\sqryElm$.

  \begin{claim}
  \label{clm:inscan(mulcan)}
    The instance $\dbElm[\qryElm][*]$ is multicanonical for $\qryElm$.
    Moreover, there are a $\qryElm$-preserving mapping $\piFun \colon
    \adom(\dbElm[\qryElm][*]) \to \adom(\dbElm)$, i.e., a mapping every constant of $\qryElm$ in its
    domain, and an injection $\iotaFun
    \colon \dbElm[\qryElm][*] \to \dbElm$ such that $\iotaFun(\atmElm) =
    \piFun(\atmElm)$, for all $\atmElm \in \dbElm[\qryElm][*]$.
  \end{claim}
  %
  \begin{proof}
    The instance $\dbElm[\qryElm][*]$ is multicanonical by the construction
    above.
    Every element of the active domain of $\dbElm[\qryElm][*]$ occurs in a
    constructed copy.  The isomorphism for that copy takes a term of its
    closure query to this element, so the displayed formula assigns its
    $\piFun$-image.  To see that the definition is consistent, suppose an
    element $v$ of
    $\dbElm[\qryElm][*]$ can be read both as
    $\alpha_{\sqryElm[1],\KSet[1]}(x_1)$ and as
    $\alpha_{\sqryElm[2],\KSet[2]}(x_2)$.  If $v\in\con{\qryElm}$, both
    definitions fix $v$.  Otherwise, $v$ is a fresh value and occurs in facts of
    $\KSet[1]$ and $\KSet[2]$, and these facts lie in the same non-ground
    component $\CSet$ of $\dbElm[\qryElm][*]$.  By multicanonicity, $\CSet$ is a
    net-image element.  Since $\CSet$ intersects both $\KSet[1]$ and
    $\KSet[2]$, Item~\ref{lem:netimg(dis)} of \cref{lem:netimg} forces
    $\sqryElm[1]=\sqryElm[2]$ and $\KSet[1]=\KSet[2]$.  Since
    $\alpha_{\sqryElm[1],\KSet[1]}$ is injective, $x_1=x_2$, and the two
    definitions agree.  Hence
    $\piFun$ is $\qryElm$-preserving.
    Define
    $\iotaFun\colon \dbElm[\qryElm][*]\to\dbElm$ by
    $\iotaFun(\atmElm)=\piFun(\atmElm)$ for every fact
    $\atmElm\in\dbElm[\qryElm][*]$.  Each fact lies in a constructed copy
    $\KSet$, which is a net-image element by the construction in
    \cref{lem:netimgcnt2mulcan}, and $\piFun$ maps its facts into
    $\JSet_{\KSet}\subseteq\dbElm$.

    The map $\iotaFun$ is injective.  Inside one net-image element
    $\KSet\in\nimgFun[\qryElm](\sqryElm,\dbElm[\qryElm][*])$,
    $\alpha_{\sqryElm,\KSet}$ identifies its facts with the atoms of
    $\sqryElm$.  Item~\ref{lem:netimg(bij)} of \cref{lem:netimg} says that
    $\homFun_{\sqryElm,\KSet}$ maps those atoms to distinct facts of
    $\JSet_{\KSet}$.  Across two different net-image elements,
    injectivity follows from Item~\ref{lem:netimg(dis)} of \cref{lem:netimg}.
    Indeed, suppose that facts from
    $\KSet[i]\in\nimgFun[\qryElm](\sqryElm[i],\dbElm[\qryElm][*])$, for
    $i\in\{1,2\}$, have the same $\piFun$-image.  Then the original net-image
    elements $\JSet_{\KSet[1]}$ and $\JSet_{\KSet[2]}$ intersect, and therefore
    $\sqryElm[1]=\sqryElm[2]$ and
    $\JSet_{\KSet[1]}=\JSet_{\KSet[2]}$.  Since
    $\Gamma_{\sqryElm[1]}$ is a bijection, this implies
    $\KSet[1]=\KSet[2]$, and we are back to the injectivity inside one
    net-image element.
  \end{proof}

  We next compare the net-image elements of each closure query in the original
  and in the constructed instance.

  \begin{claim}
  \label{clm:inscan(eql)}
    For every $\sqryElm \in \denot{\qryElm}$, there exists a bijection
    $\gammaFun[\sqryElm] \colon \nimgFun[\qryElm](\sqryElm, \dbElm[\qryElm][*])
    \to \nimgFun[\qryElm](\sqryElm, \dbElm)$.
  \end{claim}
  \begin{proof}
    Take $\gammaFun[\sqryElm]\defeq\Gamma_{\sqryElm}$.

  \end{proof}

  The next claim gives the one-sided preservation needed for arbitrary
  containing queries.

  \begin{claim}
  \label{clm:inscan(inq)}
    For every \BCQ $\pqryElm$ such that
    $\con{\pqryElm} \subseteq \con{\qryElm}$, there exists an injection
    \[
      \thetaFun[\pqryElm] \colon
      \HomSet(\pqryElm, \dbElm[\qryElm][*]) \to
      \HomSet(\pqryElm, \dbElm).
    \]
  \end{claim}
  \begin{proof}
    Let $\pqryElm$ be a \BCQ such that
    $\con{\pqryElm} \subseteq \con{\qryElm}$.  For
    $\homFun\in\HomSet(\pqryElm,\dbElm[\qryElm][*])$, define
    \[
      \thetaFun[\pqryElm](\homFun)\defeq\piFun\circ\homFun.
    \]
    Since $\piFun$ maps every fact of $\dbElm[\qryElm][*]$ to a fact of
    $\dbElm$ and fixes every constant of $\con{\pqryElm}$, the composition
    is a homomorphism from $\pqryElm$ to $\dbElm$.
    We show that $\thetaFun[\pqryElm]$ is injective.  Let
    $\homFun[1],\homFun[2]\in\HomSet(\pqryElm,\dbElm[\qryElm][*])$ be distinct.
    Since every term in the domain of a homomorphism occurs in a body atom,
    there is an atom $\atmElm\in\pqryElm$ whose image under
    $\homFun[1]$ differs from its image under $\homFun[2]$.  By the injectivity
    of $\iotaFun$ from
    \cref{clm:inscan(mulcan)}, we have
    \[
      \piFun(\homFun[1](\atmElm))
      =
      \iotaFun(\homFun[1](\atmElm))
      \neq
      \iotaFun(\homFun[2](\atmElm))
      =
      \piFun(\homFun[2](\atmElm)).
    \]
    Hence $\piFun\circ\homFun[1]\neq\piFun\circ\homFun[2]$, and
    $\thetaFun[\pqryElm]$ is injective.
  \end{proof}

  We can now prove the two assertions of the theorem.  For
  Item~\ref{thm:inscan(eql)} of \cref{thm:inscan}, let
  $\sqryElm[0]\in\comp{\qryElm}$ and put
  $\der{\sqryElm}[0]\defeq\denot{\sqryElm[0]}[\qryElm]$.  By
  Item~\ref{lem:netimg(hom)} of \cref{lem:netimg}, the number of homomorphisms
  from
  $\der{\sqryElm}[0]$ to an instance is a weighted sum, over all
  $\sqryElm'\preccurlyeq\der{\sqryElm}[0]$, of the cardinalities of
  the net-image sets of $\sqryElm'$.  By \cref{clm:inscan(eql)}, those
  cardinalities are the same in $\dbElm$ and in $\dbElm[\qryElm][*]$.  Hence
  \[
    \card{\HomSet(\der{\sqryElm}[0],\dbElm)}
    =
    \card!{\HomSet(\der{\sqryElm}[0],\dbElm[\qryElm][*])}.
  \]
  Multiplying these equalities over all components of $\qryElm$ and using
  \cref{prp:muc} gives
  \[
    \card{\HomSet(\qryElm,\dbElm)}
    =
    \card!{\HomSet(\qryElm,\dbElm[\qryElm][*])}.
  \]

  Finally, Item~\ref{thm:inscan(inq)} of \cref{thm:inscan} follows immediately
  from \cref{clm:inscan(inq)}, since an injection
  $\HomSet(\pqryElm,\dbElm[\qryElm][*])\to\HomSet(\pqryElm,\dbElm)$ implies the
  displayed inequality on cardinalities.
\end{proof}

\cref{thm:inscan} essentially allows us to look for query containment  (or for counterexamples) only amongst multicanonical instances. This is stated in the next theorem.

\begin{theorem}[Multicanonical Characterisation]
\label{thm:mulcanchr}
  For all {\BJUQ}s $\qryElm$ and {\BCQ}s $\pqryElm$, the following statements
  are equivalent:
  \begin{enumerate}[a)]
  \item\label{thm:mulcanchr(arbdbs)}
    $\qryElm \not\bsinc \pqryElm$;
  \item\label{thm:mulcanchr(mulcandbs)}
    there exists a multicanonical instance $\dbElm[][*]$ for $\qryElm$ such that
    $\qryElm[\bs][{\dbElm[][*]}] \not\subseteq \pqryElm[\bs][{\dbElm[][*]}]$.
  \end{enumerate}
\end{theorem}
\begin{proof}
  The implication from~\ref{thm:mulcanchr(mulcandbs)}
  to~\ref{thm:mulcanchr(arbdbs)} is immediate, since multicanonical instances
  are particular instances.
  Therefore, we can just focus on the converse direction.
  Assume $\qryElm \not\bsinc \pqryElm$.
  Two cases may now arise.
  \begin{itemize}
  \item
    \textbf{[$\con{\pqryElm} \not\subseteq \con{\qryElm}$]:}
    Let $\conElm \in \con{\pqryElm} \setminus \con{\qryElm}$ be a constant not
    occurring in $\qryElm$.
    In addition, let $\dbElm[][*]$ be an instance obtained from a single
    canonical copy of $\qryElm$, choosing all fresh values outside
    $\con{\pqryElm} \setminus \con{\qryElm}$.
    Clearly, $\dbElm[][*]$ is multicanonical for $\qryElm$, the set
    $\HomSet(\qryElm, \dbElm[][*])$ is non-empty, and $\conElm \notin
    \adom(\dbElm[][*])$.
    Since every homomorphism from $\pqryElm$ must fix $\con{\pqryElm}$ and
    $\conElm$ occurs in $\pqryElm$, there are no homomorphisms from $\pqryElm$
    to $\dbElm[][*]$.
    Hence $\card{\HomSet(\qryElm, \dbElm[][*])} > \card{\HomSet(\pqryElm,
    \dbElm[][*])} = 0$, which proves $\dbElm[][*]$ to be the required
    multicanonical witness to bag-set non-inclusion.
  \item
    \textbf{[$\con{\pqryElm} \subseteq \con{\qryElm}$]:}
    Since both queries are Boolean, there is an instance $\dbElm$ such that
    $\card{\HomSet(\qryElm, \dbElm)} > \card{\HomSet(\pqryElm, \dbElm)}$.
    Now, apply~\cref{thm:inscan} to $\qryElm$ and $\dbElm$ and let $\dbElm[][*]
    \defeq \dbElm[\qryElm][*]$ be the resulting multicanonical instance.
    By Item~\ref{thm:inscan(eql)} of \cref{thm:inscan}, the number of
    homomorphisms of $\qryElm$ is preserved, \ie, $\card!{\HomSet(\qryElm,
    \dbElm[][*])} = \card{\HomSet(\qryElm, \dbElm)}$.
    By Item~\ref{thm:inscan(inq)} of the same theorem, the number of
    homomorphisms of $\pqryElm$ cannot increase, \ie, $\card!{\HomSet(\pqryElm,
    \dbElm[][*])} \leq \card{\HomSet(\pqryElm, \dbElm)}$.
    Therefore, $\card!{\HomSet(\qryElm, \dbElm[][*])} > \card!{\HomSet(\pqryElm,
    \dbElm[\qryElm][*])}$, which proves $\dbElm[][*]$ to be the required
    multicanonical witness to bag-set non-inclusion.
    \qedhere
  \end{itemize}
\end{proof}




\subsection{Counting Homomorphisms}
\label{sec:commul;sub:cnthom}

Our objective is to characterise the containment problem via Diophantine
inequalities and for this we would want to characterise the multiplicity of a
query on a multicanonical instance as a polynomial over the sizes of net-images.
In this section we will establish a counting formula for the multiplicity of an
arbitrary query on the multicanonical instance for a query $\qryElm$. First, the
next definition defines a query constructed as (actually isomorphic to) the
conjunction of all the non-ground elements of the closure (called the
multicanonical expansion), that can also be enhanced with all ground atoms of a
query selection.

\begin{definition}[Multicanonical Expansion]
\label{def:mce}
  A \emph{multicanonical expansion} of a \BCQ $\qryElm$ is a \BCQ
  $\mce{\qryElm}$ for which there exists a bijective function $\iotaFun \colon
  \denot{\qryElm}[\top] \to \comp{\mce{\qryElm}}$ such that $\rqryElm \approx
  \iotaFun(\rqryElm)$, for all {\BCQ}s $\rqryElm \in \denot{\qryElm}[\top]$.
  A \emph{multicanonical expansion} of $\qryElm$ \wrt a ground selection $\GSet$
  for $\qryElm$ is the \BCQ $\mce{\qryElm, \GSet} \defeq \mce{\qryElm} \cup
  \PhiSet[\GSet]$.
\end{definition}

\begin{example}[Multicanonical expansions]
\label{exm:mce}
The expansion separates the non-ground closure from the chosen ground part.  Let
$\bar C,\bar D,\bar U_0,\bar U,\bar W_0,\bar W$ be fresh renamed copies of the
six non-ground closure elements of \cref{exm:muc}.  Their union is the common
non-ground expansion $\mce{q_{\mathsf{6}}}$.  The two ground selections of
\cref{exm:grnsel} then give
$\mce{q_{\mathsf{6}},G_1}=\mce{q_{\mathsf{6}}}\cup\{R(c,c),R(d,d)\}$ and
$\mce{q_{\mathsf{6}},G_2}=\mce{q_{\mathsf{6}}}\cup\{R(c,d),R(d,c)\}$.
Thus the two expansions have the same non-ground facts and differ only on their
ground facts.
\end{example}

In a way reminiscent to set containment, in order to characterise the number of
homomorphisms a query $\pqryElm$ has on a multicanonical instance for $\qryElm$
we will use the number and structure of homomorphisms $\pqryElm$ has to the
multicanonical expansion of $\qryElm$. To define this further we need some
auxiliary notions below. In particular we need the notions of a subquery of a
query, a selection function that maps images (of subqueries) of components back
to their components, a component partitioning of a query that defines components
on the constants fixed by a given homomorphism and the notion of isomorphic
embeddings, that is, the isomoprhisms that map closure elements to their
net-image elements.

\begin{definition}
Given a $\BJUQ$ $\qryElm$ we define the following:
\begin{itemize}
\item
  A \emph{subquery set} is the set of all non-empty subsets of a query body:
  $\sub{\QSet} \defeq \set{ \emptyset \neq \sqryElm \subseteq \qryElm }{ \qryElm
  \in \QSet }$.
\item
  A \emph{selection function} maps a subquery of a component $\wqryElm$ of the
  multicanonical expansion to the query in the closure that corresponds to that
  component $\wqryElm$ (per \cref{def:mce}): $\gammaFun \colon
  \sub{\comp{\mce{\qryElm}}} \to \denot{\qryElm}[\top]$ such that
  $\gammaFun(\vqryElm) = \iotaFun^{-1}(\wqryElm)$, for all $\wqryElm \in
  \comp{\mce{\qryElm}}$ and $\emptyset \neq \vqryElm \subseteq \wqryElm$, where
  $\iotaFun$ is the bijection of \cref{def:mce}.
\item
  Given a homomorphism $\homFun$ from a query $\pqryElm$ to $\qryElm$ (or the
  $\mce{\qryElm}$), let $\der{\homFun}$ be the most general homomorphism that
  only grounds the variables that are ground by $\homFun$ and preserves
  everything else as it is; that is, $\der{\homFun}(\varElm) =
  \homFun(\varElm)$, if $\homFun(\varElm) \in \ConSet$, and
  $\der{\homFun}(\varElm) = \varElm$, otherwise.
  By grounding a set of terms of $\pqryElm$ we can consider the components
  $\comp!{\der{\homFun}(\pqryElm)}$.
  Then the \emph{$\homFun$-component set} of a query $\pqryElm$ is the set of
  subsets of $\pqryElm$ consisting of the full preimages of the components of
  $\der{\homFun}(\pqryElm)$: $\comp[][\homFun]{\pqryElm} \defeq \set!{ \set!{
  \atmElm \in \pqryElm }{ \der{\homFun}(\atmElm) \in \tqryElm } }{ \tqryElm \in
  \comp!{\der{\homFun}(\pqryElm)} }$.
  Intuitively, $\comp[][\homFun]{\pqryElm}$ partitions $\pqryElm$ into
  subqueries that, mapped through $\der{\homFun}$, share only constants.
\item
  Given a ground selection $\GSet$, we also set the ground and non-ground parts
  of an $\homFun$-component set respectively: $\comp[\bot][\homFun]{\pqryElm}
  \defeq \set{ \wqryElm \in \comp[][\homFun]{\pqryElm} }{ \homFun(\wqryElm)
  \subseteq \PhiSet[\GSet] }$ and $\comp[\top][\homFun]{\pqryElm} \defeq
  \comp[][\homFun]{\pqryElm} \setminus \comp[\bot][\homFun]{\pqryElm}$.
  Given a homomorphism $\homFun$ from $\pqryElm$ to the multicanonical expansion
  of $\qryElm$ \wrt $\GSet$, \ie, $\homFun \in \HomSet(\pqryElm, \mce{\qryElm,
  \GSet})$, observe that for all $\tqryElm \in \comp[\top][\homFun]{\pqryElm}$,
  $\homFun(\tqryElm) \in \sub{\comp{\mce{\qryElm}}}$.
\item
  An \emph{isomorphic embedding} of $\qryElm$ into an instance $\dbElm$ is an
  indexed sequence of isomorphisms $\alphaFun = \{ \alphaFun[\sqryElm][\JSet]
  \in \IsoSet(\iotaFun(\sqryElm), \JSet) \}_{\sqryElm \in
  \denot{\qryElm}[\top]}^{ \JSet \in \nimgFun[\qryElm](\sqryElm, \dbElm)}$, that
  is, all isomorphisms that map an element of the non-ground closure to a
  net-image element and vice-versa.
  Notice that given $\mce{\qryElm}$ and bijection $\iotaFun$ as in
  \cref{def:mce}, for all $\sqryElm \in \denot{\qryElm}[\top]$ and $\tqryElm
  \subseteq \iotaFun(\sqryElm)$ it holds that $\alphaFun[\sqryElm][\JSet] \in
  \HomSet(\tqryElm, \JSet)$.
\end{itemize}
\end{definition}

Let $\qryElm$ be a \BJUQ and $\dbElm$ an instance. By
\cref{lem:arbins2netimgcnt}, $\dbElm$ induces a net-image counting
$\nFun[\dbElm]$ for $\qryElm$, with ground selection
$\GSet[{\nFun[\dbElm]}]$ as in \cref{def:netimgcnt}. We write
$\GSet[\dbElm]$ for this ground selection when $\qryElm$ is understood.

Let $\pqryElm$ be a \BCQ.
A homomorphism $\homFun$ from $\pqryElm$ to
$\mce{\qryElm,\GSet[\dbElm]}$ records the closure element reached by each
non-ground $\homFun$-component. The same $\homFun$ can give rise to several
homomorphisms into $\dbElm$, which its lens groups together. An isomorphic
embedding $\alphaFun$ supplies an isomorphism from the corresponding expansion
copy to each net-image element of that closure element. The lens of $\homFun$
consists of the homomorphisms that keep the images of its ground
$\homFun$-components and, on each non-ground component, follow $\homFun$ with
the isomorphism supplied by $\alphaFun$ to a net-image element of the closure
element it reaches.

\begin{definition}[Lens]
\label{def:len}
  Given a \BJUQ $\qryElm$, a \BCQ $\pqryElm$, and an instance $\dbElm$, a
  \emph{lens} of a homomorphism $\homFun \in \HomSet(\pqryElm, \mce{\qryElm,
  \GSet[\dbElm]})$ over $\dbElm$ \wrt an isomorphic embedding $\alphaFun$ of
  $\qryElm$ into $\dbElm$ is the set $\LSet[\qryElm][\pqryElm](\dbElm, \homFun,
  \alphaFun)$ of all homomorphisms $\trn{\homFun} \in \HomSet(\pqryElm, \dbElm)$
  satisfying the following conditions:
  \begin{enumerate}[1)]
  \item\label{def:len(grn)}
    for all components $\wqryElm \in \comp[\bot][\homFun]{\pqryElm}$, it holds
    that $\trn{\homFun}(\wqryElm) = \homFun(\wqryElm)$;
  \item\label{def:len(nongrn)}
    for all components $\wqryElm \in \comp[\top][\homFun]{\pqryElm}$, there
    exists a net-image element $\JSet \in
    \nimgFun[\qryElm](\gammaFun(\homFun(\wqryElm)), \dbElm)$ such that
    $\trn{\homFun}(\varElm) =
    \alphaFun[{\gammaFun(\homFun(\wqryElm))}][\JSet](\homFun(\varElm))$, for all
    $\varElm \in \ter{\wqryElm}$.
  \end{enumerate}
\end{definition}

\begin{example}[A complete lens]
\label{exm:len}
Consider again $\qryElm[6]$ and $p_1$ of \cref{exm:4}.
The query $p_1$ has the ground atom $R(c,d)$ and two
$U$-components, on variables $(w_1,t_1)$ and $(w_2,t_2)$.
Let $\dbElm$ be the instance of \cref{exm:img}. Its selected ground copy is
$E^*=\{R(c,d),R(d,c)\}$, corresponding to the ground selection
$G_2=\{E\}$ with $E=\{R(c,d),R(d,c)\}$.
The relevant net-image elements of $\dbElm$ are the two $U$-copies
$U_i=\{S(\lambda_i,\mu_i,e),S(\mu_i,\lambda_i,e)\}$ for $i=1,2$
and the single $U_0$-copy $U^0_1=\{S(\nu,\nu,e)\}$.

Consider a homomorphism $h:p_1\to\mce{\qryElm[6],G_2}$ that
fixes $R(c,d)$, maps the first $U$-component isomorphically to the
copy $\bar U$ of $U$ in the expansion, and maps the second to the
copy $\bar U_0$ of $U_0$ by identifying $w_2$ and $t_2$.
Fix an isomorphic embedding $\alpha$ with isomorphisms
$\bar U\cong U_i$ for $i=1,2$ and $\bar U_0\cong U^0_1$.

In the lens $L^{p_1}_{\qryElm[6]}(\dbElm,h,\alpha)$, the ground atom
is fixed. The first $U$-component can use either of the two $U$
net-image elements of $\dbElm$, while the second must use $U^0_1$.
Thus the lens has two homomorphisms, which differ only in the $U$-copy
used by the first component.
\end{example}

The next theorem proves that lenses actually fully partition the set of homomorphisms from $\pqryElm$ to a multicanonical instance.

\begin{theorem}[Lens Partitioning]
\label{thm:lenpar}
  For each \BJUQ $\qryElm$, \BCQ $\pqryElm$, multicanonical instance $\dbElm$
  for $\qryElm$, and isomorphic embedding $\alphaFun$ of $\qryElm$ into
  $\dbElm$, the set $\set!{ \LSet[\qryElm][\pqryElm](\dbElm, \homFun, \alphaFun)
  }{ \homFun \in \HomSet(\pqryElm, \mce{\qryElm, \GSet[\dbElm]}) }$ of lenses is
  a partition of the set of homomorphisms $\HomSet(\pqryElm, \dbElm)$.
\end{theorem}
\begin{proof}
  Every fact of $\dbElm$ belongs to a unique net-image element.
  A fact containing a value outside
  $\ConSet$ lies in a non-ground component of $\dbElm$, which is
  a net-image element by Item~\ref{def:mulcandbs(cov)} of
  \cref{def:mulcandbs}. If a fact contains only constants,
  Item~\ref{def:mulcandbs(atm)} gives an element
  $\sqryElm\in\denot{\qryElm}$ containing it and having a net-image
  element. This $\sqryElm$ is ground, since every atom of a non-ground
  element of $\denot{\qryElm}$ contains a variable by
  \cref{thm:juqmucchr}. By Item~\ref{def:mulcandbs(iso)}, its net-image
  element is $\sqryElm$ itself. Hence
  $\sqryElm\in\GSet[\dbElm]$, and the fact belongs to
  $\PhiSet[{\GSet[\dbElm]}]$. An isomorphism also maps every atom of a
  non-ground closure element to a fact containing a value outside
  $\ConSet$. Uniqueness in both cases follows from
  Item~\ref{lem:netimg(dis)} of \cref{lem:netimg}.

  Let $\trn{\homFun}\in\HomSet(\pqryElm,\dbElm)$. For each atom $A$ of
  $\pqryElm$, let $(s_A,J_A)$ be the unique pair such that
  $J_A\in\nimgFun[\qryElm](s_A,\dbElm)$ and
  $\trn{\homFun}(A)\in J_A$. If $s_A$ is ground, set
  $\homFun(A)=\trn{\homFun}(A)$. Otherwise, let $\homFun(A)$ be the
  unique atom of $\iotaFun(s_A)$ mapped to $\trn{\homFun}(A)$ by
  $\alphaFun[{s_A}][J_A]$.

  These atom images agree on shared variables. If a variable $\varElm$
  occurs in atoms $A$ and $B$ and $\trn{\homFun}(\varElm)$ is a constant,
  the inverse isomorphisms fix it, and the ground case leaves it
  unchanged. Otherwise, $\trn{\homFun}(A)$ and
  $\trn{\homFun}(B)$ share the value $\trn{\homFun}(\varElm)$ and lie in
  the same non-ground component of $\dbElm$. They therefore have the
  same pair $(s_A,J_A)=(s_B,J_B)$ and use the same inverse isomorphism.
  The constants of $\pqryElm$ are fixed as well. Thus these atom images
  define a homomorphism
  $\homFun\in\HomSet(\pqryElm,\mce{\qryElm,\GSet[\dbElm]})$.

  We check that $\trn{\homFun}$ belongs to the lens of $\homFun$.
  For $T\in\comp[\bot][\homFun]{\pqryElm}$, the construction gives
  $\trn{\homFun}(T)=\homFun(T)$, as required by
  Item~\ref{def:len(grn)} of \cref{def:len}. Now let
  $T\in\comp[\top][\homFun]{\pqryElm}$ and choose an atom $A\in T$.
  Atoms of $T$ are linked through variables that $\homFun$ does not
  map to constants. The isomorphisms and their inverses fix constants,
  so $\trn{\homFun}$ maps each such variable to a value outside
  $\ConSet$. The two facts at each link share that value. Following
  these links, all atoms of $T$ use the same non-ground component of
  $\dbElm$, and hence the same $J_A$. Their images under $\homFun$ lie in
  $\iotaFun(s_A)$, giving $\gammaFun(\homFun(T))=s_A$. For every
  variable $\varElm$ of $T$, the construction gives
  $\trn{\homFun}(\varElm)=
  \alphaFun[{s_A}][J_A](\homFun(\varElm))$. The equality also holds
  for the constants of $T$, which all these maps fix. This is
  Item~\ref{def:len(nongrn)} of \cref{def:len}. Hence the lenses cover
  $\HomSet(\pqryElm,\dbElm)$.

  Finally, suppose $\trn{\homFun}$ belongs to the lenses of
  $\homFun_1$ and $\homFun_2$. For each atom $A$ of $\pqryElm$, the
  fact $\trn{\homFun}(A)$ determines $(s_A,J_A)$. A ground net-image
  element contains only constants, while every fact of a non-ground
  net-image element contains a value outside $\ConSet$. Thus the fact
  also determines which case of \cref{def:len} applies. If $s_A$ is
  ground, the ground condition gives
  $\homFun_1(A)=\trn{\homFun}(A)=\homFun_2(A)$. Otherwise both use
  the non-ground condition. Its net-image element and closure query
  must be $(s_A,J_A)$ by uniqueness, so both use the fixed isomorphism
  $\alphaFun[{s_A}][J_A]$. Its inverse determines the same atom in
  $\iotaFun(s_A)$. Thus
  $\homFun_1(A)=\homFun_2(A)$ for every atom $A$, and
  $\homFun_1=\homFun_2$. The lenses are pairwise disjoint.
\end{proof}

By \cref{thm:lenpar}, counting homomorphisms into a multicanonical instance
amounts to summing the sizes of its lenses. The next theorem gives each lens
size as a product over the non-ground $\homFun$-components, with one factor
equal to the number of net-image elements of the closure element reached by
that component.

\begin{theorem}[Lens Size]
\label{thm:lensiz}
  For each \BJUQ $\qryElm$, \BCQ $\pqryElm$, multicanonical instance $\dbElm$
  for $\qryElm$, homomorphism $\homFun \in \HomSet(\pqryElm, \mce{\qryElm,
  \GSet[\dbElm]})$, and isomorphic embedding $\alphaFun$ of $\qryElm$ into
  $\dbElm$, it holds that $\card{\LSet[\qryElm][\pqryElm](\dbElm, \homFun,
  \alphaFun)} = \prod_{\wqryElm \in \comp[\top][\homFun]{\pqryElm}}
  \card{\nimgFun[\qryElm](\gammaFun(\homFun(\wqryElm)), \dbElm)}$.
\end{theorem}
\begin{proof}
  Let $\homFun$ be any homomorphism from $\pqryElm$ to
  $\mce{\qryElm,\GSet[\dbElm]}$.
  By \cref{def:len}, each homomorphism in the lens of $\homFun$ uses, for every
  $T\in\comp[\top][\homFun]{\pqryElm}$, a net-image element
  $J_T\in\nimgFun[\qryElm](\gammaFun(\homFun(T)),\dbElm)$.
  This element is unique: an atom of $T$ maps to a fact in $J_T$, and
  distinct net-image elements are disjoint by
  Item~\ref{lem:netimg(dis)} of \cref{lem:netimg}.
  A ground $\homFun$-component has a single ground atom as its image,
  so Item~\ref{def:len(grn)} of \cref{def:len} fixes the image of every
  atom in it.

  Conversely, choose one such $J_T$ for every non-ground
  $\homFun$-component $T$. On $T$, apply the fixed isomorphism
  $\alphaFun[{\gammaFun(\homFun(T))}][J_T]$ to the image under $\homFun$;
  on each ground component, use $\homFun$. A variable shared by two different
  $\homFun$-components is mapped by $\homFun$ to a constant, which the
  isomorphisms fix. The maps therefore agree on shared variables.
  The non-ground components map into their chosen net-image elements,
  and the ground components map to the selected ground facts of
  $\dbElm$. Thus they define a homomorphism in the lens.

  Since $\alphaFun$ is fixed, each choice of the $J_T$ gives exactly
  one homomorphism in the lens, and every homomorphism in the lens
  gives one choice for each $T$. Taking cardinalities gives the stated
  product.
\end{proof}

Bringing the two previous theorems together we can now get the multiplicity of $\pqryElm$ on a multicanonical instance by summing the lenses for all homomorphisms $\homFun$ that $\pqryElm$ has on $\mce{\qryElm,G_I}$.

\begin{theorem}[Homomorphism Counting]
\label{thm:homcnt}
  For each \BJUQ $\qryElm$, \BCQ $\pqryElm$, and multicanonical instance
  $\dbElm$ for $\qryElm$, it holds that $\card{\HomSet(\pqryElm, \dbElm)} \!=\!
  \sum_{\homFun \in \HomSet(\pqryElm, \mce{\qryElm, \GSet[\dbElm]})} \!
  \prod_{\wqryElm \in \comp[\top][\homFun]{\pqryElm}} \!
  \card{\nimgFun[\qryElm](\gammaFun(\homFun(\wqryElm)), \dbElm)}$.
\end{theorem}
\begin{proof}
  First observe that there exists at least one isomorphic embedding $\alphaFun$
  of $\qryElm$ into $\dbElm$.
  Indeed, being $\dbElm$ multicanonical, for every \BCQ $\rqryElm \in
  \denot{\qryElm}$ and net-image element $\JSet \in \nimgFun[\qryElm](\rqryElm,
  \dbElm)$, it holds true that $\JSet \approx \rqryElm$ by
  Item~\ref{def:mulcandbs(iso)} of \cref{def:mulcandbs}.
  If $\rqryElm$ is non-ground, then $\rqryElm \approx \iotaFun(\rqryElm)$, and
  by~\cref{def:mce} one may choose an arbitrary isomorphism
  $\alphaFun[\rqryElm][\JSet] \in \IsoSet(\iotaFun(\rqryElm), \JSet)$ from
  $\iotaFun(\rqryElm)$ to $\JSet$, for every such pair of \BCQ $\rqryElm$ and
  net-image element $\JSet$.
  Hence, an isomorphic embedding $\alphaFun$ of $\qryElm$ into $\dbElm$ surely
  exists.

  At this point, by~\cref{thm:lenpar}, each lens
  $\LSet[\qryElm][\pqryElm](\dbElm, \homFun, \alphaFun)$, induced by some
  homomorphism $\homFun
  \in \HomSet(\pqryElm, \mce{\qryElm, \GSet[\dbElm]})$, is a part in a partition
  of $\HomSet(\pqryElm,\dbElm)$.
  Thus, we immediately derive
  \[
    \card{\HomSet(\pqryElm, \dbElm)}
  =
    \sum_{\homFun \in \HomSet(\pqryElm, \mce{\qryElm, \GSet[\dbElm]})}
    \card{\LSet[\qryElm][\pqryElm](\dbElm, \homFun, \alphaFun)}.
  \]
  Now, by applying~\cref{thm:lensiz} to each summand, we obtain the required
  equality:
  \[
    \card{\HomSet(\pqryElm,\dbElm)}
  =
    \sum_{\homFun \in \HomSet(\pqryElm, \mce{\qryElm, \GSet[\dbElm]})}
    \prod_{\wqryElm \in \comp[\top][\homFun]{\pqryElm}}
    \card{\nimgFun[\qryElm](\gammaFun(\homFun(\wqryElm)), \dbElm)}.
    \qedhere
  \]
\end{proof}

For every \BJUQ $\qryElm$, \BCQ $\pqryElm$, ground selection $\GSet$ for
$\qryElm$, and function $\nFun \colon \denot{\qryElm}[\top] \to \SetN$, consider
the formula
\[
  \NFun[\qryElm][\pqryElm](\GSet, \nFun)
\defeq
  \sum_{\homFun \in \HomSet(\pqryElm, \mce{\qryElm, \GSet})}
  \prod_{\wqryElm \in \comp[\top][\homFun]{\pqryElm}}
  \nFun(\gammaFun(\homFun(\wqryElm))).
\]
As an immediate consequence of the previous theorem
and~\cref{lem:arbins2netimgcnt}, one can observe that the number of
homomorphisms from a \BCQ $\pqryElm$ to a multicanonical instance $\dbElm$ for a
\BJUQ $\qryElm$ can be computed by evaluating the formula
$\NFun[\qryElm][\pqryElm]$ on the ground selection $\GSet[\dbElm]$ and net-image
counting $\nFun[\dbElm]$.

\begin{example}[Computing the net-image polynomial]
\label{exm:nfun}
Consider again $\qryElm[6]$ and $p_1$ of \cref{exm:4}, and take the ground
selection $G_2=\{E\}$ of \cref{exm:grnsel}. Let $n$ be a net-image counting
with this ground selection. Every homomorphism fixes the ground atom
$R(c,d)$ of $p_1$.
Each $U$-component of $p_1$ has two homomorphisms to $\bar U$, corresponding
to the two automorphisms of $U$. It has one homomorphism to $\bar U_0$,
which identifies its two variables.
Thus each component contributes $2n(U)+n(U_0)$. The two components have
disjoint variables, so
\[
  \NFun[{\qryElm[6]}][p_1](G_2,n)=(2n(U)+n(U_0))^2.
\]
For the net-image counting $n_2$ of \cref{exm:netimgcnt}, we have
$n_2(U)=2$ and $n_2(U_0)=1$, giving $25$. This agrees with the instance
$\dbElm$ of \cref{exm:img}, where each $U$-component of $p_1$ has five
homomorphisms.
\end{example}

The following corollary follows.

\begin{corollary}[Homomorphism Counting]
\label{cor:homcnt}
  For each \BJUQ $\qryElm$, \BCQ $\pqryElm$ and multicanonical instance $\dbElm$
  for $\qryElm$, it holds that $\card{\HomSet(\pqryElm, \dbElm)} =
  \NFun[\qryElm][\pqryElm](\GSet[\dbElm], \nFun[\dbElm])$.
\end{corollary}




\subsection{Poset Characterisation}
\label{sec:commul;sub:poschr}

The counting formula of \cref{sec:commul;sub:cnthom}
expresses the multiplicity of an arbitrary BCQ in terms of net-image countings, and this will, intuitively,  construct a polynomial on unknown net-image sizes. We could use this polynomial characterisation to examine containment by comparing a polynomial for a containee query $\qryElm$ against a polynomial for the containing query $\pqryElm$. However, our diophantine approach in the next section is extending techniques from \cite{KM19} and \cite{KM25} where we compare a monomial that characterises $\qryElm$ against a polynomial for $\pqryElm$ - and per \cref{prp:muc}, in order to get a monomial we need to work with the total number of images of the closure elements.
The total homomorphism counts of the elements of the closure will give us the multiplicity for $\qryElm$ as the product of these counts for its components. 
To relate the net-images counts with total homomorphism counts, which are called \emph{gross-images} in \cite{KM25}, observe that a net-image copy of a closure element $\sqryElm$ may receive homomorphisms from any closure element $\rqryElm$ with  $\sqryElm \preccurlyeq r$. The contribution of each such copy depends only on r and s. We equip the closure poset of \cref{thm:juqmucchr} with weights recording these contributions and use convolution to recover the net-image counts from the total counts.

\begin{definition}[Multiplicity Weight]
\label{def:mulwgh}
  The \emph{multiplicity weight} of a \BJUQ $\qryElm$ is the function $\mFun
  \colon \!\!\preccurlyeq\; \to \SetN$ defined as $\mFun(\sqryElm, \rqryElm)
  \defeq \card{\HomSet(\rqryElm, \sqryElm)}$, for all {\BCQ}s $\rqryElm,
  \sqryElm \in \denot{\qryElm}$ with $\sqryElm \preccurlyeq \rqryElm$.
\end{definition}

The order records which closure elements can contribute to a given count, while the weight records how many homomorphisms each copy contributes. Notice the direction of the arguments: $\mFun(\sqryElm, \rqryElm)$ counts homomorphisms from $\rqryElm$ into the lower element $\sqryElm$. In particular, the diagonal weight $\mFun(\rqryElm, \rqryElm)$  need not be one. It counts the automorphisms of $\rqryElm$ , so a single net-image copy of $\rqryElm$  may already contribute several homomorphisms from $\rqryElm$.

\begin{example}[Multiplicity weights]
\label{exm:mulwgh}
In the closure of $\qryElm[6]$ from \cref{exm:muc}, every comparable pair
among $C_0,D_0,E,C,D$ has multiplicity weight $1$.

For $U_0,U,W_0,W$, the weights are:
\[
\begin{array}{c|cccc}
 m(s,r) & r=U_0&r=U&r=W_0&r=W\\ \hline
s=U_0&1&1&1&1\\
s=U& &2& &2\\
s=W_0& & &1&1\\
s=W& & & &2
\end{array}
\]
A blank entry means that the row element is not below the column element in
the closure order. The values $m(U,U)=m(W,W)=2$ count the two automorphisms
of $U$ and $W$, respectively; $m(U,W)=2$ counts the two homomorphisms from
$W$ to $U$.

The diagram shows the closure order with weights on its nodes and edges.
At a node $r$ the weight is $m(r,r)$; on an edge joining a lower element
$s$ to an upper element $r$, it is $m(s,r)$.
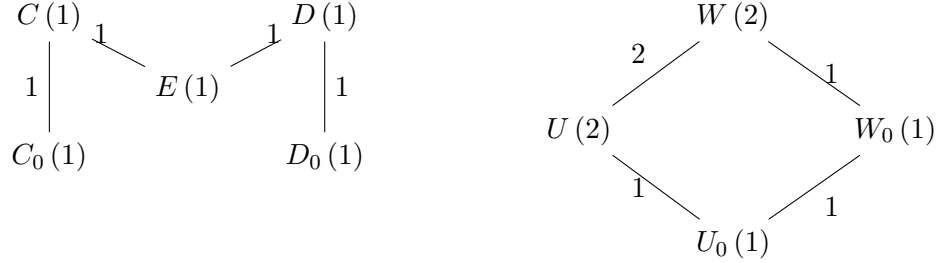
\[
\begin{tikzpicture}[baseline=(current bounding box.center),node distance=12mm]
\node (C) {$C\,(1)$};
\node[right=25mm of C] (D) {$D\,(1)$};
\node[below=of C] (C0) {$C_0\,(1)$};
\node[below=of D] (D0) {$D_0\,(1)$};
\node[below right=3mm and 7mm of C] (E) {$E\,(1)$};
\draw (C0)--node[left] {$1$}(C);
\draw (E)--node[above left] {$1$}(C);
\draw (E)--node[above right] {$1$}(D);
\draw (D0)--node[right] {$1$}(D);

\node[right=42mm of D] (W) {$W\,(2)$};
\node[below left=of W] (U) {$U\,(2)$};
\node[below right=of W] (W0) {$W_0\,(1)$};
\node[below right=of U] (U0) {$U_0\,(1)$};
\draw (U)--node[above left] {$2$}(W);
\draw (W0)--node[right] {$1$}(W);
\draw (U0)--node[left] {$1$}(U);
\draw (U0)--node[below right] {$1$}(W0);
\end{tikzpicture}
\]
The comparable pair $(U_0,W)$ has no edge in the diagram; its weight
$m(U_0,W)=1$ is included in the table.
\end{example}

The following lemma establishes the properties of these weights needed for the arithmetic characterisation. The first property, monotonicity, describes how a fixed lower element contributes to elements above it. The second, invertibility, allows us to pass between net-image counts and total homomorphism counts. Finally, the integrality property identifies a diagonal factor that clears the denominators of each recovered net-image count; this will be used in the next subsection to obtain polynomials with integer coefficients. Note that the last property is precisely the assumption considered at the end
of~\cref{sec:prl;sub:poscon}, when discussing the integrality and non-negativity
of the function $\nFun$.

\begin{lemma}
\label{lem:mulwgh}
  The multiplicity weight of a \BJUQ $\qryElm$ enjoys the following properties:
  \begin{enumerate}[a)]
  \item\label{lem:mulwgh(mon)}
    for all $\tqryElm, \sqryElm, \rqryElm \in \denot{\qryElm}$ with
    $\tqryElm \preccurlyeq \sqryElm \preccurlyeq \rqryElm$, it holds that
    $\mFun(\tqryElm, \sqryElm) \leq \mFun(\tqryElm, \rqryElm)$;
  \item\label{lem:mulwgh(inv)}
    $\mFun$ admits a convolution inverse $\mFun^{-1}$;
  \item\label{lem:mulwgh(int)}
    $\mFun(\rqryElm, \rqryElm) \cdot \mFun^{-1}(\sqryElm, \rqryElm) \in \SetZ$,
    for all {\BCQ}s $\rqryElm, \sqryElm \in \denot{\qryElm}$ with $\sqryElm
    \preccurlyeq \rqryElm$;
  \end{enumerate}
\end{lemma}
\begin{proof}
  We show the three properties separately, by exploiting well-known results in
  the general theories of homomorphisms and convolution on finite posets
  (see~\cref{sec:prl} for further details).
  \begin{itemize}
  \item\textbf{[\ref{lem:mulwgh(mon)}]:}
  Let $\tqryElm \preccurlyeq \sqryElm \preccurlyeq \rqryElm$, and let
    $\homFun \colon \rqryElm \to \sqryElm$ witness $\sqryElm \preccurlyeq
    \rqryElm$. Precomposition with $\homFun$ maps each $\homFun' \in
    \HomSet(\sqryElm, \tqryElm)$ to $\homFun' \cmp \homFun \in
    \HomSet(\rqryElm, \tqryElm)$. This map is injective. Indeed, suppose that
    $\homFun[1]', \homFun[2]' \in \HomSet(\sqryElm, \tqryElm)$ satisfy
    $\homFun[1]' \cmp \homFun = \homFun[2]' \cmp \homFun$. As shown in the proof
    of \cref{thm:juqmucchr}, $\homFun$ is onto the variables of $\sqryElm$.
    Hence $\homFun[1]'$ and $\homFun[2]'$ agree on every variable of
    $\sqryElm$. They also agree on its constants, which homomorphisms fix.
    Thus $\homFun[1]' = \homFun[2]'$, and by \cref{def:mulwgh},
    \[
      \mFun(\tqryElm, \sqryElm)
      = \card{\HomSet(\sqryElm, \tqryElm)}
      \leq \card{\HomSet(\rqryElm, \tqryElm)}
      = \mFun(\tqryElm, \rqryElm).
    \]
  \item\textbf{[\ref{lem:mulwgh(inv)}]:}
    By \cref{thm:juqmucchr}, the structure $\PStr = \tuple {\denot{\qryElm}}
    {\preccurlyeq}$ is a finite poset.
    Moreover, for every $\rqryElm \in \denot{\qryElm}$, the identity
    homomorphism clearly belongs to $\HomSet(\rqryElm, \rqryElm)$.
    Hence, $\mFun(\rqryElm, \rqryElm) = \card{\HomSet(\rqryElm, \rqryElm)} \geq
    1$.
    Therefore, $\mFun$ admits an inverse $\mFun^{-1}$, as required
    by~\cref{lem:mulwgh(inv)}.
  \item\textbf{[\ref{lem:mulwgh(int)}]:}
    We show that $\mFun(\sqryElm, \rqryElm) \equiv 0 \pmod{\mFun(\sqryElm,
    \sqryElm)}$, for all $\rqryElm, \sqryElm \in \denot{\qryElm}$ with $\sqryElm
    \preccurlyeq \rqryElm$, which is a property known to imply the statement
    reported in ~\cref{lem:mulwgh(int)}.
    Consider the natural action of the automorphism $\alphaFun \in
    \AutSet(\sqryElm)$ on the homomorphism $\homFun \in \HomSet(\rqryElm,
    \sqryElm)$ by post-composition: $\alphaFun(\homFun) \defeq \alphaFun \cmp
    \homFun$.
    This action is free.
    Indeed, suppose that $\alphaFun(\homFun) = \homFun$.
    If $\sqryElm$ is ground, then $\alphaFun$ is necessarily the identity.
    Otherwise, by~\cref{thm:juqmucchr}, $\sqryElm$ is connected and
    join-uniform, thus, every variable of $\sqryElm$ occurs in each atom of the
    query.
    Hence, every homomorphism $\homFun \in \HomSet(\rqryElm, \sqryElm)$ is onto
    on the variables of $\sqryElm$.
    Since all automorphisms fix the constants of $\sqryElm$, the equality
    $\alphaFun(\homFun) = \homFun$ forces $\alphaFun$ to fix all terms of
    $\sqryElm$ as well.
    Thus, $\alphaFun$ is the identity.
    As a consequence, $\HomSet(\rqryElm, \sqryElm)$ is partitioned into orbits
    of size $\card{\AutSet(\sqryElm)}$, which implies that
    $\card{\HomSet(\rqryElm, \sqryElm)}$ is always a multiple of
    $\card{\AutSet(\sqryElm)}$, \ie, $\card{\HomSet(\rqryElm, \sqryElm)} \equiv
    0 \pmod{\card{\AutSet(\sqryElm)}}$.
    Now, $\HomSet(\sqryElm, \sqryElm) = \AutSet(\sqryElm)$, since every
    endomorphism of a connected and join-uniform \BCQ is necessarily an
    automorphism.
    Therefore, we obtain $\mFun(\sqryElm, \rqryElm) \equiv 0
    \pmod{\mFun(\sqryElm, \sqryElm)}$, as $\mFun(\sqryElm, \rqryElm) =
    \card{\HomSet(\rqryElm, \sqryElm)}$ and $\mFun(\sqryElm, \sqryElm) =
    \card{\HomSet(\sqryElm, \sqryElm)}$ by~\cref{def:mulwgh}.
    \qedhere
  \end{itemize}
\end{proof}


For a multicanonical instance with net-image counting $\nFun$, each of the $\nFun(\sqryElm)$ copies of $s$ contributes $\mFun(\sqryElm, \rqryElm)$ homomorphisms from $\rqryElm$. The resulting total homomorphism count for $\rqryElm$ is therefore

\[
g(r)=\sum_{s\preceq r}m(s,r)n(s)=(m\star n)(r).
\]

By Lem. \ref{lem:mulwgh}(b), the net-image counts can be recovered uniquely as $n = m^{-1} *g$. However, a function g taking natural-number values need not yield a valid net-image counting under this inverse transformation: the recovered values must be non-negative integers, and their ground coordinates must describe a ground selection. The next definition incorporates precisely these requirements. 


\begin{definition}[Homomorphism Counting]
\label{def:homcnt}
  A \emph{homomorphism counting} for a \BJUQ $\qryElm$ is a function $\gFun
  \colon \denot{\qryElm} \to \SetN$ such that the function $\nFun \defeq
  \mFun^{-1} \star \gFun$ is a net-image counting for $\qryElm$, where $\mFun$
  is the multiplicity weight of $\qryElm$.
  For such a function $\gFun$, the set $\emph{\GSet[\gFun]} \defeq \GSet[\nFun]$
  is also introduced.
  Finally, a homomorphism counting $\gFun$ is \emph{non-trivial} for $\qryElm$
  if $\GSet[\gFun]$ in non-trivial for $\qryElm$.
\end{definition}


\begin{example}[From net-image counting to homomorphism counting]
\label{exm:homcnt}
Take the net-image counting $n_2$ of \cref{exm:netimgcnt} and let
$g=m\star n_2$. The weights in \cref{exm:mulwgh} give:
\[
\begin{array}{c|ccccccccc}
 r & C_0&D_0&E&C&D&U_0&U&W_0&W\\ \hline
 n_2(r)&0&0&1&2&3&1&2&1&1\\
 g(r)&0&0&1&3&4&1&5&2&8
\end{array}
\]
All weights contributing to $g(C)$ are $1$, so
$g(C)=n_2(C)+n_2(C_0)+n_2(E)=3$. The weights
$m(U,U)=m(U,W)=m(W,W)=2$ give
\[
  \begin{aligned}
    g(U)&=2n_2(U)+n_2(U_0)=5,\\
    g(W)&=2n_2(W)+2n_2(U)+n_2(W_0)+n_2(U_0)=8.
  \end{aligned}
\]
Conversely, $n_2=m^{-1}\star g$. For instance,
$n_2(W)=(g(W)-g(U)-g(W_0)+g(U_0))/2=1$.
\end{example}

The next lemma shows that convolution transforms the actual net-image counts of an instance into its actual homomorphism counts, and that inversion recovers exactly those net-image counts. 

\begin{lemma}
\label{lem:nethomcnt}
  Let $\qryElm$ be a \BJUQ and $\nFun, \gFun \colon \denot{\qryElm} \to \SetN$
  two functions such that  $\gFun = \mFun \mathop{\star} \nFun$, where $\mFun$
  is the multiplicity weight of $\qryElm$.
  Then, for every instance $\dbElm$, it holds that $\nFun(\rqryElm) =
  \card{\nimgFun[\qryElm](\rqryElm, \dbElm)}$, for all {\BCQ}s $\rqryElm \in
  \denot{\qryElm}$, \iff $\gFun(\rqryElm) = \card{\HomSet(\rqryElm, \dbElm)}$,
  for all {\BCQ}s $\rqryElm \in \denot{\qryElm}$.
\end{lemma}
\begin{proof}
  First observe that, by Item~\ref{lem:netimg(hom)} of~\cref{lem:netimg},
  \cref{lem:arbins2netimgcnt}, and~\cref{def:mulwgh}, it follows that
  $\card{\HomSet(\rqryElm, \dbElm)} = \sum_{\sqryElm \preccurlyeq \rqryElm}
  \mFun(\sqryElm, \rqryElm) \cdot \nFun[\dbElm](\sqryElm) = (\mFun \star
  \nFun[\dbElm])(\rqryElm)$.
  Now, if $\nFun(\rqryElm) = \card{\nimgFun[\qryElm](\rqryElm, \dbElm)}$, for
  all {\BCQ}s $\rqryElm \in \denot{\qryElm}$, then $\nFun = \nFun[\dbElm]$,
  from which it clearly follows that $\card{\HomSet(\rqryElm,\dbElm)} = (\mFun
  \star \nFun[\dbElm])(\rqryElm) = (\mFun \star \nFun)(\rqryElm) =
  \gFun(\rqryElm)$.
  Conversely, if $\gFun(\rqryElm) = \card{\HomSet(\rqryElm, \dbElm)}$, for all
  {\BCQ}s $\rqryElm \in \denot{\qryElm}$, then $\gFun = \mFun \star
  \nFun[\dbElm]$.
  Now, by~\cref{lem:mulwgh(inv)} of~\cref{lem:mulwgh}, $\mFun$ admits a
  convolution inverse $\mFun^{-1}$.
  Thus, $\nFun[\dbElm] = \mFun^{-1} \star \gFun = \mFun^{-1} \star (\mFun \star
  \nFun) = \nFun$, from which it follows that $\nFun(\rqryElm) =
  \card{\nimgFun[\qryElm](\rqryElm, \dbElm)}$, for all {\BCQ}s $\rqryElm \in
  \denot{\qryElm}$.
\end{proof}

We can now express the multiplicities of $\qryElm$ and $\pqryElm$ using the same homomorphism counting $\gFun$. For $\qryElm$, component independence gives a product of the corresponding closure counts. When $\GSet_{\gFun}$ is non-trivial, every ground component of $\qryElm$ contributes a factor of one, so only the non-ground components need appear in this product. For $\pqryElm$, we apply the formula of \cref{sec:commul;sub:cnthom} after recovering the net-image counting by convolution inversion.

For every BJUQ $\qryElm$, \BCQ $\pqryElm$, ground selection $\GSet$ for
$\qryElm$, and function $\gFun \colon \denot{\qryElm} \to \SetN$, consider the
formulae $\MFun[\qryElm](\gFun) \defeq \prod_{\rqryElm \in \comp[\top]{\qryElm}}
\gFun(\denot{\rqryElm}[\qryElm])$ and $\PFun[\qryElm][\pqryElm](\GSet, \gFun)
\defeq \NFun[\qryElm][\pqryElm](\GSet, \mFun^{-1} \star \gFun)$.

Expression $\MFun[\qryElm](\gFun)$ gives the multiplicity of $\qryElm$, while $\PFun[\qryElm][\pqryElm](\GSet, \gFun)$ gives that of $\pqryElm$ on a multicanonical instance realising $\gFun$. The first expression is a monomial in the total counts; the second is obtained by substituting linear expressions for the net-image parameters. This substitution may simplify the counting formula, as the following example shows.

\begin{samepage}
\begin{example}[Computing the homomorphism-count polynomial]
\label{exm:nethomcnt}
Consider again $q=\qryElm[6]$ and $p=p_1$ of \cref{exm:4}, with the ground
selection $G_2=\{E\}$. Let $g$ be a homomorphism counting for $q$ with this
ground selection, and put $n=m^{-1}\star g$. By \cref{exm:nfun}, each of the
two $U$-components of $p$ contributes $2n(U)+n(U_0)$. The weights of
\cref{exm:mulwgh} give $g(U)=2n(U)+n(U_0)$, so
\[
  \mathsf{P}_q^p(G_2,g)=\mathsf{N}_q^p(G_2,n)
  =(2n(U)+n(U_0))^2=g(U)^2.
\]
For the homomorphism counting $g$ in \cref{exm:homcnt}, we obtain
\[
  \mathsf{M}_q(g)=g(C)g(D)g(W)g(U)=480
  >25=\mathsf{P}_q^p(G_2,g).
\]
\end{example}
\end{samepage}

The next theorem combines this numerical description with the multicanonical characterisation of \cref{thm:mulcanchr}. It shows that if $\pqryElm$ does not contain $\qryElm$ a counterexample instance yields a non-trivial homomorphism counting for which the count of $\qryElm$ exceeds that of $\pqryElm$. Every such counting determines valid net-image counts and hence, by \cref{lem:netimgcnt2mulcan}, a multicanonical instance witnessing non-containment. The search for a counterexample can therefore be expressed entirely in terms of counts on the finite weighted closure poset.

\begin{theorem}[Poset Characterisation]
  \label{thm:poschr}
  For all {\BJUQ}s $\qryElm$ and {\BCQ}s $\pqryElm$, the following statements
  are equivalent:
  \begin{enumerate}[a)]
  \item\label{thm:poschr(noninc)}
    $\qryElm \not\bsinc \pqryElm$;
  \item\label{thm:poschr(mpisol)}
    there exists a non-trivial homomorphism counting $\gFun$ for $\qryElm$ such
    that $\MFun[\qryElm](\gFun) > \PFun[\qryElm][\pqryElm](\GSet[\gFun],
    \gFun)$.
  \end{enumerate}
\end{theorem}
\begin{proof}

  \textbf{[\ref{thm:poschr(noninc)} $\Rightarrow$
  \ref{thm:poschr(mpisol)}].}
  Suppose $\qryElm\not\bsinc\pqryElm$. By \cref{thm:mulcanchr}, there is a
  multicanonical instance $\dbElm[][*]$ for $\qryElm$ with
  \[
    \card{\HomSet(\qryElm,\dbElm[][*])}
    > \card{\HomSet(\pqryElm,\dbElm[][*])}.
  \]
  Let $\nFun(\sqryElm)\defeq
  \card{\nimgFun[\qryElm](\sqryElm,\dbElm[][*])}$ for every
  $\sqryElm\in\denot{\qryElm}$. By \cref{lem:arbins2netimgcnt}, $\nFun$ is a
  net-image counting for $\qryElm$. Set $\gFun\defeq\mFun\star\nFun$ and
  $\GSet\defeq\GSet[\nFun]$. Since $\mFun^{-1}\star\gFun=\nFun$,
  \cref{def:homcnt} gives a homomorphism counting $\gFun$ with
  $\GSet[\gFun]=\GSet$. By \cref{lem:nethomcnt},
  \[
    \gFun(\sqryElm)=\card{\HomSet(\sqryElm,\dbElm[][*])}
    \quad\text{for every }\sqryElm\in\denot{\qryElm}.
  \]

  The strict inequality implies that $\qryElm$ has a homomorphism into
  $\dbElm[][*]$. Hence each ground component $\rqryElm$ of $\qryElm$ has its
  unique ground image there. Every ground fact of a multicanonical instance
  belongs to a selected ground net-image element, so
  $\denot{\rqryElm}[\qryElm]\subseteq\PhiSet[\GSet]$. Thus
  $\denot{\rqryElm}[\qryElm]\in\LambdaSet[\GSet]$ for every ground component,
  and $\GSet$ is non-trivial. Each ground component has one homomorphism, so
  the definition of $\MFun[\qryElm]$ and \cref{prp:muc} give
  \[
    \MFun[\qryElm](\gFun)
    = \card{\HomSet(\qryElm,\dbElm[][*])}.
  \]
  Since $\mFun^{-1}\star\gFun=\nFun$, the definition of
  $\PFun[\qryElm][\pqryElm]$ and \cref{cor:homcnt} give
  \[
    \PFun[\qryElm][\pqryElm](\GSet,\gFun)
    =\card{\HomSet(\pqryElm,\dbElm[][*])}.
  \]
  The chosen $\gFun$ therefore satisfies
  $\MFun[\qryElm](\gFun)>\PFun[\qryElm][\pqryElm](\GSet,\gFun)$.

  \smallskip
  \noindent
  \textbf{[\ref{thm:poschr(mpisol)} $\Rightarrow$
  \ref{thm:poschr(noninc)}].}
  Let $\gFun$ be a non-trivial homomorphism counting satisfying the inequality
  in the statement. Set $\nFun\defeq\mFun^{-1}\star\gFun$ and
  $\GSet\defeq\GSet[\gFun]=\GSet[\nFun]$. By \cref{def:homcnt}, $\nFun$ is a
  net-image counting for $\qryElm$. Applying
  \cref{lem:netimgcnt2mulcan}, choose a multicanonical instance
  $\dbElm[\nFun]$ for which
  \[
    \card{\nimgFun[\qryElm](\sqryElm,\dbElm[\nFun])}
    =\nFun(\sqryElm)
    \quad\text{for every }\sqryElm\in\denot{\qryElm}.
  \]
  In particular, the ground selection of $\dbElm[\nFun]$ is $\GSet$.
  Since $\gFun=\mFun\star\nFun$, \cref{lem:nethomcnt} gives
  \[
    \gFun(\sqryElm)=\card{\HomSet(\sqryElm,\dbElm[\nFun])}
    \quad\text{for every }\sqryElm\in\denot{\qryElm}.
  \]

  Since $\GSet$ is non-trivial, each ground component $\rqryElm$ of $\qryElm$
  satisfies $\denot{\rqryElm}[\qryElm]\in\LambdaSet[\GSet]$. Its facts
  therefore belong to $\PhiSet[\GSet]$, which is contained in
  $\dbElm[\nFun]$, so this component has exactly one homomorphism into
  $\dbElm[\nFun]$. Using the definition of $\MFun[\qryElm]$ and
  \cref{prp:muc}, we obtain
  \[
    \MFun[\qryElm](\gFun)
    =\card{\HomSet(\qryElm,\dbElm[\nFun])}.
  \]
  Since $\mFun^{-1}\star\gFun=\nFun$, the definition of
  $\PFun[\qryElm][\pqryElm]$ and \cref{cor:homcnt} give
  \[
    \PFun[\qryElm][\pqryElm](\GSet,\gFun)
    =\card{\HomSet(\pqryElm,\dbElm[\nFun])}.
  \]
  The assumed inequality now gives
  $\card{\HomSet(\qryElm,\dbElm[\nFun])}
  >\card{\HomSet(\pqryElm,\dbElm[\nFun])}$, so
  $\qryElm\not\bsinc\pqryElm$.

\end{proof}




\subsection{Polynomial Characterisation}
\label{sec:commul;sub:polchr}

The poset characterisation of the previous subsection still expresses its arithmetic condition through homomorphism countings and convolution inversion. We now make this condition explicit as a monomial-polynomial inequality with integer coefficients, together with constraints ensuring that its unknowns describe a realisable counting. Two steps are needed. First, we group the homomorphisms used in the counting formula according to the net-image factors they contribute. Second, we use the symmetries of the target closure elements to cancel the denominators introduced by inversion. The grouping is formalised by the following notion of profile.

\begin{definition}[Profile]
\label{def:prf}
  A \emph{profile} for a \BJUQ $\qryElm$, a \BCQ $\pqryElm$, and a ground
  selection $\GSet$ for $\qryElm$ is a function $\imath \colon
  \comp[\top][\homFun]{\pqryElm} \to \denot{\qryElm}[\top]$ for some
  homomorphism $\homFun \in \HomSet(\pqryElm, \mce{\qryElm, \GSet})$ such that
  $\imath(\wqryElm) = \gammaFun(\homFun(\wqryElm))$, for all $\wqryElm \in
  \comp[\top][\homFun]{\pqryElm}$.
  A homomorphism $\homFun$ related to a given profile $\imath$ via the above
  correspondence is called \emph{$\imath$-homomorphism}.
  The sets of profiles $\imath$ and corresponding $\imath$-homomorphisms are
  denoted by $\Prof[\qryElm][\pqryElm]{\GSet}$ and
  $\HSet[\qryElm][\pqryElm](\GSet, \imath)$, respectively.
\end{definition}

\begin{samepage}
\begin{example}[Profiles]
\label{exm:prf}
Consider again $q=\qryElm[6]$, $p=p_1$, and $G_2=\{E\}$, as in
\cref{exm:len}. A profile records which closure element is reached by each
non-ground component of $p$. Here these are the two $U$-components $U_1,U_2$.
Each has two homomorphisms to $\bar U$, corresponding to the two automorphisms
of $U$, and one to $\bar U_0$. The nine homomorphisms from $p$ to the
expansion therefore split into four profile classes:
\[
\begin{array}{c|c}
\text{profile} & \text{number of homomorphisms}\\ \hline
(U,U)&4\\
(U,U_0)&2\\
(U_0,U)&2\\
(U_0,U_0)&1
\end{array}
\]
The homomorphism $h$ of \cref{exm:len} has profile $(U,U_0)$.
\end{example}
\end{samepage}

The table illustrates a general partition: every homomorphism to the
expansion has a unique profile.

\begin{lemma}
\label{lem:prfpar}
  For each \BJUQ $\qryElm$, \BCQ $\pqryElm$, and ground selection $\GSet$ for
  $\qryElm$, the set $\set!{ \HSet[\qryElm][\pqryElm](\GSet, \imath) }{ \imath
  \in \Prof[\qryElm][\pqryElm]{\GSet} }$ is a partition of the set of
  homomorphisms $\HomSet(\pqryElm, \mce{\qryElm, \GSet})$.
\end{lemma}
\begin{proof}

  Let $\homFun$ be any homomorphism from $\pqryElm$ to
  $\mce{\qryElm,\GSet}$. Define $\imath$ on
  $\comp[\top][\homFun]{\pqryElm}$ by
  $\imath(\wqryElm)\defeq\gammaFun(\homFun(\wqryElm))$. By
  \cref{def:prf}, $\imath$ is a profile and $\homFun$ is an
  $\imath$-homomorphism. If $\homFun$ is also an $\imath'$-homomorphism,
  then $\imath'$ has the same domain and agrees with $\imath$ on every
  $\wqryElm$ in it. Thus $\imath'=\imath$, and every homomorphism belongs
  to exactly one set $\HSet[\qryElm][\pqryElm](\GSet,\imath)$.
  Each such set is non-empty, since every profile is defined through an
  $\imath$-homomorphism.

\end{proof}

Homomorphisms with the same profile reach the same closure element on each
non-ground component and thus contribute the same product to the counting
formula. We can therefore group the sum by profile and multiply each product
by the size of its class.

\begin{lemma}
\label{lem:polrwt}
  Let $\qryElm$ be a \BJUQ, $\pqryElm$ a \BCQ, and $\GSet$ a ground selection
  for $\qryElm$.
  Then, for all functions $\nFun \colon \denot{\qryElm}[\top] \to \SetN$, it
  holds that $\NFun[\qryElm][\pqryElm](\GSet, \nFun) = \sum_{\imath \in
  \Prof[\qryElm][\pqryElm]{\GSet}} \card{\HSet[\qryElm][\pqryElm](\GSet,
  \imath)} \cdot \prod_{\wqryElm \in \dom{\imath}} \nFun(\imath(\wqryElm))$.
\end{lemma}
\begin{proof}

  By definition, $\NFun[\qryElm][\pqryElm](\GSet,\nFun)$ is a sum over
  $\homFun\in\HomSet(\pqryElm,\mce{\qryElm,\GSet})$.
  The partition in \cref{lem:prfpar} lets us group this sum by the sets
  $\HSet[\qryElm][\pqryElm](\GSet,\imath)$. For an
  $\imath$-homomorphism $\homFun$, \cref{def:prf} gives
  $\comp[\top][\homFun]{\pqryElm}=\dom{\imath}$ and
  $\gammaFun(\homFun(\wqryElm))=\imath(\wqryElm)$ for each
  $\wqryElm\in\dom{\imath}$. Its summand in
  $\NFun[\qryElm][\pqryElm](\GSet,\nFun)$ is therefore
  $\prod_{\wqryElm\in\dom{\imath}}\nFun(\imath(\wqryElm))$, independently of the
  $\imath$-homomorphism chosen. There are
  $\card{\HSet[\qryElm][\pqryElm](\GSet,\imath)}$ such summands for each
  $\imath$, which gives the formula.

\end{proof}

\begin{samepage}
\begin{example}[Profile rewriting]
\label{exm:polrwt}
Consider again $q=\qryElm[6]$, $p=p_1$, and $G_2$. For the net-image counting
$n$ of \cref{exm:nfun}, the four profiles of \cref{exm:prf} give
\[
\begin{aligned}
  \mathsf{N}_q^p(G_2,n)
  &=4n(U)^2+2n(U)n(U_0)+2n(U_0)n(U)+n(U_0)^2\\
  &=(2n(U)+n(U_0))^2.
\end{aligned}
\]
The products of the automorphism counts for these profiles are also $4,2,2,1$.
Thus dividing each class size by the corresponding product gives coefficient
$1$ for every profile.
\end{example}
\end{samepage}

For each non-ground component, we can compose the restriction of a homomorphism
with an automorphism of the expansion copy it reaches. The automorphisms can be
chosen independently, and different choices give different homomorphisms with
the same profile. This accounts for the divisibility in the next lemma.

\begin{lemma}
\label{lem:coediv}
  Let $\qryElm$ be a \BJUQ, $\pqryElm$ a \BCQ, and $\GSet$ a ground selection
  for $\qryElm$.
  Then, for every profile $\imath \in \Prof[\qryElm][\pqryElm]{\GSet}$, it holds
  that $\card{\HSet[\qryElm][\pqryElm](\GSet, \imath)}$ is divisible by
  $\prod_{\wqryElm \in \dom{\imath}} \mFun(\imath(\wqryElm), \imath(\wqryElm))$.
\end{lemma}
\begin{proof}

  Let $\imath\in\Prof[\qryElm][\pqryElm]{\GSet}$, and write
  $s_{\wqryElm}\defeq\imath(\wqryElm)$ for each
  $\wqryElm\in\dom{\imath}$. By the definition of $\gammaFun$, each
  $\homFun\in\HSet[\qryElm][\pqryElm](\GSet,\imath)$ maps
  $\wqryElm$ into $\iotaFun(s_{\wqryElm})$. Let
  $\mathcal A\defeq\prod_{\wqryElm\in\dom{\imath}}
  \AutSet(\iotaFun(s_{\wqryElm}))$, with one factor for each
  $\wqryElm$, even when two components have the same $s_{\wqryElm}$.
  For $\alphaFun=(\alphaFun[\wqryElm])_{\wqryElm\in\dom{\imath}}
  \in\mathcal A$ and
  $\homFun\in\HSet[\qryElm][\pqryElm](\GSet,\imath)$, let
  $\alphaFun(\homFun)$ agree with
  $\alphaFun[\wqryElm]\cmp\homFun$ on each non-ground
  $\homFun$-component $\wqryElm$ and with $\homFun$ on its ground
  components.

  These restrictions agree on variables shared by different components:
  $\homFun$ maps such variables to constants, which the automorphisms fix.
  They therefore define a homomorphism. Since the automorphisms fix
  constants and permute variables, the
  $\homFun$-components do not change. Each non-ground component still maps
  into $\iotaFun(s_{\wqryElm})$, so $\alphaFun(\homFun)$ is an
  $\imath$-homomorphism. This defines an action of $\mathcal A$ on
  $\HSet[\qryElm][\pqryElm](\GSet,\imath)$.

  This action is free. If $\alphaFun(\homFun)=\homFun$, then
  $\alphaFun[\wqryElm]$ fixes the image of $\homFun$ on every
  $\wqryElm\in\dom{\imath}$. By \cref{thm:juqmucchr},
  $s_{\wqryElm}$ is connected and join-uniform; hence every variable of
  $\iotaFun(s_{\wqryElm})$ occurs in each of its atoms. Since
  $\wqryElm$ contains an atom, the restriction of $\homFun$ to
  $\wqryElm$ is onto the variables of this copy. Thus
  $\alphaFun[\wqryElm]$ fixes all its variables, as well as its
  constants, and is the identity.

  As observed in the proof of Item~\ref{lem:mulwgh(int)} of
  \cref{lem:mulwgh}, every endomorphism of $s_{\wqryElm}$ is an
  automorphism. Since $\iotaFun(s_{\wqryElm})\approx s_{\wqryElm}$,
  \cref{def:mulwgh} gives
  \[
    \card{\mathcal A}
    =\prod_{\wqryElm\in\dom{\imath}}
      \mFun(s_{\wqryElm},s_{\wqryElm}).
  \]
  Every orbit has this size because the action is free. The orbits partition
  $\HSet[\qryElm][\pqryElm](\GSet,\imath)$, so its cardinality is a
  multiple of the displayed product.

\end{proof}

The resulting quotients will serve as coefficients in the polynomial below.


For each profile $\imath \in \Prof[\qryElm][\pqryElm]{\GSet}$, consider the
number $\beta_{\GSet, \imath} \defeq \frac{\card{\HSet[\qryElm][\pqryElm](\GSet,
\imath)}} {\prod_{\wqryElm \in \dom{\imath}} \mFun(\imath(\wqryElm),
\imath(\wqryElm))}$.
Clearly, by the above \cref{lem:coediv}, this is a natural number, \ie,
$\beta_{\GSet, \imath} \in \SetN$.
To rewrite the net-image counts in \cref{lem:polrwt} using homomorphism counts,
let us choose an arbitrary vector of unknowns $\unkVec \in \UnkSet^{n}$
indexed by the {\BCQ}s in $\denot{\qryElm}[\top]$, \ie, one unknown
$\unkElm[\rqryElm]$ per \BCQ $\rqryElm \in \denot{\qryElm}[\top]$, where $n
\defeq \card{\denot{\qryElm}[\top]}$.
For each $\rqryElm\in\denot{\qryElm}[\top]$, convolution inversion gives an
expression for its net-image count in terms of the unknowns and the ground
selection. Multiplying this expression by $\mFun(\rqryElm,\rqryElm)$, we set
\[
  \DeltaFun[\rqryElm](\GSet, \unkVec)
\defeq
  \mFun(\rqryElm, \rqryElm) \cdot \left( \sum_{\sqryElm \in
  \denot{\qryElm}[\top]}^{\sqryElm \preccurlyeq \rqryElm} \mFun^{-1}(\sqryElm,
  \rqryElm) \cdot \unkElm[\sqryElm] + \sum_{\sqryElm \in
  \LambdaSet[\GSet]}^{\sqryElm \preccurlyeq \rqryElm} \mFun^{-1}(\sqryElm,
  \rqryElm) \right) \,.
\]
By~\cref{lem:mulwgh(int)} of~\cref{lem:mulwgh}, this is a Diophantine linear
polynomial, \ie, a polynomial with integer coefficients.
Finally, also consider the monomial $\MFun[\qryElm](\unkVec) \defeq
\prod_{\rqryElm \in \comp[\top]{\qryElm}} \unkElm[\rqryElm]$ and the polynomial
$\PFun[\qryElm][\pqryElm](\GSet, \unkVec) \defeq \sum_{\imath \in
\Prof[\qryElm][\pqryElm]{\GSet}} \beta_{\GSet, \imath} \cdot \prod_{\wqryElm \in
\dom{\imath}} \DeltaFun[\imath(\wqryElm)](\GSet, \unkVec)$.

Obviously, both are Diophantine as well, since the first is just a monomial with
coefficient $1$, while the second is obtained from Diophantine polynomials by
finite sums and products.

When the unknowns are the homomorphism counts of a multicanonical instance with
non-trivial ground selection $\GSet$, the monomial and polynomial give the
multiplicities of $\qryElm$ and $\pqryElm$ on that instance.


\begin{samepage}
\begin{example}[Polynomial characterisations]
\label{exm:polchr}
Consider $q=\qryElm[6]$ and $p=p_1$ of \cref{exm:4}. For the ground selection
$G_2=\{E\}$, write $\Delta_r$ for $\Delta_r(G_2,u)$. On the six non-ground
closure elements, these polynomials are
\[
\begin{array}{rclcrcl}
\Delta_C&=&u_C-1,&&\Delta_D&=&u_D-1,\\
\Delta_{U_0}&=&u_{U_0},&&\Delta_U&=&u_U-u_{U_0},\\
\Delta_{W_0}&=&u_{W_0}-u_{U_0},&&
\Delta_W&=&u_W-u_U-u_{W_0}+u_{U_0}.
\end{array}
\]
The terms $-1$ in $\Delta_C$ and $\Delta_D$ remove the contribution of $E$.
Since $m(U,U)=m(W,W)=2$, the corresponding net-image counts are
$\Delta_U/2$ and $\Delta_W/2$. The six net-image counts must therefore satisfy
$\Delta_C,\Delta_D,\Delta_{U_0},\Delta_U/2,\Delta_{W_0},\Delta_W/2\in\SetN$.

By \cref{exm:polrwt}, the polynomials for $p_1$ are
\[
\begin{aligned}
  \mathsf{M}_q(u)&=u_Cu_Du_Wu_U,\\
  \mathsf{P}_q^p(G_2,u)&=(\Delta_U+\Delta_{U_0})^2=u_U^2.
\end{aligned}
\]
The values $u_r=g(r)$ from \cref{exm:homcnt} meet these conditions and give
$\mathsf{M}_q(u)=480>25=\mathsf{P}_q^p(G_2,u)$.

The other containing query $q_7$ of \cref{exm:4} adds the independent atom
$R(s,c)$. For any ground selection $G$ for $q$, each $C$-copy contributes one
fact matching this atom, as do $C_0$ and $E$ when selected. These are exactly
the contributions to $u_C$, so
\[
  \mathsf{P}_q^{q_7}(G,u)=u_C\mathsf{M}_q(u).
\]
Since $u_C$ is a factor of $\mathsf{M}_q(u)$, the strict inequality
$\mathsf{M}_q(u)>\mathsf{P}_q^{q_7}(G,u)$ cannot hold for $u_C\in\SetN$.
\end{example}
\end{samepage}

The following theorem restates \cref{thm:poschr} as a polynomial inequality
with constraints on the unknowns just introduced. It is the main result of the ``database'' part of our paper, characterising containment as constrained monomial-polynomial inequality.

\begin{theorem}[Polynomial Characterisation]
\label{thm:polchr}
  For all {\BJUQ}s $\qryElm$ and {\BCQ}s $\pqryElm$, the following statements
  are equivalent:
  \begin{enumerate}[a)]
  \item\label{thm:polchr(noninc)}
    $\qryElm \not\bsinc \pqryElm$;
  \item\label{thm:polchr(mpisol)}
    there exist a non-trivial ground selection $\GSet$ for $\qryElm$ and a
    Diophantine solution $\solVec$ of the inequality $\MFun[\qryElm](\unkVec) >
    \PFun[\qryElm][\pqryElm](\GSet, \unkVec)$ satisfying the following
    constraints, for all $\rqryElm \in \denot{\qryElm}[\top]$:
    \[
      \sum_{\sqryElm \in \denot{\qryElm}[\top]}^{\sqryElm \preccurlyeq \rqryElm}
      \mFun^{-1}(\sqryElm, \rqryElm) \cdot \solElm[\sqryElm] + \sum_{\sqryElm
      \in \LambdaSet[\GSet]}^{\sqryElm \preccurlyeq \rqryElm}
      \mFun^{-1}(\sqryElm, \rqryElm) \in \SetN \,.
    \]
  \end{enumerate}
\end{theorem}
\begin{proof}
  We reduce the result to \cref{thm:poschr} by relating a vector
  $\solVec$ to a homomorphism counting $\gFun$. Let $\GSet$ be a ground
  selection for $\qryElm$, and let $\chiFun[\GSet]$ and
  $\chiFun[{\LambdaSet[\GSet]}]$ be the characteristic functions of
  $\GSet$ and $\LambdaSet[\GSet]$ on $\denot{\qryElm}$.
  The function $\chiFun[\GSet]$ is a net-image counting by
  \cref{def:netimgcnt}. Since it vanishes on non-ground elements, the
  construction in the proof of \cref{lem:netimgcnt2mulcan} gives the
  instance $\PhiSet[\GSet]$, whose net-image counts are
  $\chiFun[\GSet]$. For a ground $\rqryElm$, a homomorphism to this
  instance exists exactly when $\rqryElm\subseteq\PhiSet[\GSet]$ and is
  then unique. Thus \cref{lem:nethomcnt} and the definition of
  $\LambdaSet[\GSet]$ give the ground-coordinate identity
  \[
    (\mFun\star\chiFun[\GSet])(\rqryElm)
    =\card{\HomSet(\rqryElm,\PhiSet[\GSet])}
    =\chiFun[{\LambdaSet[\GSet]}](\rqryElm)
    \quad\text{for every ground }\rqryElm.
  \]
  By \cref{thm:juqmucchr}, every atom of a non-ground element of
  $\denot{\qryElm}$ contains a variable, so no non-ground element lies
  below a ground one. At a ground element, the convolution actions of
  $\mFun$ and $\mFun^{-1}$ therefore depend only on the ground
  coordinates of their argument. Since
  $\mFun^{-1}\star(\mFun\star\chiFun[\GSet])=\chiFun[\GSet]$, the
  identity above also gives
  \[
    (\mFun^{-1}\star\chiFun[{\LambdaSet[\GSet]}])(\rqryElm)
    =\chiFun[\GSet](\rqryElm)
    \quad\text{for every ground }\rqryElm.
  \]

  Let $\solVec$ be a vector of natural numbers indexed
  by $\denot{\qryElm}[\top]$. Extend $\solVec$ to a
  function $\gFun$ on $\denot{\qryElm}$ by setting
  $\gFun(\rqryElm)=\solElm[\rqryElm]$ on non-ground elements and
  $\gFun(\rqryElm)=\chiFun[{\LambdaSet[\GSet]}](\rqryElm)$ on ground elements.
  Put $\nFun\defeq\mFun^{-1}\star\gFun$. The inverse identity gives
  $\nFun=\chiFun[\GSet]$ on ground elements. For a non-ground
  $\rqryElm$, splitting the convolution sum into non-ground and ground
  elements yields
  \[
    \nFun(\rqryElm)
    =\sum_{\sqryElm\in\denot{\qryElm}[\top]}^{\sqryElm\preccurlyeq\rqryElm}
      \mFun^{-1}(\sqryElm,\rqryElm)\cdot\solElm[\sqryElm]
    +\sum_{\sqryElm\in\LambdaSet[\GSet]}^{\sqryElm\preccurlyeq\rqryElm}
      \mFun^{-1}(\sqryElm,\rqryElm).
  \]
  Hence the constraints in Item~\ref{thm:polchr(mpisol)} hold exactly when
  $\nFun$ takes values in $\SetN$ on non-ground elements. In that case,
  $\nFun$ is a net-image counting with $\GSet[\nFun]=\GSet$. By
  \cref{def:homcnt}, $\gFun$ is then a homomorphism counting with
  $\GSet[\gFun]=\GSet$.

  The definition of $\DeltaFun[\rqryElm]$ and the expression for
  $\nFun(\rqryElm)$ give
  \[
    \DeltaFun[\rqryElm](\GSet,\solVec)
    =\mFun(\rqryElm,\rqryElm)\,\nFun(\rqryElm)
    \quad\text{for every }\rqryElm\in\denot{\qryElm}[\top].
  \]
  For each profile $\imath$, the factors
  $\mFun(\imath(\wqryElm),\imath(\wqryElm))$ cancel the denominator
  of $\beta_{\GSet,\imath}$:
  \[
    \beta_{\GSet,\imath}
    \prod_{\wqryElm\in\dom{\imath}}
      \DeltaFun[\imath(\wqryElm)](\GSet,\solVec)
    =\card{\HSet[\qryElm][\pqryElm](\GSet,\imath)}
      \prod_{\wqryElm\in\dom{\imath}}
        \nFun(\imath(\wqryElm)).
  \]
  When the constraints hold, $\nFun$ takes values in $\SetN$, so
  \cref{lem:polrwt} and the definition of
  $\PFun[\qryElm][\pqryElm](\GSet,\gFun)$ before
  \cref{thm:poschr} yield
  \[
    \PFun[\qryElm][\pqryElm](\GSet,\solVec)
    =\NFun[\qryElm][\pqryElm](\GSet,\nFun)
    =\PFun[\qryElm][\pqryElm](\GSet,\gFun).
  \]
  Also $\MFun[\qryElm](\solVec)=\MFun[\qryElm](\gFun)$, since the
  products defining $\MFun[\qryElm]$ involve only non-ground elements,
  on which $\gFun$ agrees with $\solVec$.

  \smallskip
  \noindent
  \textbf{[\ref{thm:polchr(noninc)} $\Rightarrow$
  \ref{thm:polchr(mpisol)}].}
  Suppose $\qryElm\not\bsinc\pqryElm$. By \cref{thm:poschr}, choose a
  non-trivial homomorphism counting $\gFun$ with
  $\MFun[\qryElm](\gFun)>
  \PFun[\qryElm][\pqryElm](\GSet[\gFun],\gFun)$. Set
  $\GSet\defeq\GSet[\gFun]$ and let $\solVec$ be the restriction of
  $\gFun$ to $\denot{\qryElm}[\top]$. The net-image counting
  $\mFun^{-1}\star\gFun$ agrees with $\chiFun[\GSet]$ on ground elements.
  Since $\gFun=\mFun\star(\mFun^{-1}\star\gFun)$, the
  ground-coordinate identity gives
  $\gFun=\chiFun[{\LambdaSet[\GSet]}]$ on ground elements. Thus
  $\gFun$ is the extension of $\solVec$ used above, and
  $\nFun=\mFun^{-1}\star\gFun$ takes values in $\SetN$. The constraints
  hold, and the equalities for $\MFun$ and $\PFun$ turn the chosen
  inequality into $\MFun[\qryElm](\solVec)>
  \PFun[\qryElm][\pqryElm](\GSet,\solVec)$. Since $\gFun$ is
  non-trivial, so is $\GSet$, as required by
  Item~\ref{thm:polchr(mpisol)}.

  \smallskip
  \noindent
  \textbf{[\ref{thm:polchr(mpisol)} $\Rightarrow$
  \ref{thm:polchr(noninc)}].}
  Let $\GSet$ and $\solVec$ satisfy Item~\ref{thm:polchr(mpisol)}, and
  form $\gFun$ and $\nFun$ as above. The constraints make $\nFun$ a
  net-image counting with $\GSet[\nFun]=\GSet$, so $\gFun$ is a
  homomorphism counting with $\GSet[\gFun]=\GSet$ by
  \cref{def:homcnt}. The equalities for $\MFun$ and $\PFun$ turn the
  polynomial inequality into
  $\MFun[\qryElm](\gFun)>
  \PFun[\qryElm][\pqryElm](\GSet[\gFun],\gFun)$.
  Since $\GSet$ is non-trivial, $\gFun$ is non-trivial, and
  \cref{thm:poschr} gives $\qryElm\not\bsinc\pqryElm$.

\end{proof}





\section{A Diophantine Problem}
\label{sec:dphprb}

The polynomial characterisation of \cref{thm:polchr} reduces non-containment to
the existence of a Diophantine solution of a specific monomial-polynomial
inequality.
The unknowns occurring in this inequality, however, cannot be assigned
independently.
Indeed, the gross homomorphism countings characterising multicanonical instances
are obtained from the corresponding net-image countings through a weighted
convolution, while the contribution of the fixed ground part is encoded by an
offset.
Hence, the resulting Diophantine problem has to be solved under additional
arithmetic constraints induced by the underlying finite weighted poset, which we
call \emph{naturality constraints}.

We now leave the database setting temporarily aside and analyse the structural
properties of the resulting arithmetic problem, with the aim of providing an
algorithmic solution to it.

It is important to recall that deciding the existence of a Diophantine solution
for an arbitrary polynomial inequality is, in general, undecidable~\cite{IR95},
due to its tight relationship with Hilbert's 10th
problem~\cite{Dav73,Rob73,Mat93}.
Therefore, our goal is to identify sufficient conditions under which the
solvability problem for the particular inequalities arising from the containment
characterisation becomes decidable.
The technical development generalises the approach introduced in~\cite{KM19} and
further extended in~\cite{KM25}.
In particular, we retain the same monomial-polynomial viewpoint, while the
additional constraints are now expressed through weighted convolution over an
arbitrary finite poset.

Given an $n$-vector of unknowns $\unkVec \in \UnkSet^{n}$ and an exponent vector
$\expVec = \{ \expElm[\unkElm] \}_{\unkElm \in \unkVec} \in \SetN[][n]$, we
denote by $\unkVec^{\expVec}$ the unitary monomial $\prod_{\unkElm \in \unkVec}
\unkElm^{\expElm[\unkElm]}$.
An $n$-\emph{Monomial-Polynomial Inequality} ($n$-\MPI) is an expression
$\MFun(\unkVec) > \PFun(\unkVec)$, where $\MFun(\unkVec) = \unkVec^{\expVec[0]}$
is a unitary monomial and $\PFun(\unkVec) = \sum_{i = 1}^{m} \alphaElm[i]
\unkVec^{\expVec[i]}$, for some $m \in \SetN$, exponent vectors $\{ \expVec[i]
\}_{i = 0}^{m} \subseteq \SetN[][n]$, and non-zero coefficients $\{ \alphaElm[i]
\}_{i = 1}^{m} \subseteq \SetR[\neq 0]$.
If $m = 0$, we let $\PFun$ be the zero polynomial.
An $n$-\emph{Generalised Monomial-Polynomial Inequality} ($n$-\GMPI) is defined
analogously, but allows exponent vectors in $\SetR[\geq 0][n]$.
Thus, its two sides need not be polynomials in the usual sense.

We identify every valuation of $\unkVec$ with its corresponding $n$-vector
$\solVec = \{ \solElm[\unkElm] \}_{\unkElm \in \unkVec} \in \SetR[\geq 0][n]$.
Such a valuation is a \emph{solution} of an $n$-\MPI or $n$-\GMPI if
$\MFun(\solVec) > \PFun(\solVec)$ and it is \emph{Diophantine} if $\solVec \in
\SetN[][n]$.
For a scalar $\rhoElm \in \SetR[\geq 0]$, we write $\solVec^{\rhoElm}$ for the
componentwise power vector $\{ \solElm[\unkElm]^{\rhoElm} \}_{\unkElm \in
\unkVec}$.
This notation is distinct from $\unkVec^{\expVec}$, which denotes the product
defining a monomial.

The degree $\deg{\unkVec^{\expVec}}$ of a possibly generalised monomial is the
sum of its exponents, while the degree $\deg{\PFun}$ of a non-zero possibly
generalised polynomial is the maximum degree of its monomials.
We set $\deg{0} \defeq -\infty$.
For $\PFun(\unkVec) = \sum_{i = 1}^{m} \alphaElm[i] \unkVec^{\expVec[i]}$, we
denote by $\PFun[+](\unkVec)$ its \emph{positive part}, obtained by retaining
exactly the monomials with positive coefficient, and set $\deg[+]{\PFun} \defeq
\deg{\PFun[+]}$.
Thus, $\deg[+]{\PFun} = -\infty$ whenever $\PFun[+]$ is the zero polynomial.
Whenever a maximum of non-negative quantities indexed by the positive monomials
of $\PFun$ is taken, we set the maximum of the empty family to $0$.

To abstract the additional dependencies among the unknowns, we introduce the
following poset structures.

\begin{definition}[Weighted Posets]
  \label{def:wghpos}
  A \emph{weighted poset} is a structure $\PStr = \tuple {\PSet} {\preccurlyeq}
  {\wFun} {\oFun}$, where $\tuple {\PSet} {\preccurlyeq}$ is a finite poset,
  $\wFun \colon {\preccurlyeq} \to \SetR$ is a function defined on the set of
  all comparable pairs of the poset such that $\wFun(\pElm, \pElm) \neq 0$, for
  all $\pElm \in \PSet$, and $\oFun \colon \PSet \to \SetR$ is a function
  defined on all the elements of the poset.
  A \emph{stratified poset} is a weighted poset such that $0 < \wFun(\sElm,
  \rElm) \leq \wFun(\sElm, \pElm)$, for all $\sElm \preccurlyeq \rElm
  \preccurlyeq \pElm$, and $(\wFun^{-1} \star \oFun)(\pElm) \geq 0$, for all
  $\pElm \in \PSet$.
  A \emph{multiplicity poset} is a stratified poset such that $\wFun(\sElm,
  \pElm) \in \SetN[+]$, for all $\sElm \preccurlyeq \pElm$, and $(\wFun^{-1}
  \star \oFun)(\pElm) \in \SetN$, for all $\pElm \in \PSet$.
\end{definition}

The three notions separate the properties needed at different stages of the
argument.
A weighted poset provides the convolutional structure and its inverse.
Stratification additionally guarantees that non-negative contributions are
propagated monotonically along the order and that the offset itself has a
non-negative inverse contribution.
Finally, a multiplicity poset strengthens these requirements to the integral
setting needed for Diophantine solutions.

The inequalities considered below need not involve every element of the poset.
We therefore consider valuations defined on an upward-closed subset and complete
them with the prescribed offset outside their domain.

\begin{definition}[Admissibility \& Naturality]
  \label{def:admnat}
  A function $\fFun \colon \PSet[\top] \to \SetR$ defined on an upward-closed
  subset $\PSet[\top] \subseteq \PSet$ of a weighted poset $\PStr = \tuple
  {\PSet} {\preccurlyeq} {\wFun} {\oFun}$ is \emph{admissible} (\resp,
  \emph{natural}) \wrt $\PStr$ if its extension $\fFun[\oFun] \colon \PSet \to
  \SetR$ defined as $\fFun[\oFun](\pElm) \defeq \fFun(\pElm)$, if $\pElm \in
  \PSet[\top]$, and $\fFun[\oFun](\pElm) \defeq \oFun(\pElm)$, otherwise, for
  all $\pElm \in \PSet$, satisfies the following property: $(\wFun^{-1} \star
  (\fFun[\oFun] - \oFun))(\pElm) \geq 0$ (\resp, $(\wFun^{-1} \star
  (\fFun[\oFun] - \oFun))(\pElm) \in \SetN$), for all $\pElm \in \PSet[\top]$.
\end{definition}

Thus, $\wFun^{-1} \star (\fFun[\oVec] - \oVec)$ isolates the contribution that
is free to vary once the fixed offset has been removed.
Admissibility requires these contributions to be non-negative, while naturality
requires them to be natural numbers.
In particular, the latter is the abstract counterpart of requiring the net-image
countings recovered in \cref{thm:polchr} to be natural numbers.
We use the same qualifiers for solutions whose underlying valuations are
admissible or natural \wrt the corresponding weighted poset.
Similarly, a solution is called \emph{positive} or \emph{increasing} whenever
its underlying valuation has the corresponding property.

A second complication is that the polynomial on the right-hand side of the
inequality may contain negative coefficients.
Such coefficients can arise from M\"obius inversion and prevent us from directly
applying the positive-coefficient argument of~\cite{KM19}.
Nevertheless, the polynomials produced by the homomorphism-counting construction
retain a stronger semantic non-negativity property: on Diophantine natural
valuations, they remain non-negative even after removing one unitary copy of any
monomial having a positive coefficient.
We isolate precisely this property.

\begin{definition}[Strong Non-Negativeness]
\label{def:strnonneg}
  Let $\PStr = \tuple {\unkVec} {\preccurlyeq} {\wFun} {\oVec}$ be a stratified
  poset, $\unkVec[\top] \subseteq \unkVec$ an upward-closed vector of unknowns,
  and $\PFun(\unkVec[\top]) = \sum_{i = 1}^{m} \alphaElm[i]
  \unkVec[\top]^{\expVec[i]}$ a polynomial on $\unkVec[\top]$.
  We say that $\PFun(\unkVec[\top])$ is \emph{strongly non-negative} \wrt
  $\PStr$ if, for every Diophantine natural valuation $\solVec$ \wrt $\PStr$, it
  holds that $\PFun(\solVec) \geq 0$ and $\PFun(\solVec) - \solVec^{\expVec[i]}
  \geq 0$, for every $i \in \num{m}$ with $\alphaElm[i] > 0$.
\end{definition}

Strong non-negativeness allows positive monomials to be treated individually in
the asymptotic arguments below, despite the possible presence of negative
coefficients.
In particular, if $\solVec$ is a Diophantine natural solution of
$\MFun(\unkVec[\top]) > \PFun(\unkVec[\top])$, then $\MFun(\solVec) >
\PFun(\solVec) \geq \solVec^{\expVec[i]}$, for every monomial of $\PFun$ with
positive coefficient.
This is the inequality that will eventually be translated into a strict linear
constraint on exponent vectors.

We can now state the Diophantine problem studied throughout this section.

\begin{problem}[Monomial-Polynomial Inequality Problem]
\label{prb:nmpi}
  Let $\PStr = \tuple {\unkVec} {\preccurlyeq} {\wFun} {\oVec}$ be a
  multiplicity poset and $\MFun(\unkVec[\top]) > \PFun(\unkVec[\top])$ an
  $n$-\MPI defined over an upward-closed $n$-vector of unknowns $\unkVec[\top]
  \subseteq \unkVec$, where $\PFun(\unkVec[\top])$ is strongly non-negative \wrt
  $\PStr$.
  Is there a Diophantine natural solution \wrt $\PStr$ for the $n$-\MPI?
\end{problem}

In line with the methodology presented in~\cite{KM19}, the solution of the
introduced decision problem proceeds by reducing it to a homogeneous system of
linear inequalities and solving that system.
Unlike the direct reduction employed there, however, the additional naturality
constraints require several intermediate steps.
We first isolate the one-variable argument, where solvability is controlled by
the degree of the positive monomials.
We then encode an $n$-variable inequality as a one-variable inequality
parameterised by $n$ exponents, obtaining a homogeneous linear system.
After characterising increasing solutions, we show how admissibility and
naturality can be enforced over an arbitrary multiplicity poset.
Finally, we remove the assumption that all unknowns are positive by restricting
the problem to the upward-closed support that is forced to remain non-zero and
turn the resulting characterisation into a decision procedure.



\subsection{Solutions of $1$-{\GMPI}s}
\label{sec:dphprb;sub:sol1gmpi}

The existence of a solution for an $n$-\GMPI cannot, in general, be inferred by
simply comparing the degrees of its monomial and polynomial.
For a single unknown, however, solvability is closely related to the comparison
between these degrees.
This observation constitutes the basic ingredient of the approach in~\cite{KM19}
and allows us to turn the nonlinear inequality into linear constraints on
exponent vectors.

We first consider the case in which all coefficients of the polynomial are
positive.
If a $1$-\GMPI admits a solution greater than or equal to $1$, the monomial must
have degree strictly greater than every monomial occurring in the polynomial.
Indeed, otherwise, a highest-degree term of the polynomial already dominates the
monomial at every such solution.
The following lemma formalises this necessary condition.

\begin{lemma}
\label{lem:1gmpisoldeg}
  For every $1$-\GMPI $\MFun(\vunkElm) > \PFun(\vunkElm)$ admitting a solution
  in $\SetR[\geq 1]$, with $\PFun(\vunkElm)$ having coefficients in $\SetR[\geq
  1]$, it holds that $\deg{\MFun(\vunkElm)} > \deg{\PFun(\vunkElm)}$.
\end{lemma}
\begin{proof}
  Let $\MFun(\vunkElm) = \vunkElm^{\expElm[0]}$ and $\PFun(\vunkElm) = \sum_{i =
  1}^{m} \alphaElm[i] \vunkElm^{\expElm[i]}$, for some $\{ \alphaElm[i] \}_{i =
  1}^{m} \subseteq \SetR[\geq 1]$ and $\{ \expElm[i] \}_{i = 0}^{m} \subseteq
  \SetR[\geq 0]$, with $m \in \SetN$.
  If $m = 0$, the thesis is immediate, so assume $m > 0$.
  Let $\zetasolElm \in \SetR[\geq 1]$ be a solution of the $1$-\GMPI and
  suppose, by way of contradiction, that $\deg{\MFun(\vunkElm)} \leq
  \deg{\PFun(\vunkElm)}$.
  Then, there exists $j \in \num{m}$ such that $\expElm[0] \leq \expElm[j] =
  \deg{\PFun(\vunkElm)}$.
  Thus, $\frac{\alphaElm[j] \zetasolElm^{\expElm[j]}}{\zetasolElm^{\expElm[0]}}
  \geq \frac{\zetasolElm^{\expElm[j]}}{\zetasolElm^{\expElm[0]}} =
  \zetasolElm^{\expElm[j] - \expElm[0]} \geq 1$, since $\zetasolElm \geq 1$ and
  $\expElm[j] - \expElm[0] \geq 0$.
  This implies that $\alphaElm[j] \zetasolElm^{\expElm[j]} \geq
  \zetasolElm^{\expElm[0]}$.
  Observe now that $\PFun(\zetasolElm) - \alphaElm[j] {\zetasolElm}^{\expElm[j]}
  = \sum_{i = 1, i \neq j}^{m} \alphaElm[i] \zetasolElm^{\expElm[i]} \geq 0$,
  since $\PFun(\vunkElm)$ has coefficients in $\SetR[\geq 1]$.
  Hence,
  \[
    \PFun(\zetasolElm)
  =
    (\PFun(\zetasolElm) - \alphaElm[j] {\zetasolElm}^{\expElm[j]}) +
    \alphaElm[j] {\zetasolElm}^{\expElm[j]}
  \geq
    \alphaElm[j] {\zetasolElm}^{\expElm[j]}
  \geq
    {\zetasolElm}^{\expElm[0]}
  =
    \MFun(\zetasolElm) \,,
  \]
  which contradicts the fact that $\zetasolElm$ is a solution of
  $\MFun(\vunkElm) > \PFun(\vunkElm)$.
\end{proof}

The converse direction holds under substantially weaker assumptions on the
coefficients.
In particular, monomials with negative coefficients can only decrease the value
of the right-hand side and, therefore, do not obstruct the existence of
sufficiently large solutions.
It is consequently enough for the monomial on the left-hand side to have degree
strictly greater than the positive degree of the polynomial.
In this case, every positive monomial becomes asymptotically negligible with
respect to the monomial on the left.
Hence, beyond a suitable threshold, their sum is strictly smaller than the
left-hand side, while monomials with negative coefficients can only reinforce
the inequality.
This yields not just one solution, but an entire unbounded interval of
solutions greater than that threshold.

\begin{lemma}
\label{lem:1gmpidegsol}
  For every $1$-\GMPI $\MFun(\vunkElm) > \PFun(\vunkElm)$ with
  $\deg{\MFun(\vunkElm)} > \deg[+]{\PFun(\vunkElm)}$, there exists a value $\ell
  \in \SetR[\geq 0]$ such that every value $\zetasolElm \in \SetR[> \ell]$ is a
  solution of the $1$-\GMPI.
\end{lemma}
\begin{proof}
  Let $\MFun(\vunkElm) = \vunkElm^{\expElm[0]}$ and $\PFun(\vunkElm) = \sum_{i =
  1}^{m} \alphaElm[i] \vunkElm^{\expElm[i]}$, for some $\{ \alphaElm[i] \}_{i =
  1}^{m} \subseteq \SetR[\neq 0]$ and $\{ \expElm[i] \}_{i = 0}^{m} \subseteq
  \SetR[\geq 0]$, with $m \in \SetN$.
  If $m = 0$, every positive value of $\vunkElm$ is a solution, so the thesis
  follows by taking $\ell \defeq 0$.
  Hence, assume $m > 0$.
  Since $\deg{\MFun(\vunkElm)} > \deg[+]{\PFun(\vunkElm)}$, it holds that
  $\expElm[0] > \expElm[i]$, for all $i \in \num{m}$ with $\alphaElm[i] > 0$.
  Thus, since $\expElm[i] - \expElm[0] < 0$, it holds that $\lim_{\zetasolElm
  \to +\infty} \frac{\alphaElm[i]
  \zetasolElm^{\expElm[i]}}{\zetasolElm^{\expElm[0]}} = 0$, for all $i \in
  \num{m}$ with $\alphaElm[i] > 0$.
  By expanding the definition of converging limit, we derive the existence of a
  real number $\ell_{i} \in \SetR[\geq 0]$ such that $\frac{\alphaElm[i]
  \zetasolElm^{\expElm[i]}}{\zetasolElm^{\expElm[0]}} < \frac{1}{m}$, for all
  $\zetasolElm \in \SetR$ with $\zetasolElm > \ell_{i}$.
  Now, let $\ell \defeq \max[i \in \num{m}][{\alphaElm[i] > 0}]\,
  \ell_{i} \in \SetR[\geq 0]$.
  By construction, $\frac{\alphaElm[i]
  \zetasolElm^{\expElm[i]}}{\zetasolElm^{\expElm[0]}} < \frac{1}{m}$, \ie,
  $\alphaElm[i] \zetasolElm^{\expElm[i]} < \frac{\zetasolElm^{\expElm[0]}}{m}$,
  for all $\zetasolElm \in \SetR$ with $\zetasolElm > \ell$ and $i \in \num{m}$
  with $\alphaElm[i] > 0$.
  The same inequality trivially holds when $\alphaElm[i] < 0$, since
  $\zetasolElm > 0$.
  Hence, it holds for every $i \in \num{m}$.
  As a consequence,
  \[
    \PFun(\zetasolElm)
  =
    \sum_{i = 1}^{m} \alphaElm[i] \zetasolElm^{\expElm[i]}
  <
    \sum_{i = 1}^{m} \frac{\zetasolElm^{\expElm[0]}}{m}
  =
    \zetasolElm^{\expElm[0]}
  =
    \MFun(\zetasolElm) \,.
  \]
  Hence, every $\zetasolElm \in \SetR[> \ell]$ is a solution of the $1$-\GMPI.
\end{proof}




\subsection{$n$-Exponent-Parameterised $1$-{\MPI}s}
\label{sec:dphprb;sub:nexppar1mpi}

The previous subsection shows that the solvability of a $1$-\GMPI can be
controlled by comparing the degree of its left-hand monomial with those of the
positive monomials on the right-hand side.
We now exploit this observation by parameterising these degrees through an
$n$-vector of exponents.
This turns the relevant degree comparisons into homogeneous linear inequalities
over the parameters.

Let $\unkVec \in \UnkSet^{n}$ be an $n$-vector of parameter unknowns and
$\vunkElm \in \UnkSet \setminus \unkVec$ a distinguished value unknown.
An $n$-\emph{Exponent-Parameterised $1$-\GMPI} ($n$-exppar $1$-\GMPI) is an
inequality $\MFun(\unkVec; \vunkElm) > \PFun(\unkVec; \vunkElm)$ of the form
$\MFun(\unkVec; \vunkElm) = \vunkElm^{\expVec[0]^{\intercal} \cdot \unkVec}$ and
$\PFun(\unkVec; \vunkElm) = \sum_{i = 1}^{m} \alphaElm[i]
\vunkElm^{\expVec[i]^{\intercal} \cdot \unkVec}$, for some $m \in \SetN$,
non-zero coefficients $\{ \alphaElm[i] \}_{i = 1}^{m} \subseteq \SetR[\neq 0]$,
and exponent vectors $\{ \expVec[i] \}_{i = 0}^{m} \subseteq \SetR[\geq 0][n]$.
If all these exponent vectors belong to $\SetN[][n]$, we call the inequality an
$n$-exppar $1$-\MPI.
Substituting any valuation $\rsolVec \in \SetR[\geq 0][n]$ for the parameter
unknowns yields an ordinary $1$-\GMPI whose exponents are the scalar products
$\expVec[i]^{\intercal} \cdot \rsolVec$.
Thus, the parameter vector controls the relative growth of all monomials with
respect to the common value unknown $\vunkElm$.

Let $\PStr = \tuple {\unkVec} {\preccurlyeq}$ be a finite poset on the parameter
unknowns.
We associate with every $n$-exppar $1$-\GMPI the following homogeneous linear
system over $\unkVec$.
The first family of inequalities requires the degree of the left-hand monomial
to be strictly greater than that of every positive monomial on the right.
The second family requires the parameter valuation to be increasing along the
poset, while the last one forces its minimal coordinates to be non-negative.

\[
  \tuple* {\MFun(\unkVec; \vunkElm) \!>\! \PFun(\unkVec; \vunkElm)} {\PStr}
\defeq
  \Bigl\{
    (\expVec[0] - \expVec[i])^{\intercal}\! \cdot \unkVec > 0\,
  \Bigr\}_{i \in \num{m}}^{\alphaElm[i] > 0}
\cup
  \Bigl\{
    \,\unkElm[2] > \unkElm[1]\,
  \Bigr\}_{\unkElm[1], \unkElm[2] \in \unkVec}^{\unkElm[1] \vartriangleleft
  \unkElm[2]}
\cup
  \Bigl\{
    \,\unkElm \geq 0\,
  \Bigr\}_{\unkElm \in \min[\preccurlyeq] \unkVec}
\]

Accordingly, every solution of this system is a non-negative increasing
$n$-vector.
We shall also use a positive variant of the system, obtained by strengthening
the non-negativity constraints on the minimal unknowns to strict positiveness.

\[
  \tuple* {\MFun(\unkVec; \vunkElm) \!>\! \PFun(\unkVec; \vunkElm)} {\PStr} [+]
\defeq
  \Bigl\{
    (\expVec[0] - \expVec[i])^{\intercal}\! \cdot \unkVec > 0\,
  \Bigr\}_{i \in \num{m}}^{\alphaElm[i] > 0}
\cup
  \Bigl\{
    \,\unkElm[2] > \unkElm[1]\,
  \Bigr\}_{\unkElm[1], \unkElm[2] \in \unkVec}^{\unkElm[1] \vartriangleleft
  \unkElm[2]}
\cup
  \Bigl\{
    \,\unkElm > 0\,
  \Bigr\}_{\unkElm \in \min[\preccurlyeq] \unkVec}
\]

Every solution of the positive system is therefore a positive increasing
$n$-vector.
In both cases, we identify a solution of the system with its corresponding
valuation vector for $\unkVec$.
Notice that the two systems depend on the exponent vectors, the signs of the
coefficients, and the poset, but not on the magnitudes of the coefficients.
This reflects the asymptotic nature of the reduction: only positive monomials
can obstruct eventual domination by the monomial on the left-hand side.

\begin{example}
\label{exm:sys}
  The linear system records that the monomial on the left must dominate each
  positive monomial on the right.
  For the negative \JoFQ example, whose associated inequality is $u_{1} u_{2}
  u_{3} > u_{1}^{2} + (u_{2} - u_{12} - u_{23}) + u_{3}^{3}$, the three positive
  monomials on the right-hand side produce the following degree inequalities,
  using coordinates $(a_{1}, a_{12}, a_{2}, a_{23}, a_{3})$:
  \[
  \begin{array}{rcll}
    -d_{a_{1}} + d_{a_{2}} + d_{a_{3}}
  & > &
    0,
  & \text{from } u_{1}^{2}; \\
    d_{a_{1}} + d_{a_{3}}
  & > &
    0,
  & \text{from } u_{2}, \\
    d_{a_{1}} + d_{a_{2}} - 2 d_{a_{3}}
  & > &
    0,
  & \text{from } u_{3}^{3}.
  \end{array}
  \]
  These inequalities must be combined with the order constraints
  $d_{a_{12}} < d_{a_{1}}$, $d_{a_{12}} < d_{a_{2}}$, $d_{a_{23}} < d_{a_{2}}$,
  $d_{a_{23}} < d_{a_{3}}$, and $d_{a_{12}}, d_{a_{23}} > 0$.
  One integer solution is, for example, $d_{a_{12}} = d_{a_{23}} = 1$,
  $d_{a_{1}} = 2$, $d_{a_{2}} = 3$, and $d_{a_{3}} = 2$.

  For the \JUQ inequality $u_{C} u_{D} u_{W} u_{U} > u_{U}^{2}$, there is only
  one positive monomial on the right-hand side.
  It gives the degree inequality $d_{C} + d_{D} + d_{W} - d_{U} > 0$.
  The poset strictness constraints are $d_{C_{0}} < d_{C}$, $d_{E} < d_{C}$,
  $d_{D_{0}} < d_{D}$, $d_{E} < d_{D}$, $d_{U_{0}} < d_{U}$, $d_{U_{0}} <
  d_{W_{0}}$, $d_{U} < d_{W}$, and $d_{W_{0}} < d_{W}$, with $d_{C_{0}},
  d_{D_{0}}, d_{E}, d_{U_{0}} > 0$.
  One integer solution is $d_{C_{0}} = d_{D_{0}} = d_{E} = d_{U_{0}} = 1$,
  $d_{C} = d_{D} = d_{U} = d_{W_{0}} = 2$, and $d_{W} = 3$.
\end{example}

The two systems differ only in the constraints imposed on the minimal unknowns.
Since all degree and order constraints are strict, a sufficiently small uniform
positive perturbation of any solution preserves them while making all minimal
coordinates positive.
Hence, the distinction does not affect feasibility.

\begin{lemma}
\label{lem:syseqv}
  Let $\PStr = \tuple {\unkVec} {\preccurlyeq}$ be a finite poset and
  $\MFun(\unkVec; \vunkElm) > \PFun(\unkVec; \vunkElm)$ an $n$-exppar $1$-\GMPI.
  Then, the $\tuple* {\MFun(\unkVec; \vunkElm) \!>\! \PFun(\unkVec; \vunkElm)}
  {\PStr}$-system admits a solution \iff the $\tuple* {\MFun(\unkVec; \vunkElm)
  \!>\! \PFun(\unkVec; \vunkElm)} {\PStr} [+]$-system admits a solution.
\end{lemma}
\begin{proof}
  We only focus on one of the two directions of the equivalence, the other being
  trivial.
  \begin{itemize}
  \item
    \textbf{[$\Rightarrow$]}
    Let $\rsolVec = \{ \rsolElm[\unkElm] \}_{\unkElm \in \unkVec} \in \SetR[\geq
    0][n]$ be a solution of the $\tuple*{\MFun(\unkVec; \vunkElm) \!>\!
    \PFun(\unkVec; \vunkElm)}{\PStr}$-system.
    We shall show that, for all sufficiently small $\varepsilon \in \SetR[+]$,
    the $n$-vector $\ssolVec[\varepsilon] \defeq \rsolVec + \varepsilon
    \vec*{1}$ is a solution of the $\tuple*{\MFun(\unkVec; \vunkElm) \!>\!
    \PFun(\unkVec; \vunkElm)}{\PStr} [+]$-system.
    First, let $\unkElm \in \min[\preccurlyeq] \unkVec$.
    Since $\rsolElm[\unkElm] \geq 0$, we have $\ssolElm[\varepsilon\unkElm] =
    \rsolElm[\unkElm] + \varepsilon > 0$.
    Now, let $\unkElm[1], \unkElm[2] \in \unkVec$ with $\unkElm[1]
    \vartriangleleft \unkElm[2]$.
    Since the same value $\varepsilon$ is added to all coordinates, we have
    $\ssolElm[\varepsilon{\unkElm[2]}] - \ssolElm[\varepsilon{\unkElm[1]}] =
    (\rsolElm[{\unkElm[2]}] + \varepsilon) - (\rsolElm[{\unkElm[1]}] +
    \varepsilon) = \rsolElm[{\unkElm[2]}] - \rsolElm[{\unkElm[1]}] > 0$.
    Thus, all increasing constraints are preserved as well.
    It remains to verify that, for all sufficiently small $\varepsilon > 0$, the
    $n$-vector $\ssolVec[\varepsilon]$ satisfies all the inequalities
    $(\expVec[0] - \expVec[i])^{\intercal} \cdot \unkVec > 0$, for all $i \in
    \num{m}$ with $\alphaElm[i] > 0$.
    If there is no $i \in \num{m}$ with $\alphaElm[i] > 0$, there are no degree
    constraints and any $\varepsilon > 0$ completes the proof.
    Hence, assume that at least one such index exists.
    Let $\etaElm[i](\varepsilon) \defeq (\expVec[0] - \expVec[i])^{\intercal}
    \cdot \ssolVec[\varepsilon] = (\expVec[0] - \expVec[i])^{\intercal} \cdot
    (\rsolVec + \varepsilon \vec*{1}) = (\expVec[0] - \expVec[i])^{\intercal}
    \cdot \rsolVec + (\expVec[0] - \expVec[i])^{\intercal} \cdot \varepsilon
    \vec*{1}$.
    Observe that $\etaElm[i](\varepsilon)$ is a continuous function of
    $\varepsilon$ with $\etaElm[i](0) > 0$.
    Therefore, for every $i \in \num{m}$ with $\alphaElm[i] > 0$, there exists
    $\ell_{i} > 0$ such that $\etaElm[i](\varepsilon) > 0$, for all $0 \leq
    \varepsilon \leq \ell_{i}$.
    After choosing $\varepsilon \defeq \min_{i \in \num{m}}^{\alphaElm[i] > 0}
    \ell_{i}$, we obtain $(\expVec[0] - \expVec[i])^{\intercal} \cdot
    \ssolVec[\varepsilon] = \etaElm[i](\varepsilon) > 0$, for all $i \in
    \num{m}$ with $\alphaElm[i] > 0$.
    Therefore, $\ssolVec[\varepsilon]$ is a solution of the $\tuple*
    {\MFun(\unkVec; \vunkElm) \!>\! \PFun(\unkVec; \vunkElm)} {\PStr}
    [+]$-system.
    \qedhere
  \end{itemize}
\end{proof}

We can now relate solutions of an $n$-exppar $1$-\GMPI to solutions of its
associated linear system.
Fixing the parameter vector of a solution yields an ordinary $1$-\GMPI.
When the polynomial has coefficients in $\SetR[\geq 1]$, \cref{lem:1gmpisoldeg}
forces the degree of its left-hand monomial to dominate all degrees on the
right.
Together with the increasing assumption on the parameter vector, this gives all
inequalities of the associated system.

\begin{lemma}
\label{lem:nexppar1mpisolsys}
  Let $\PStr = \tuple {\unkVec} {\preccurlyeq}$ be a finite poset and
  $\MFun(\unkVec; \vunkElm) > \PFun(\unkVec; \vunkElm)$ an $n$-exppar $1$-\GMPI,
  admitting a solution $(\rsolVec, \zetasolElm) \in \SetR[\geq 0][n] \times
  \SetR[\geq 1]$, with $\rsolVec$ increasing and $\PFun(\unkVec; \vunkElm)$
  having coefficients in $\SetR[\geq 1]$.
  Then, the $\tuple* {\MFun(\unkVec; \vunkElm) \!>\! \PFun(\unkVec; \vunkElm)}
  {\PStr} [+]$-system admits a solution.
\end{lemma}
\begin{proof}
  Let $\MFun(\unkVec; \vunkElm) = \vunkElm^{\expVec[0]^{\intercal} \cdot
  \unkVec}$ and $\PFun(\unkVec; \vunkElm) = \sum_{i = 1}^{m} \alphaElm[i]
  \vunkElm^{\expVec[i]^{\intercal} \cdot \unkVec}$, for some $\{ \alphaElm[i]
  \}_{i = 1}^{m} \subseteq \SetR[\geq 1]$ and $\{ \expVec[i] \}_{i = 0}^{m}
  \subseteq \SetR[\geq 0][n]$, with $m \in \SetN$.
  If $m = 0$, the thesis is immediate, so assume $m > 0$.
  Now, consider the $1$-\GMPI $\MFun[][\ast](\vunkElm) >
  \PFun[][\ast](\vunkElm)$, where $\MFun[][\ast](\vunkElm) \defeq
  \MFun(\rsolVec; \vunkElm) = \vunkElm^{(\expVec[0]^{\intercal} \cdot
  \rsolVec)}$ and $\PFun[][\ast](\vunkElm) \defeq \PFun(\rsolVec; \vunkElm) =
  \sum_{i = 1}^{m} \alphaElm[i] \vunkElm^{(\expVec[i]^{\intercal} \cdot
  \rsolVec)}$.
  Clearly, $\deg{\MFun[][\ast](\vunkElm)} = \deg{\MFun(\rsolVec; \vunkElm)}$,
  $\deg{\PFun[][\ast](\vunkElm)} = \deg{\PFun(\rsolVec; \vunkElm)}$, and
  $\PFun[][\ast](\vunkElm)$ has coefficients in $\SetR[\geq 1]$.
  Moreover, $\zetasolElm$ is a solution of this $1$-\GMPI, since
  $\MFun[][\ast](\zetasolElm) = \MFun(\rsolVec; \zetasolElm) >
  \PFun(\rsolVec; \zetasolElm) = \PFun[][\ast](\zetasolElm)$.
  Thus, by \cref{lem:1gmpisoldeg} applied to $\MFun[][\ast](\vunkElm) >
  \PFun[][\ast](\vunkElm)$, it holds that $\deg{\MFun[][\ast](\vunkElm)} >
  \deg{\PFun[][\ast](\vunkElm)}$.
  Since $\deg{\PFun[][\ast](\vunkElm)} = \max_{i \in \num{m}}
  \expVec[i]^{\intercal} \cdot \rsolVec$, we have $\expVec[0]^{\intercal} \cdot
  \rsolVec > \expVec[i]^{\intercal} \cdot \rsolVec$ and, thus, $(\expVec[0] -
  \expVec[i])^{\intercal} \cdot \rsolVec > 0$, for all $i \in \num{m}$.
  In addition, $\rsolElm[\unkElm] \geq 0$, for all $\unkElm \in
  \min[\preccurlyeq] \unkVec$, since $\rsolVec \in \SetR[\geq 0][n]$, and
  $\rsolElm[{\unkElm[2]}] > \rsolElm[{\unkElm[1]}]$, for all $\unkElm[1],
  \unkElm[2] \in \unkVec$ with $\unkElm[1] \vartriangleleft \unkElm[2]$, since
  $\rsolVec$ is increasing.
  As a consequence, $\rsolVec$ is a solution of the $\tuple* {\MFun(\unkVec;
  \vunkElm) \!>\! \PFun(\unkVec; \vunkElm)} {\PStr}$-system.
  Finally, by \cref{lem:syseqv}, the $\tuple* {\MFun(\unkVec; \vunkElm) \!>\!
  \PFun(\unkVec; \vunkElm)} {\PStr} [+]$-system admits a solution too, as
  required by the statement of the lemma.
\end{proof}

For the converse construction, we shall need a positive integral parameter
vector of bounded representation size.
Since an $n$-exppar $1$-\MPI has natural exponent vectors, its associated linear
system has integer, hence rational, coefficients and is homogeneous.
Thus, standard bounds for integer solutions of rational linear systems apply.

An $n$-vector $\dsolVec = \{ \dsolElm[\unkElm] \}_{\unkElm \in \unkVec} \in
\SetN[][n]$ is \emph{$s$-bit-bounded}, for $s \in \SetN$, if $\size*{\dsolVec}
\defeq \sum_{\unkElm \in \unkVec} (1 + \ceil{\log_{2}(\dsolElm[\unkElm] + 1)})
\leq s$.
The \emph{facet complexity} of a linear system is the maximum length of the
binary encoding of its inequalities, determined by the binary representations of
the rational coefficients occurring in them~\cite{Sch86}.
For the systems associated with an $n$-exppar $1$-\MPI, every strict constraint
can be written as $L(\unkVec) > 0$, where $L$ is the integer linear form
occurring on its left-hand side.
On integer vectors, this constraint is equivalent to $L(\unkVec) \geq 1$.
When applying bounds on integer solutions, we therefore understand the facet
complexity of such a system as that of the resulting non-strict rational system.

\begin{lemma}
\label{lem:syssol}
  Let $\PStr = \tuple {\unkVec} {\preccurlyeq}$ be a finite poset and
  $\MFun(\unkVec; \vunkElm) > \PFun(\unkVec; \vunkElm)$ an $n$-exppar $1$-\MPI.
  Then, the $\tuple* {\MFun(\unkVec; \vunkElm) \!>\! \PFun(\unkVec; \vunkElm)}
  {\PStr} [+]$-system admits a solution \iff it admits a $(6 \phi
  n^{3})$-bit-bounded Diophantine positive solution, where $\phi$ is its facet
  complexity.
\end{lemma}
\begin{proof}
  We only focus on one of the two directions of the equivalence, the other being
  trivial.
  \begin{itemize}
  \item
    \textbf{[$\Rightarrow$]}
    Suppose the system of homogeneous linear inequalities admits a solution.
    Since the coefficients are all rational, there necessarily exists a rational
    solution $\rsolVec = \{ \rsolElm[\unkElm] \}_{\unkElm \in \unkVec} \in
    \SetQ[+][n]$ for it~\cite{NW88}.
    Recall that, for every $\unkElm \in \unkVec$, there are two natural numbers
    $\pElm[\unkElm], \qElm[\unkElm] \in \SetN[+]$ with $\gcd(\pElm[\unkElm],
    \qElm[\unkElm]) = 1$, such that $\rsolElm[\unkElm] = \pElm[\unkElm] /
    \qElm[\unkElm]$.
    Now, let $\dsolElm \defeq \lcm[\unkElm \in \unkVec]\, \qElm[\unkElm] \geq
    1$.
    Thanks to the homogeneity of the system, the natural vector $\dsolVec = \{
    \dsolElm \cdot \rsolElm[\unkElm] \}_{\unkElm \in \unkVec} \in \SetN[+][n]$
    is clearly a Diophantine positive solution as well.
    Moreover, since $\dsolVec$ is integral and all coefficients of the system
    are integers, the left-hand side of every strict constraint takes an integer
    value on $\dsolVec$.
    Hence, every such value is at least $1$, and we actually have:
    \begin{itemize}
    \item
      $(\expVec[0] - \expVec[i])^{\intercal}\! \cdot \dsolVec \geq 1$, for all
      $i \in \num{m}$ with $\alphaElm[i] > 0$;
    \item
      $\dsolElm[{\unkElm[2]}] - \dsolElm[{\unkElm[1]}] \geq 1$, for all
      $\unkElm[1], \unkElm[2] \in \unkVec$ with $\unkElm[1] \vartriangleleft
      \unkElm[2]$;
    \item
      $\dsolElm[\unkElm] \geq 1$, for all $\unkElm \in \min[\preccurlyeq]
      \unkVec$.
    \end{itemize}
    At this point, by applying a known result about the size of integer
    solutions of a linear system~\cite[Corollary~17.1b]{Sch86}, one obtains the
    existence of a Diophantine positive solution of bit-size bounded by $6 \phi
    n^{3}$, where $\phi$ is the facet complexity of the system.
    \qedhere
  \end{itemize}
\end{proof}

The preceding lemma lets us choose the exponent parameters to be positive
integers whenever the positive system is feasible.
Evaluating the parameterised inequality at such a vector produces a $1$-\MPI
whose left-hand degree is strictly greater than its positive degree.
By \cref{lem:1gmpidegsol}, every sufficiently large value of $\vunkElm$ then
satisfies the resulting inequality.
This gives the converse bridge from the linear system to the parameterised
problem.

\begin{lemma}
\label{lem:nexppar1mpisyssol}
  Let $\PStr = \tuple {\unkVec} {\preccurlyeq}$ be a finite poset and
  $\MFun(\unkVec; \vunkElm) > \PFun(\unkVec; \vunkElm)$ an $n$-exppar $1$-\MPI
  for which the $\tuple* {\MFun(\unkVec; \vunkElm) \!>\! \PFun(\unkVec;
  \vunkElm)} {\PStr} [+]$-system admits a solution.
  Then, there exist an increasing $n$-vector $\dsolVec \in \SetN[+][n]$ and a
  value $\ell \in \SetR[\geq 0]$ such that, for every value $\zetasolElm \in
  \SetR[> \ell]$, the pair $(\dsolVec, \zetasolElm)$ is a solution of the
  $n$-exppar $1$-\MPI.
\end{lemma}
\begin{proof}
  Let $\MFun(\unkVec; \vunkElm) = \vunkElm^{\expVec[0]^{\intercal} \cdot
  \unkVec}$ and $\PFun(\unkVec; \vunkElm) = \sum_{i = 1}^{m} \alphaElm[i]
  \vunkElm^{\expVec[i]^{\intercal} \cdot \unkVec}$, for some $\{ \alphaElm[i]
  \}_{i = 1}^{m} \subseteq \SetR[\neq 0]$ and $\{ \expVec[i] \}_{i = 0}^{m}
  \subseteq \SetN[][n]$, with $m \in \SetN$.
  Since the $\tuple* {\MFun(\unkVec; \vunkElm) \!>\! \PFun(\unkVec; \vunkElm)}
  {\PStr} [+]$-system of homogeneous linear inequalities admits a solution, by
  \cref{lem:syssol}, there exists a Diophantine positive solution $\dsolVec \in
  \SetN[+][n]$ for it.
  Since $\dsolVec$ satisfies the successor inequalities of the system, it is
  clearly increasing.
  Now, similarly to the proof of \cref{lem:nexppar1mpisolsys}, consider the
  $1$-\MPI $\MFun[][\ast](\vunkElm) > \PFun[][\ast](\vunkElm)$, where
  $\MFun[][\ast](\vunkElm) \defeq \MFun(\dsolVec; \vunkElm) =
  \vunkElm^{(\expVec[0]^{\intercal} \cdot \dsolVec)}$ and
  $\PFun[][\ast](\vunkElm) \defeq \PFun(\dsolVec; \vunkElm) = \sum_{i = 1}^{m}
  \alphaElm[i] \vunkElm^{(\expVec[i]^{\intercal} \cdot \dsolVec)}$.
  By construction, $\deg{\MFun[][\ast](\vunkElm)} >
  \deg[+]{\PFun[][\ast](\vunkElm)}$, since $(\expVec[0]^{\intercal} \cdot
  \dsolVec) - (\expVec[i]^{\intercal} \cdot \dsolVec) = (\expVec[0] -
  \expVec[i])^{\intercal} \cdot \dsolVec > 0$, for all $i \in \num{m}$ with
  $\alphaElm[i] > 0$.
  Thus, by \cref{lem:1gmpidegsol} applied to $\MFun[][\ast](\vunkElm) >
  \PFun[][\ast](\vunkElm)$, there exists a value $\ell \in \SetR[\geq 0]$ such
  that every value $\zetasolElm \in \SetR[> \ell]$ is a solution of this
  $1$-\MPI, which implies that $\MFun(\dsolVec; \zetasolElm) =
  \MFun[][\ast](\zetasolElm) > \PFun[][\ast](\zetasolElm) = \PFun(\dsolVec;
  \zetasolElm)$, as required by the statement of the lemma.
\end{proof}




\subsection{Increasing Solution of $n$-{\MPI}s}
\label{sec:dphprb;sub:incsolnmpi}

The previous subsection characterises the exponent parameters that make a
one-variable inequality eventually true.
We now use this characterisation to return to the original $n$-dimensional
inequality.
The key observation is that a positive vector can be represented, once one of
its coordinates is greater than $1$, through powers of a common base.
Under this representation, an increasing vector gives rise to increasing
exponent parameters, while the original $n$-\GMPI becomes the corresponding
$n$-exppar $1$-\GMPI.
Conversely, an integral solution of the positive linear system provides integer
exponents from which arbitrarily large increasing solutions of the original
$n$-\MPI can be constructed.

We first isolate the elementary fact ensuring the existence of a suitable common
base whenever the polynomial is non-zero.

\begin{lemma}
\label{lem:nmpisol}
  Let $\PStr = \tuple {\unkVec} {\preccurlyeq}$ be a finite poset and
  $\MFun(\unkVec) > \PFun(\unkVec)$ an $n$-\GMPI admitting a solution $\solVec =
  \{ \solElm[\unkElm] \}_{\unkElm \in \unkVec} \in \SetR[\geq 1][n]$, with a
  non-zero $\PFun(\unkVec)$ having coefficients in $\SetR[\geq 1]$.
  Then, for at least one $\unkElm \in \unkVec$, it holds that $\solElm[\unkElm]
  > 1$.
\end{lemma}
\begin{proof}
  Let $\MFun(\unkVec) = \unkVec^{\expVec[0]}$ and $\PFun(\unkVec) = \sum_{i =
  1}^{m} \alphaElm[i] \unkVec^{\expVec[i]}$, for some $\{ \alphaElm[i] \}_{i =
  1}^{m} \subseteq \SetR[\geq 1]$ and $\{ \expVec[i] \}_{i = 0}^{m} \subseteq
  \SetR[\geq 0][n]$ with $m \in \SetN[+]$.
  Suppose, by way of contradiction, that $\solElm[\unkElm] \leq 1$, for all
  $\unkElm \in \unkVec$.
  Since $\solVec \in \SetR[\geq 1][n]$, this means that $\solVec = \vec*{1}$.
  Obviously, $\MFun(\vec*{1}) = 1$.
  Moreover, $\PFun(\vec*{1}) = \sum_{i = 1}^{m} \alphaElm[i]
  \vec*{1}^{\expVec[i]} = \sum_{i = 1}^{m} \alphaElm[i] \geq m \geq 1$, since
  $\PFun(\unkVec)$ has coefficients in $\SetR[\geq 1]$.
  Hence, $1 = \MFun(\solVec) > \PFun(\solVec) \geq 1$, which is clearly
  impossible.
\end{proof}

The previous lemma allows any coordinate of the solution greater than $1$ to be
chosen as a common logarithmic base.
Every other coordinate can then be expressed as a non-negative real power of
this base.
Moreover, taking logarithms with base greater than $1$ preserves the strict
order among the coordinates.
The zero-polynomial case does not require this representation and will be
handled separately.
These observations allow an increasing solution of the original $n$-\GMPI to be
transferred to its exponent-parameterised counterpart.

\begin{lemma}
\label{lem:nmpiinccohsolsys}
  Let $\PStr = \tuple {\unkVec} {\preccurlyeq}$ be a finite poset and
  $\MFun(\unkVec) > \PFun(\unkVec)$ an $n$-\GMPI admitting an increasing
  solution in $\SetR[\geq 1][n]$, with $\PFun(\unkVec)$ having coefficients in
  $\SetR[\geq 1]$.
  Then, the $\tuple* {\MFun(\unkVec) \!>\! \PFun(\unkVec)} {\PStr} [+]$-system
  admits a solution.
\end{lemma}
\begin{proof}
  Let $\MFun(\unkVec) = \unkVec^{\expVec[0]}$ and $\PFun(\unkVec) = \sum_{i =
  1}^{m} \alphaElm[i] \unkVec^{\expVec[i]}$, for some $\{ \alphaElm[i] \}_{i =
  1}^{m} \subseteq \SetR[\geq 1]$ and $\{ \expVec[i] \}_{i = 0}^{m} \subseteq
  \SetR[\geq 0][n]$ with $m \in \SetN$.
  Moreover, consider the $n$-exppar $1$-\GMPI $\MFun[][\ast](\unkVec; \vunkElm)
  > \PFun[][\ast](\unkVec; \vunkElm)$, where $\MFun[][\ast](\unkVec; \vunkElm)
  \defeq \allowbreak \vunkElm^{(\expVec[0]^{\intercal} \cdot \unkVec)}$ and
  $\PFun[][\ast](\unkVec; \vunkElm) \defeq \sum_{i = 1}^{m} \alphaElm[i]
  \vunkElm^{(\expVec[i]^{\intercal} \cdot \unkVec)}$.
  Two cases arise.
  If $\PFun(\unkVec)$ is the zero polynomial then $\PFun[][\ast](\unkVec;
  \vunkElm)$ is the zero $n$-exppar polynomial.
  Since $\PStr$ is finite, there exists an increasing $n$-vector $\rsolVec \in
  \SetR[\geq 0][n]$.
  Choose any $\zetasolElm \in \SetR[\geq 1]$.
  Then, $(\rsolVec, \zetasolElm)$ is a solution of the $n$-exppar $1$-\GMPI.
  If $\PFun(\unkVec)$ is a non-zero polynomial, instead, let $\solVec = \{
  \solElm[\unkElm] \}_{\unkElm \in \unkVec} \in \SetR[\geq 1][n]$ be an
  increasing solution of the $n$-\GMPI.
  By \cref{lem:nmpisol}, there exists $\unkElm[][*] \in \unkVec$ such that
  $\solElm[{\unkElm[][*]}] > 1$.
  Let $\zetasolElm \defeq \solElm[{\unkElm[][*]}]$.
  In addition, let $\rsolVec = \{ \rsolElm[\unkElm] \}_{\unkElm \in \unkVec} \in
  \SetR[\geq 0][n]$ be the $n$-vector of non-negative real numbers defined as
  $\rsolElm[\unkElm] \defeq \log_{\zetasolElm}(\solElm[\unkElm])$, for all
  $\unkElm \in \unkVec$.
  Clearly, $\solElm[\unkElm] = \zetasolElm^{\rsolElm[\unkElm]}$.
  Moreover, $\rsolVec$ is increasing, since $\solVec$ is increasing.
  Now, one can easily see that $(\rsolVec, \zetasolElm) \in \SetR[\geq 0][n]
  \times \SetR[\geq 1]$ is a solution of the $n$-exppar $1$-\GMPI, since
  $\MFun[][\ast](\rsolVec; \zetasolElm) = \MFun(\solVec) > \PFun(\solVec) =
  \PFun[][\ast](\rsolVec; \zetasolElm)$.
  Indeed,
  \[
    {\MFun[][\ast](\rsolVec; \zetasolElm)}
  =
    {\zetasolElm}^{\left( \expVec[0]^{\intercal} \cdot \rsolVec \right)}
  =
    {\zetasolElm}^{\sum_{\unkElm \in \unkVec} \expElm[0\unkElm]
    \rsolElm[\unkElm]}
  =
    \prod_{\unkElm \in \unkVec} {\zetasolElm}^{\expElm[0\unkElm]
    \rsolElm[\unkElm]}
  =
    \prod_{\unkElm \in \unkVec} \left( {\zetasolElm}^{\rsolElm[\unkElm]}
    \right)^{\expElm[0\unkElm]}
  =
    \prod_{\unkElm \in \unkVec} {\solElm[\unkElm]}^{\expElm[0\unkElm]}
  =
    {\solVec}^{\expVec[0]}
  =
    {\MFun(\solVec)} \,.
  \]
  The same reasoning applies to all monomials of the polynomial $\PFun(\unkVec)$
  and the $n$-exppar polynomial $\PFun[][\ast](\unkVec; \vunkElm)$, thus,
  $\PFun(\solVec) = \PFun[][\ast](\rsolVec; \zetasolElm)$.
  Summing up, in both cases, there exists a solution $(\rsolVec, \zetasolElm)
  \in \SetR[\geq 0][n] \times \SetR[\geq 1]$ of the $n$-exppar $1$-\GMPI, with
  $\rsolVec$ increasing and $\PFun[][\ast](\unkVec; \vunkElm)$ having
  coefficients in $\SetR[\geq 1]$.
  Therefore, the thesis immediately follows by applying
  \cref{lem:nexppar1mpisolsys} to $\MFun[][\ast](\unkVec; \vunkElm) >
  \PFun[][\ast](\unkVec; \vunkElm)$.
\end{proof}

We now prove the converse direction.
In this case, no sign restriction on the coefficients is required.
A solution of the positive system provides positive integral exponent parameters
and, by the results of the previous subsection, every sufficiently large value
of the common base satisfies the associated $n$-exppar $1$-\MPI.
Substituting the corresponding powers of that base for the original unknowns
then yields an increasing Diophantine vector.
The construction is stable under raising all its coordinates to the same
positive power, a property that will later allow these solutions to be made
arbitrarily large.

\begin{lemma}
\label{lem:nmpisysincsol}
  Let $\PStr = \tuple {\unkVec} {\preccurlyeq}$ be a finite poset and
  $\MFun(\unkVec) > \PFun(\unkVec)$ an $n$-\MPI for which the $\tuple*
  {\MFun(\unkVec) \!>\! \PFun(\unkVec)} {\PStr} [+]$-system admits a solution.
  Then, there exists an increasing $n$-vector $\solVec \in \SetN[> 1][n]$ such
  that $\solVec^{\rhoElm}$ is an increasing solution of the $n$-\MPI, for all
  values $\rhoElm \in \SetR[\geq 1]$.
\end{lemma}
\begin{proof}
  Let $\MFun(\unkVec) = \unkVec^{\expVec[0]}$ and $\PFun(\unkVec) = \sum_{i =
  1}^{m} \alphaElm[i] \unkVec^{\expVec[i]}$, for some $\{ \alphaElm[i] \}_{i =
  1}^{m} \subseteq \SetR[\neq 0]$ and $\{ \expVec[i] \}_{i = 0}^{m} \subseteq
  \SetN[][n]$, with $m \in \SetN$.
  Moreover, similarly to the proof of \cref{lem:nmpiinccohsolsys}, consider the
  $n$-exppar $1$-\MPI $\MFun[][\ast](\unkVec; \vunkElm) > \PFun[][\ast](\unkVec;
  \vunkElm)$, where $\MFun[][\ast](\unkVec; \vunkElm) \defeq
  \vunkElm^{(\expVec[0]^{\intercal} \cdot \unkVec)}$ and $\PFun[][\ast](\unkVec;
  \vunkElm) \defeq \sum_{i = 1}^{m} \alphaElm[i]
  \vunkElm^{(\expVec[i]^{\intercal} \cdot \unkVec)}$.
  Since the two systems $\tuple* {\MFun(\unkVec) > \PFun(\unkVec)} {\PStr} [+]$
  and $\tuple* {\MFun[][\ast](\unkVec; \vunkElm) > \PFun[][\ast](\unkVec;
  \vunkElm)} {\PStr} [+]$ of homogeneous linear inequalities coincide, by
  \cref{lem:nexppar1mpisyssol}, there exist an increasing $n$-vector $\dsolVec =
  \{ \dsolElm[\unkElm] \}_{\unkElm \in \unkVec} \in \SetN[+][n]$ and a value
  $\ell \in \SetR[\geq 0]$ such that, for every value $\zetasolElm \in \SetR[>
  \ell]$, the pair $(\dsolVec, \zetasolElm)$ is a solution of the $n$-exppar
  $1$-\MPI, \ie, $ \MFun[][\ast](\dsolVec; \zetasolElm) >
  \PFun[][\ast](\dsolVec; \zetasolElm)$.
  Now, let $\zetasolElm[][\ast] \defeq \floor{2 + \ell} > 1 + \ell \geq 1$
  and $\solVec = \{ \solElm[\unkElm]  \}_{\unkElm \in \unkVec} \in \SetN[+][n]$
  be the $n$-vector of natural numbers, whose components are defined as
  $\solElm[\unkElm] \defeq {\zetasolElm[][\ast]}^{\dsolElm[\unkElm]}$, for all
  $\unkElm \in \unkVec$.
  Clearly, $\solVec$ is increasing, since $\dsolVec$ is increasing and
  $\zetasolElm[][\ast] > 1$.
  The same holds true for all $\solVec^{\rhoElm}$, given a value $\rhoElm \in
  \SetR[\geq 1]$.
  Moreover, $\solVec$ belongs to $\SetN[> 1][n]$, since both $\dsolVec \in
  \SetN[+][n]$ and $\zetasolElm[][\ast] \in \SetN[> 1]$.
  At this point, due to the fact that ${\zetasolElm[][\ast]}^{\rhoElm} > (1 +
  \ell)^{\rhoElm} \geq 1 + \ell > \ell$,
  it is immediate to see that $\solVec^{\rhoElm}$ is also a solution of the
  $n$-\MPI, since $\MFun(\solVec^{\rhoElm}) = \MFun[][\ast](\dsolVec;
  {\zetasolElm[][\ast]}^{\rhoElm}) > \PFun[][\ast](\dsolVec;
  {\zetasolElm[][\ast]}^{\rhoElm}) = \PFun(\solVec^{\rhoElm})$.
  Indeed,
  \begin{align*}
    {\MFun(\solVec^{\rhoElm})}
  & =
    {\left( \solVec^{\rhoElm} \right)^{\expVec[0]}}
  =
    {\prod_{\unkElm \in \unkVec} \left( \solElm[\unkElm]^{\rhoElm}
    \right)^{\expElm[0\unkElm]}}
  =
    \prod_{\unkElm \in \unkVec} \left( \left(
    {\zetasolElm[][\ast]}^{\dsolElm[\unkElm]} \right)^{\rhoElm}
    \right)^{\expElm[0\unkElm]}
  =
    \prod_{\unkElm \in \unkVec} \left( \left( {\zetasolElm[][\ast]}^{\rhoElm}
    \right)^{\dsolElm[\unkElm]} \right)^{\expElm[0\unkElm]}
  =
    \prod_{\unkElm \in \unkVec} \left( {\zetasolElm[][\ast]}^{\rhoElm}
    \right)^{\expElm[0\unkElm] \dsolElm[\unkElm]}
  = \\
  & =
    \left( {\zetasolElm[][\ast]}^{\rhoElm} \right)^{\sum_{\unkElm \in \unkVec}
    \expElm[0\unkElm] \dsolElm[\unkElm]}
  =
    \left( {\zetasolElm[][\ast]}^{\rhoElm} \right)^{(\expVec[0]^{\intercal}
    \cdot \dsolVec)}
  =
    {\MFun[][\ast](\dsolVec; {\zetasolElm[][\ast]}^{\rhoElm})} \,.
  \end{align*}
  The same reasoning applies to all monomials of the polynomial $\PFun(\unkVec)$
  and the $n$-exppar polynomial $\PFun[][\ast](\unkVec; \vunkElm)$, thus,
  $\PFun[][\ast](\dsolVec; {\zetasolElm[][\ast]}^{\rhoElm}) =
  \PFun(\solVec^{\rhoElm})$.
  Summing up, $\solVec^{\rhoElm}$ is an increasing solution of the $n$-\MPI, for
  all values $\rhoElm \in \SetR[\geq 1]$, as prescribed by the statement of the
  lemma.
\end{proof}

The preceding lemmata can now be combined into a complete characterisation.
For polynomials with positive coefficients, the existence of an increasing real
solution, the feasibility of the positive linear system, and the existence of
arbitrarily large Diophantine increasing solutions are equivalent.

\begin{theorem}[Increasing Solution]
\label{thm:incsol}
  For all finite posets $\PStr = \tuple {\unkVec} {\preccurlyeq}$ and
  $n$-{\MPI}s $\MFun(\unkVec) > \PFun(\unkVec)$, with $\PFun(\unkVec)$ having
  coefficients in $\SetR[\geq 1]$, the following statements are equivalent:
  \begin{enumerate}[a)]
  \item\label{thm:incsol(sol)}
    the $n$-\MPI admits an increasing solution in $\SetR[\geq 1][n]$;
  \item\label{thm:incsol(sys)}
    the $\tuple* {\MFun(\unkVec) \!>\! \PFun(\unkVec)} {\PStr} [+]$-system
    admits a solution;
  \item\label{thm:incsol(dphsol)}
    for all $\kappaVec \in \SetN[+][n]$, the $n$-\MPI admits a Diophantine
    increasing solution $\solVec[\kappaVec] \geq \kappaVec$.
  \end{enumerate}
\end{theorem}
\begin{proof}
  We prove the equivalence by focusing on the only non-immediate implication,
  namely \ref{thm:incsol(sys)} $\Rightarrow$ \ref{thm:incsol(dphsol)}, since
  \ref{thm:incsol(sol)} $\Rightarrow$ \ref{thm:incsol(sys)} clearly follows from
  \cref{lem:nmpiinccohsolsys}, while \ref{thm:incsol(dphsol)} $\Rightarrow$
  \ref{thm:incsol(sol)} is a trivial consequence of considering $\kappaVec =
  \vec*{1}$.
  \begin{itemize}
  \item
    \textbf{[\ref{thm:incsol(sys)} $\Rightarrow$ \ref{thm:incsol(dphsol)}]}
    Assume that the $\tuple* {\MFun(\unkVec) \!>\! \PFun(\unkVec)} {\PStr}
    [+]$-system admits a solution.
    By \cref{lem:nmpisysincsol}, there exists an increasing $n$-vector of values
    $\solVec \in \SetN[> 1][n]$ such that $\solVec^{\rhoElm}$ is a Diophantine
    increasing solution of the $n$-\MPI, for every $\rhoElm \in \SetN[+]$.
    Given an arbitrary $\kappaVec \in \SetN[+][n]$, let $\rhoElm[\kappaVec]
    \defeq \floor{1 + \max_{\unkElm \in \unkVec} \log_{\solElm[\unkElm]}
    \kappaElm[\unkElm]} \in \SetN[+]$ and $\solVec[\kappaVec] \defeq
    \solVec^{\rhoElm[\kappaVec]} \in \SetN[> 1][n]$.
    Since $\rhoElm[\kappaVec] \in \SetN[+]$, the vector $\solVec[\kappaVec]$
    is a Diophantine increasing solution of the $n$-\MPI.
    By construction, $\rhoElm[\kappaVec] > \log_{\solElm[\unkElm]}
    \kappaElm[\unkElm]$, for every $\unkElm \in \unkVec$.
    Hence, $\solElm[\unkElm]^{\rhoElm[\kappaVec]} \geq \kappaElm[\unkElm]$, for
    every $\unkElm \in \unkVec$, and therefore $\solVec[\kappaVec] \geq
    \kappaVec$.
    \qedhere
  \end{itemize}
\end{proof}




\subsection{General Solutions of {$n$-\MPI}s}
\label{sec:dphprb;sub:gensolnmpi}

The previous subsection characterises increasing solutions of an $n$-\MPI in
terms of the associated homogeneous linear system.
We now reintroduce the weighted convolution and offset, which determine whether
a valuation is admissible or natural.
The main difficulty is that these arithmetic constraints are not encoded in the
linear system itself.
We first show that every admissible solution bounded below by $1$ can be
approached by admissible increasing solutions without losing the inequality.
We then use this approximation to transfer the monomial domination guaranteed by
strong non-negativeness at Diophantine natural solutions to increasing ones,
before recovering naturality through the convolution structure.
Unlike the corresponding construction in~\cite{KM25}, this argument works
directly over an arbitrary finite poset and does not require passing to a
meet-semilattice restriction.

\begin{lemma}
\label{lem:soleqv}
  Let $\PStr = \tuple {\unkVec} {\preccurlyeq} {\wFun} {\oVec}$ be a stratified
  poset and $\unkVec[\top] \subseteq \unkVec$ an upward-closed $n$-vector of
  unknowns.
  An $n$-\MPI $\MFun(\unkVec[\top]) > \PFun(\unkVec[\top])$ has an admissible
  solution in $\SetR[\geq 1][n]$ \iff it has an admissible increasing solution
  in $\SetR[\geq 1][n]$.
  Moreover, for every admissible solution $\solVec \in \SetR[\geq 1][n]$, there
  exist $\varepsilon^{*} \in \SetR[+]$ and a family $\{ \solVec[\varepsilon]
  \}_{0 < \varepsilon \leq \varepsilon^{*}}$ of admissible increasing solutions
  of the $n$-\MPI in $\SetR[\geq 1][n]$ such that $\lim_{\varepsilon \to 0^{+}}
  \solVec[\varepsilon] = \solVec$.
\end{lemma}
\begin{proof}
  We only focus on one of the two directions of the equivalence, the other being
  trivial.
  \begin{itemize}
  \item
    \textbf{[$\Rightarrow$]}
    Assume that the $n$-\MPI has an admissible solution $\solVec \in \SetR[\geq
    1][n]$ and let $\solVec[\oVec] \in \SetR[][n + k]$ be its $\oVec$-extension
    as per \cref{def:admnat}, with $n + k = \card{\unkVec}$.
    In addition, consider $\etaVec \defeq \wFun^{-1} \star (\solVec[\oVec] -
    \oVec) \in \SetR[][n + k]$.
    Since $\unkVec[\top]$ is upward closed, its complement is downward closed.
    Hence, $\etaVec(\unkElm) = 0$, for every $\unkElm \in \unkVec \setminus
    \unkVec[\top]$.
    Moreover, $\etaVec(\unkElm) \geq 0$, for every $\unkElm \in \unkVec[\top]$,
    due to the admissibility hypothesis.
    Thus, $\etaVec \geq \vec*{0}$.
    Recall also that $0 < \wFun(\unkElm, \unkElm[1]) \leq \wFun(\unkElm,
    \unkElm[2])$, for every $\unkElm \preccurlyeq \unkElm[1] \preccurlyeq
    \unkElm[2]$, and $(\wFun^{-1} \star \oVec)(\unkElm) \geq 0$, for all
    $\unkElm \in \unkVec$, due to the stratification of $\PStr$.
    Hence, by the properties of the convolution action, both $\solVec[\oVec] -
    \oVec = \wFun \star \etaVec$ and $\oVec = \wFun \star (\wFun^{-1} \star
    \oVec)$ are non-decreasing.
    Therefore, $\solVec[\oVec]$ is non-decreasing and, in particular, $\solVec$
    is non-decreasing on $\unkVec[\top]$.
    Now, let $\vec*{1} \in \SetR[\geq 1][n]$ be the constant-one vector on
    $\unkVec[\top]$, $\vec*{1}\vec*{0} \in \SetR[][n + k]$ the vector on
    $\unkVec$ with value $1$ only on $\unkVec[\top]$, and define
    $\omegaVec[\top] \defeq \wFun[\top] \star \vec*{1}$ and $\omegaVec \defeq
    \wFun \star \vec*{1}\vec*{0}$, where $\wFun[\top]$ is the restriction of
    $\wFun$ to $\unkVec[\top]$.
    Since $\unkVec[\top]$ is upward closed, $\omegaVec$ is $0$ on $\unkVec
    \setminus \unkVec[\top]$ and coincides with $\omegaVec[\top]$ on
    $\unkVec[\top]$.
    Obviously, $\wFun[\top]^{-1} \star \omegaVec[\top] = \vec*{1}$ and
    $\wFun^{-1} \star \omegaVec = \vec*{1}\vec*{0}$.
    Moreover, $\omegaVec[\top]$ is increasing on $\unkVec[\top]$.
    At this point, for every $\varepsilon \in \SetR[+]$, let
    $\solVec[\varepsilon] \defeq \solVec + \varepsilon \omegaVec[\top] \in
    \SetR[\geq 1][n]$ and consider its $\oVec$-extension $\solVec[\varepsilon,
    \oVec] \in \SetR[][n + k]$.
    Clearly, $\solVec[\varepsilon, \oVec] = \solVec[\oVec] + \varepsilon
    \omegaVec \in \SetR[][n + k]$.
    Then,
    \begin{alignat*}{3}
      {\wFun^{-1} \star \left( \solVec[\varepsilon, \oVec] - \oVec \right)}
    & =
      {\wFun^{-1} \star \left( \solVec[\oVec] + \varepsilon \omegaVec - \oVec
      \right)}
    && = \\
    & =
      {\wFun^{-1} \star \left( \solVec[\oVec] - \oVec \right) + \wFun^{-1} \star
      \left( \varepsilon \omegaVec \right)}
    && = \\
    & =
      {\etaVec + \varepsilon \left( \wFun^{-1} \star \omegaVec \right)}
    && =
      {\etaVec + \varepsilon \vec*{1}\vec*{0}}
    \geq
      {\varepsilon \vec*{1}\vec*{0}}
    \geq
      {\vec*{0}} \,.
    \end{alignat*}
    Thus, $\solVec[\varepsilon]$ is admissible.
    Moreover, $\etaVec + \varepsilon \vec*{1}\vec*{0}$ is strictly positive on
    $\unkVec[\top]$.
    Hence, $\solVec[\varepsilon, \oVec] - \oVec =
    \wFun \star (\etaVec + \varepsilon \vec*{1}\vec*{0})$ is increasing on
    $\unkVec[\top]$.
    Since $\oVec$ is non-decreasing, $\solVec[\varepsilon, \oVec]$ is increasing
    on $\unkVec[\top]$, and so $\solVec[\varepsilon]$ is increasing too.
    Finally, let $\QFun(\unkVec[\top]) \defeq \MFun(\unkVec[\top]) -
    \PFun(\unkVec[\top])$.
    Since $\solVec$ is a solution of the $n$-\MPI, it holds that $\QFun(\solVec)
    > 0$.
    By continuity of $\QFun$ and the fact that $\solVec[\varepsilon] \to
    \solVec$ as $\varepsilon \to 0^{+}$, there exists $\varepsilon^{*} \in
    \SetR[+]$ such that $\QFun(\solVec[\varepsilon]) > 0$, for every $0 <
    \varepsilon \leq \varepsilon^{*}$.
    Therefore, the family $\{ \solVec[\varepsilon] \}_{0 < \varepsilon \leq
    \varepsilon^{*}}$ consists of admissible increasing solutions of the
    $n$-\MPI in $\SetR[\geq 1][n]$ such that $\lim_{\varepsilon \to 0^{+}}
    \solVec[\varepsilon] = \solVec$.
    \qedhere
  \end{itemize}
\end{proof}

The previous lemma removes the first obstruction: every admissible solution can
be approached by admissible increasing solutions while preserving the
inequality.
To apply the characterisation of the previous subsection, however, we still need
to control the possible negative coefficients of $\PFun$.
Strong non-negativeness provides precisely this control at Diophantine natural
solutions.
At any such solution of the original inequality, every positive unitary monomial
of $\PFun$ is bounded above by $\PFun$ and, therefore, strictly dominated by
$\MFun$.
By continuity, the finitely many resulting strict inequalities are preserved at
all sufficiently close admissible increasing solutions.
Raising all coordinates of one such solution to a common sufficiently large
power amplifies this strict domination.
At the same time, the weighted inversion shows that admissibility is eventually
preserved under the same operation.
This yields the following reduction to an arbitrarily scaled positive part of
the polynomial.

\begin{lemma}
\label{lem:solred}
  Let $\PStr = \tuple {\unkVec} {\preccurlyeq} {\wFun} {\oVec}$ be a stratified
  poset, $\unkVec[\top] \subseteq \unkVec$ an upward-closed $n$-vector of
  unknowns, and $\MFun(\unkVec[\top]) > \PFun(\unkVec[\top])$ an $n$-\MPI, with
  $\PFun(\unkVec[\top])$ strongly non-negative \wrt $\PStr$.
  If $\MFun(\unkVec[\top]) > \PFun(\unkVec[\top])$ has a Diophantine natural
  solution in $\SetN[+][n]$, then $\MFun(\unkVec[\top]) > \alphaElm
  \PFun[+](\unkVec[\top])$ has an admissible increasing solution in $\SetR[\geq
  1][n]$, for every $\alphaElm \in \SetR[\geq 0]$.
\end{lemma}
\begin{proof}
  Let $\MFun(\unkVec[\top]) = \unkVec[\top]^{\expVec[0]}$ and
  $\PFun(\unkVec[\top]) = \sum_{i = 1}^{m} \alphaElm[i]
  \unkVec[\top]^{\expVec[i]}$, for some $\{ \alphaElm[i] \}_{i = 1}^{m}
  \subseteq \SetR[\neq 0]$ and $\{ \expVec[i] \}_{i = 0}^{m} \subseteq
  \SetN[][n]$, with $m \in \SetN$, and recall that $\PFun[+](\unkVec[\top]) =
  \sum_{1 \leq i \leq m}^{\alphaElm[i] > 0} \alphaElm[i]
  \unkVec[\top]^{\expVec[i]}$.
  Assume that the $n$-\MPI $\MFun(\unkVec[\top]) > \PFun(\unkVec[\top])$ has a
  Diophantine natural solution $\dsolVec \in \SetN[+][n]$.
  Since $\PFun(\unkVec[\top])$ is strongly non-negative \wrt $\PStr$, for every
  $i \in \num{m}$ with $\alphaElm[i] > 0$, it holds that $\dsolVec^{\expVec[0]}
  = \MFun(\dsolVec) > \PFun(\dsolVec) \geq \dsolVec^{\expVec[i]}$.
  Since $\dsolVec$ is admissible, by~\cref{lem:soleqv}, there exist
  $\varepsilon^{*} \in \SetR[+]$ and a family $\{ \solVec[\varepsilon] \}_{0 <
  \varepsilon \leq \varepsilon^{*}}$ of admissible increasing solutions of the
  $n$-\MPI in $\SetR[\geq 1][n]$ such that $\lim_{\varepsilon \to 0^{+}}
  \solVec[\varepsilon] = \dsolVec$.
  For every $i \in \num{m}$ with $\alphaElm[i] > 0$, consider the polynomial
  \[
    \QFun[i](\unkVec[\top])
  \defeq
    \unkVec[\top]^{\expVec[0]} - \unkVec[\top]^{\expVec[i]}.
  \]
  By the inequality above, $\QFun[i](\dsolVec) > 0$.
  Hence, by continuity and since there are only finitely many such indices,
  there exists $\varepsilon^{\dagger} \in \SetR[+]$, with $\varepsilon^{\dagger}
  \leq \varepsilon^{*}$, such that $\QFun[i](\solVec[\varepsilon]) > 0$, for
  every $i \in \num{m}$ with $\alphaElm[i] > 0$ and $0 < \varepsilon \leq
  \varepsilon^{\dagger}$.
  Fix any such $\varepsilon$ and put $\solVec \defeq \solVec[\varepsilon]$.
  Then $\solVec$ is an admissible increasing solution of the original
  $n$-\MPI and, for every $i \in \num{m}$ with $\alphaElm[i] > 0$, it holds
  that $\solVec^{\expVec[0]} > \solVec^{\expVec[i]}$.
  Fix $\alphaElm \in \SetR[\geq 0]$.
  We first show that $\solVec^{\rhoElm}$ is a solution of $\MFun(\unkVec[\top])
  > \alphaElm \PFun[+](\unkVec[\top])$, for all sufficiently large $\rhoElm \in
  \SetR[\geq 1]$.
  By the choice of $\solVec$, for every $i \in \num{m}$ with $\alphaElm[i] > 0$,
  it holds that $0 < \solVec^{\expVec[i]} / \solVec^{\expVec[0]} < 1$.
  Therefore, since only finitely many such indices exist,
  \[
    \alphaElm \sum_{1 \leq i \leq m}^{\alphaElm[i] > 0} \alphaElm[i]
    \left(
      \frac{\solVec^{\expVec[i]}}{\solVec^{\expVec[0]}}
    \right)^{\rhoElm}
  \longrightarrow
    0
  \]
  as $\rhoElm \to \infty$.
  Hence, there exists $\sigmaElm[\alphaElm] \in \SetR[\geq 1]$ such that, for
  every $\rhoElm \in \SetR[\geq {\sigmaElm[\alphaElm]}]$,
  \[
    \frac{\alphaElm \PFun[+](\solVec^{\rhoElm})}{\MFun(\solVec^{\rhoElm})}
  =
    \alphaElm \sum_{1 \leq i \leq m}^{\alphaElm[i] > 0} \alphaElm[i]
    \frac{\left( \solVec^{\rhoElm} \right)^{\expVec[i]}}{\left(
    \solVec^{\rhoElm} \right)^{\expVec[0]}}
  =
    \alphaElm \sum_{1 \leq i \leq m}^{\alphaElm[i] > 0} \alphaElm[i] \left(
    \frac{\solVec^{\expVec[i]}}{\solVec^{\expVec[0]}} \right)^{\rhoElm}
  <
    1 \,.
  \]
  Therefore, for all such $\rhoElm$, it holds that $\MFun(\solVec^{\rhoElm}) >
  \alphaElm \PFun[+](\solVec^{\rhoElm})$.
  Clearly, $\solVec^{\rhoElm}$ is also increasing on $\unkVec[\top]$, for every
  $\rhoElm \in \SetR[\geq 1]$.
  Indeed, if $\unkElm[1], \unkElm[2] \in \unkVec[\top]$ and $\unkElm[1]
  \vartriangleleft \unkElm[2]$, then $1 \leq \solElm[{\unkElm[1]}] <
  \solElm[{\unkElm[2]}]$, which implies that $\solElm[{\unkElm[1]}]^{\rhoElm} <
  \solElm[{\unkElm[2]}]^{\rhoElm}$.
  We now prove that $\solVec^{\rhoElm}$ is admissible, for all sufficiently
  large $\rhoElm \in \SetR[\geq 1]$.
  Let $(\solVec^{\rhoElm})_{\oVec}$ be the $\oVec$-extension of
  $\solVec^{\rhoElm}$ and $\unkElm \in \unkVec[\top]$.
  By definition of the convolution inverse, it holds that
  \begin{alignat*}{3}
    {\left( \wFun^{-1} \star \left( (\solVec^{\rhoElm})_{\oVec} - \oVec \right)
    \right)_{\unkElm}}
  & =
    {\left( \left( \wFun^{-1} \star (\solVec^{\rhoElm})_{\oVec} \right) - \left(
    \wFun^{-1} \star \oVec \right) \right)_{\unkElm}}
  && = \\
  & =
    {\left( \wFun^{-1} \star (\solVec^{\rhoElm})_{\oVec} \right)_{\unkElm} -
    \left( \wFun^{-1} \star \oVec \right)_{\unkElm}}
  && = \\
  & =
    {\frac{\solElm[\unkElm]^{\rhoElm}}{\wFun(\unkElm, \unkElm)} + \sum_{\unkElm'
    \prec \unkElm}^{\unkElm' \in \unkVec[\top]} \wFun^{-1}(\unkElm', \unkElm)
    \solElm[\unkElm']^{\rhoElm} + \sum_{\unkElm' \prec \unkElm}^{\unkElm'
    \not\in \unkVec[\top]} \wFun^{-1}(\unkElm', \unkElm) \oElm[\unkElm'] -
    \left( \wFun^{-1} \star \oVec \right)_{\unkElm}} \,.
  \end{alignat*}
  If $\solElm[\unkElm] > 1$, since $\solVec \in \SetR[\geq 1][n]$ is increasing
  and $\PStr$ is finite, $1 \leq \solElm[\unkElm'] < \solElm[\unkElm]$, for
  every $\unkElm' \prec \unkElm$ with $\unkElm' \in \unkVec[\top]$.
  Thus, the first summand has positive coefficient and strictly larger
  exponential base than every summand in the finite sum $\sum_{\unkElm' \prec
  \unkElm}^{\unkElm' \in \unkVec[\top]} \wFun^{-1}(\unkElm', \unkElm)
  \solElm[\unkElm']^{\rhoElm}$, while the remaining part $\sum_{\unkElm' \prec
  \unkElm}^{\unkElm' \not\in \unkVec[\top]} \wFun^{-1}(\unkElm', \unkElm)
  \oElm[\unkElm'] - (\wFun^{-1} \star \oVec)_{\unkElm}$ is even constant.
  Therefore, there exists $\sigmaElm[\unkElm] \in \SetR[\geq 1]$ such that
  $(\wFun^{-1} \star ((\solVec^{\rhoElm})_{\oVec} - \oVec))(\unkElm) \geq 0$,
  for every $\rhoElm \in \SetR[\geq {\sigmaElm[\unkElm]}]$.
  If $\solElm[\unkElm] = 1$, instead, then $\unkElm$ has no element lower than
  it in $\unkVec[\top]$ \wrt $\prec$, again because $\solVec \in \SetR[\geq
  1][n]$ is increasing and $\PStr$ is finite.
  Thus, the sums over predecessors $\unkElm' \in \unkVec[\top]$ of $\unkElm$
  are empty.
  Hence,
  \begin{alignat*}{4}
    {\left( \wFun^{-1} \star \left( (\solVec^{\rhoElm})_{\oVec} - \oVec \right)
    \right)_{\unkElm}}
  & =
    {\frac{\solElm[\unkElm]^{\rhoElm}}{\wFun(\unkElm, \unkElm)} + \sum_{\unkElm'
    \prec \unkElm}^{\unkElm' \in \unkVec[\top]} \wFun^{-1}(\unkElm', \unkElm)
    \solElm[\unkElm']^{\rhoElm} + \sum_{\unkElm' \prec \unkElm}^{\unkElm'
    \not\in \unkVec[\top]} \wFun^{-1}(\unkElm', \unkElm) \oElm[\unkElm'] -
    \left( \wFun^{-1} \star \oVec \right)_{\unkElm}}
  && = \\
  & =
    {\frac{\solElm[\unkElm]}{\wFun(\unkElm, \unkElm)} + \sum_{\unkElm' \prec
    \unkElm}^{\unkElm' \in \unkVec[\top]} \wFun^{-1}(\unkElm', \unkElm)
    \solElm[\unkElm'] + \sum_{\unkElm' \prec \unkElm}^{\unkElm' \not\in
    \unkVec[\top]} \wFun^{-1}(\unkElm', \unkElm) \oElm[\unkElm'] - \left(
    \wFun^{-1} \star \oVec \right)_{\unkElm}}
  && = \\
  & =
    {\left( \wFun^{-1} \star \solVec[\oVec] \right)_{\unkElm} - \left(
    \wFun^{-1} \star \oVec \right)_{\unkElm}}
  && = \\
  & =
    {\left( \left( \wFun^{-1} \star \solVec[\oVec] \right) - \left( \wFun^{-1}
    \star \oVec \right) \right)_{\unkElm}}
  && = \\
  & =
    {\left( \wFun^{-1} \star \left( \solVec[\oVec] - \oVec \right)
    \right)_{\unkElm}}
  \geq
    0 \,,
  \end{alignat*}
  since $\solVec$ is admissible.
  In this case, set $\sigmaElm[\unkElm] \defeq 1$.
  Since $\unkVec[\top]$ is finite, by taking $\sigmaElm[adm] \defeq
  \max_{\unkElm \in \unkVec[\top]} \sigmaElm[\unkElm]$, one obtains that
  $(\wFun^{-1} \star ((\solVec^{\rhoElm})_{\oVec} - \oVec))_{\unkElm} \geq 0$,
  for every $\unkElm \in \unkVec[\top]$ and $\rhoElm \in \SetR[\geq
  {\sigmaElm[adm]}]$.
  Therefore, $\solVec^{\rhoElm}$ is admissible, for every $\rhoElm \in
  \SetR[\geq {\sigmaElm[adm]}]$.
  At this point, for the fixed $\alphaElm \in \SetR[\geq 0]$ and every $\rhoElm
  \in \SetR[\geq \max {\{ \sigmaElm[adm], \sigmaElm[\alphaElm] \}}]$, the
  $n$-vector $\solVec^{\rhoElm}$ is an admissible increasing solution in
  $\SetR[\geq 1][n]$ of the $n$-\MPI $\MFun(\unkVec[\top]) > \alphaElm
  \PFun[+](\unkVec[\top])$.
\end{proof}

The previous lemma transforms a Diophantine natural positive solution of the
original inequality into an admissible increasing solution involving only
$\PFun[+]$.
For the converse direction, we need to recover naturality.
Rather than rounding an increasing solution, which need not preserve either the
inequality or the convolution constraints, we use it as a vector of net
contributions and apply the weighted convolution together with the offset.
This transformation can distort the values of the monomials only by fixed
multiplicative factors.
A suitable scaling of $\PFun[+]$ absorbs these factors and ensures that every
sufficiently large Diophantine increasing solution yields a Diophantine natural
solution of the original inequality.

\begin{lemma}
\label{lem:solcon}
  Let $\PStr = \tuple {\unkVec} {\preccurlyeq} {\wFun} {\oVec}$ be a
  multiplicity poset, $\unkVec[\top] \subseteq \unkVec$ an upward-closed
  $n$-vector of unknowns, and $\MFun(\unkVec[\top]) > \PFun(\unkVec[\top])$ an
  $n$-\MPI.
  Then, there exists $\alphaElm \in \SetR[\geq 1]$ such that $\alphaElm
  \PFun[+](\unkVec[\top])$ has coefficients in $\SetR[\geq 1]$ and, for every
  $\lambdaVec \in \SetN[+][n]$, there exists $\kappaVec[\lambdaVec] \in
  \SetN[+][n]$ such that, if $\MFun(\unkVec[\top]) > \alphaElm
  \PFun[+](\unkVec[\top])$ has a Diophantine increasing solution $\solVec \in
  \SetN[+][n]$ with $\solVec \geq \kappaVec[\lambdaVec]$, then
  $\MFun(\unkVec[\top]) > \PFun(\unkVec[\top])$ has a Diophantine natural
  increasing solution $\solVec[\lambdaVec] \in \SetN[+][n]$ with
  $\solVec[\lambdaVec] \geq \lambdaVec$.
\end{lemma}
\begin{proof}
  Let $\MFun(\unkVec[\top]) = \unkVec[\top]^{\expVec[0]}$ and
  $\PFun(\unkVec[\top]) = \sum_{i = 1}^{m} \alphaElm[i]
  \unkVec[\top]^{\expVec[i]}$, for some $\{ \alphaElm[i] \}_{i = 1}^{m}
  \subseteq \SetR[\neq 0]$ and $\{ \expVec[i] \}_{i = 0}^{m} \subseteq
  \SetN[][n]$, with $m \in \SetN$, and recall that $\PFun[+](\unkVec[\top]) =
  \sum_{1 \leq i \leq m}^{\alphaElm[i] > 0} \alphaElm[i]
  \unkVec[\top]^{\expVec[i]}$.
  Moreover, let $\wFun[\top]$ be the restriction of $\wFun$ to $\unkVec[\top]
  \times \unkVec[\top]$, $\oVec[\top]$ the restriction of $\oVec$ to
  $\unkVec[\top]$, $\iotaVec = \{ \iotaElm[\unkElm] \}_{\unkElm \in
  \unkVec[\top]} \in \SetN[+][n]$ the $n$-vector of positive natural numbers,
  whose components are defined as $\iotaElm[\unkElm] \defeq \wFun(\unkElm,
  \unkElm)$, for every $\unkElm \in \unkVec[\top]$, and $\tauVec \defeq
  \oVec[\top] + \wFun[\top] \star \vec*{1}$, where $\vec*{1} \in \SetN[+][n]$ is
  the constant-one vector on $\unkVec[\top]$.
  Since $\PStr$ is a multiplicity poset, it holds that $\wFun^{-1} \star \oVec
  \in \SetN[][n + k]$ and, therefore, $\oVec = \wFun \star (\wFun^{-1} \star
  \oVec) \in \SetN[][n + k]$, from which it follows that $\oVec[\top] \in
  \SetN[][n]$.
  Hence, $\tauVec \in \SetN[+][n]$.
  Consider also the values $\cElm[0] \defeq \iotaVec^{\expVec[0]}$ and $\cElm[i]
  \defeq \tauVec^{\expVec[i]}$, for all $i \in \num{m}$.
  If $\PFun[+](\unkVec[\top])$ is the zero polynomial, let $\alphaElm \defeq 1$.
  Otherwise, let $\alphaElm \in \SetR[\geq 1]$ be such that $\alphaElm
  \alphaElm[i] \geq 1$, for every $i \in \num{m}$ with $\alphaElm[i] > 0$, and
  $\alphaElm \geq \max_{i \in \num{m}}^{\alphaElm[i] > 0}
  \frac{\cElm[i]}{\cElm[0]}$.
  In both cases, $\alphaElm \PFun[+](\unkVec[\top])$ has coefficients in
  $\SetR[\geq 1]$.
  Now, let $\lambdaVec \in \SetN[+][n]$ be an arbitrary $n$-vector and
  $\kappaVec[\lambdaVec] = \{ \kappaElm[\unkElm] \}_{\unkElm \in \unkVec[\top]}
  \in \SetN[+][n]$ the corresponding $n$-vector with $\kappaElm[\unkElm] \defeq
  \ceil{\frac{\lambdaVec[\unkElm]}{\iotaElm[\unkElm]}}$, for every $\unkElm \in
  \unkVec[\top]$.
  Furthermore, assume that $\MFun(\unkVec[\top]) > \alphaElm
  \PFun[+](\unkVec[\top])$ has a Diophantine increasing solution $\solVec \in
  \SetN[+][n]$ with $\solVec \geq \kappaVec[\lambdaVec]$ and define
  $\solVec[\lambdaVec] \defeq \oVec[\top] + \wFun[\top] \star \solVec$.
  Before proceeding, also consider the $\oVec$-extension $\solVec[\lambdaVec,
  \oVec] \in \SetN[][n + k]$ of $\solVec[\lambdaVec]$ and the
  $\vec*{0}$-extension $\solVec[\vec*{0}] \in \SetN[][n + k]$ of $\solVec$, \ie,
  the vector that is equal to $\solVec$ on $\unkVec[\top]$ and to $\vec*{0}$ on
  $\unkVec \setminus \unkVec[\top]$.
  Observe that $\solVec[\lambdaVec, \oVec] - \oVec = \wFun \star
  \solVec[\vec*{0}]$.
  Since $\PStr$ is a multiplicity poset, it holds that $\solVec[\lambdaVec,
  \oVec] \in \SetN[][n + k]$.
  Moreover, $\wFun^{-1} \star (\solVec[\lambdaVec, \oVec] - \oVec) = \wFun^{-1}
  \star (\wFun \star \solVec[\vec*{0}]) = \solVec[\vec*{0}] \in \SetN[][n + k]$.
  Thus, $\solVec[\lambdaVec]$ is natural.
  We now prove that $\solVec[\lambdaVec]$ is increasing.
  Since $\solVec \in \SetN[+][n]$ is increasing, the convolution action
  $\wFun[\top] \star \solVec$ is increasing as well.
  In addition, $\oVec$ is non-decreasing, since $\oVec = \wFun \star (\wFun^{-1}
  \star \oVec)$ and $\wFun^{-1} \star \oVec \in \SetN[][n + k]$.
  Hence, $\oVec[\top]$ is non-decreasing on $\unkVec[\top]$, from which it
  follows that their sum $\solVec[\lambdaVec] = \oVec[\top] + \wFun[\top] \star
  \solVec$ is increasing.
  At this point, we show that $\solVec[\lambdaVec] \geq \lambdaVec$.
  First observe that, for every $\unkElm \in \unkVec[\top]$, it holds that
  $(\solVec[\lambdaVec])_{\unkElm} = \oElm[\unkElm] + (\wFun[\top] \star
  \solVec)_{\unkElm} \geq \wFun(\unkElm, \unkElm) \solElm[\unkElm] =
  \iotaElm[\unkElm] \solElm[\unkElm]$.
  Hence, by definition of $\kappaVec[\lambdaVec]$ and the fact that $\solVec
  \geq \kappaVec[\lambdaVec]$, we get $(\solVec[\lambdaVec])_{\unkElm} \geq
  \iotaElm[\unkElm] \solElm[\unkElm] \geq \iotaElm[\unkElm] \kappaElm[\unkElm] =
  \iotaElm[\unkElm] \ceil{\frac{\lambdaVec[\unkElm]}{\iotaElm[\unkElm]}} \geq
  \lambdaVec[\unkElm]$.
  Finally, it remains to prove that $\solVec[\lambdaVec]$ is a solution of the
  $n$-\MPI $\MFun(\unkVec[\top]) > \PFun(\unkVec[\top])$.
  If $\PFun[+](\unkVec[\top])$ is the zero polynomial, then
  $\MFun(\solVec[\lambdaVec]) > 0 \geq \PFun(\solVec[\lambdaVec])$, so,
  $\solVec[\lambdaVec]$ is clearly a solution.
  Otherwise, first observe that, for every $\unkElm', \unkElm \in \unkVec[\top]$
  with $\unkElm' \preccurlyeq \unkElm$, it holds that $\solElm[\unkElm'] \leq
  \solElm[\unkElm]$, since $\solVec$ is increasing and $\PStr$ is finite.
  Therefore, for every $\unkElm \in \unkVec[\top]$, it holds that
  \begin{alignat*}{3}
    {\left( \solVec[\lambdaVec] \right)_{\unkElm}}
  & =
    {\left( \oVec[\top] + \wFun[\top] \star \solVec \right)_{\unkElm}}
  && = \\
  & =
    {\oElm[\unkElm]  + \left( \wFun[\top] \star \solVec \right)_{\unkElm}}
  && \leq \\
  & \leq
    {\oElm[\unkElm] \solElm[\unkElm] + \left( \wFun[\top] \star \solVec
    \right)_{\unkElm}}
  && = \\
  & =
    {\oElm[\unkElm] \solElm[\unkElm] + \sum_{\unkElm' \preccurlyeq
    \unkElm}^{\unkElm' \in \unkVec[\top]} \wFun[\top](\unkElm', \unkElm)
    \solElm[\unkElm']}
  && \leq \\
  & \leq
    {\oElm[\unkElm] \solElm[\unkElm] + \sum_{\unkElm' \preccurlyeq
    \unkElm}^{\unkElm' \in \unkVec[\top]} \wFun[\top](\unkElm', \unkElm)
    \solElm[\unkElm]}
  && \leq \\
  & \leq
    {\left( \oElm[\unkElm] + \sum_{\unkElm' \preccurlyeq \unkElm}^{\unkElm' \in
    \unkVec[\top]} \wFun[\top](\unkElm', \unkElm) \right) \solElm[\unkElm]}
  && = \\
  & =
    {\left( \oElm[\unkElm] + \left( \wFun[\top] \star \vec*{1} \right)_{\unkElm}
    \right) \solElm[\unkElm]}
  && = \\
  & =
    {\left( \oVec[\top] + \wFun[\top] \star \vec*{1} \right)_{\unkElm}
    \solElm[\unkElm]}
  && =
    {\tauElm[\unkElm] \solElm[\unkElm]} \,.
  \end{alignat*}
  At this point, by combining the upper and lower bounds on
  $\solVec[\lambdaVec]$, \ie, $(\solVec[\lambdaVec])_{\unkElm} \leq
  \tauElm[\unkElm] \solElm[\unkElm]$ and $(\solVec[\lambdaVec])_{\unkElm} \geq
  \iotaElm[\unkElm] \solElm[\unkElm]$, respectively, one can get the following
  inequality, for every $i \in \num{m}$ with $\alphaElm[i] > 0$:
  \[
    \frac{\solVec[\lambdaVec]^{\expVec[i]}}{\solVec[\lambdaVec]^{\expVec[0]}}
  \leq
    \frac{\tauVec^{\expVec[i]}}{\iotaVec^{\expVec[0]}} \cdot
    \frac{\solVec^{\expVec[i]}}{\solVec^{\expVec[0]}}
  =
    \frac{\cElm[i]}{\cElm[0]} \cdot
    \frac{\solVec^{\expVec[i]}}{\solVec^{\expVec[0]}} \,.
  \]
  Hence,
  \begin{alignat*}{4}
    {\frac{\PFun[+](\solVec[\lambdaVec])}{\MFun(\solVec[\lambdaVec])}}
  & =
    {\frac{\sum_{1 \leq i \leq m}^{\alphaElm[i] > 0} \alphaElm[i]
    \solVec[\lambdaVec]^{\expVec[i]}}{\solVec[\lambdaVec]^{\expVec[0]}}}
  && = \\
  & =
    {\sum_{1 \leq i \leq m}^{\alphaElm[i] > 0} \alphaElm[i]
    \frac{\solVec[\lambdaVec]^{\expVec[i]}}{\solVec[\lambdaVec]^{\expVec[0]}}}
  && \leq \\
  & \leq
    {\sum_{1 \leq i \leq m}^{\alphaElm[i] > 0} \alphaElm[i]
    \frac{\cElm[i]}{\cElm[0]} \frac{\solVec^{\expVec[i]}}{\solVec^{\expVec[0]}}}
  && \leq \\
  & \leq
    {\sum_{1 \leq i \leq m}^{\alphaElm[i] > 0} \alphaElm[i] \alphaElm
    \frac{\solVec^{\expVec[i]}}{\solVec^{\expVec[0]}}}
  && \leq \\
  & \leq
    {\frac{\alphaElm \sum_{1 \leq i \leq m}^{\alphaElm[i] > 0} \alphaElm[i]
    \solVec^{\expVec[i]}}{\solVec^{\expVec[0]}}}
  && =
    {\frac{\alphaElm \PFun[+](\solVec)}{\MFun(\solVec)}}
  &&& <
    1
  \end{alignat*}
  The strict inequality in the last line follows from the fact that $\solVec$
  is a solution of $\MFun(\unkVec[\top]) > \alphaElm \PFun[+](\unkVec[\top])$.
  Thus, $\MFun(\solVec[\lambdaVec]) > \PFun[+](\solVec[\lambdaVec])$.
  Since $\solVec[\lambdaVec] \in \SetN[+][n]$, it holds that
  $\PFun(\solVec[\lambdaVec]) \leq \PFun[+](\solVec[\lambdaVec])$.
  Therefore, $\MFun(\solVec[\lambdaVec]) > \PFun(\solVec[\lambdaVec])$.
  Summing up, $\solVec[\lambdaVec]$ is a Diophantine natural increasing solution
  of the $n$-\MPI $\MFun(\unkVec[\top]) > \PFun(\unkVec[\top])$ with
  $\solVec[\lambdaVec] \geq \lambdaVec$.
\end{proof}

The two preceding lemmas can now be combined with the increasing-solution
characterisation of \cref{thm:incsol}.
They show that, as long as all unknowns are bounded below by $1$, Diophantine
natural solvability, feasibility of the positive linear system, and the
existence of arbitrarily large Diophantine natural increasing solutions are
equivalent.

\begin{theorem}[Positive Solution]
\label{thm:possol}
  For every multiplicity poset $\PStr = \tuple {\unkVec} {\preccurlyeq} {\wFun}
  {\oVec}$, upward-closed $n$-vector of unknowns $\unkVec[\top] \subseteq
  \unkVec$, and $n$-{\MPI} $\MFun(\unkVec[\top]) > \PFun(\unkVec[\top])$, with
  $\PFun(\unkVec[\top])$ strongly non-negative \wrt $\PStr$, the following
  statements are equivalent:
  \begin{enumerate}[a)]
  \item\label{thm:possol(sol)}
    the $n$-\MPI admits a Diophantine natural solution in $\SetN[+][n]$;
  \item\label{thm:possol(sys)}
    the $\tuple* {\MFun(\unkVec[\top]) \!>\! \PFun[+](\unkVec[\top])}
    {\der{\PStr}} [+]$-system admits a solution, where $\der{\PStr} \defeq
    \tuple {\unkVec[\top]} {{\preccurlyeq} \cap \unkVec[\top] \times
    \unkVec[\top]}$;
  \item\label{thm:possol(dphsol)}
    for all $\lambdaVec \!\in\! \SetN[+][n]$, the $n$-\MPI admits a Diophantine
    natural increasing solution $\solVec[\lambdaVec] \geq \lambdaVec$.
  \end{enumerate}
\end{theorem}
\begin{proof}
  We prove the equivalence by focusing on the only non-immediate implications,
  namely \ref{thm:possol(sol)} $\Rightarrow$~\ref{thm:possol(sys)}
  and~\ref{thm:possol(sys)} $\Rightarrow$~\ref{thm:possol(dphsol)},
  since~\ref{thm:possol(dphsol)} $\Rightarrow$~\ref{thm:possol(sol)} follows
  immediately by considering $\lambdaVec = \vec*{1}$.
  Before proceeding, observe that the two systems $\tuple* {\MFun(\unkVec[\top])
  > \PFun[+](\unkVec[\top])} {\der{\PStr}} [+]$ and $\tuple*
  {\MFun(\unkVec[\top]) > \alphaElm \PFun[+](\unkVec[\top])} {\der{\PStr}} [+]$
  of homogeneous linear inequalities coincide, for every $\alphaElm \in
  \SetR[+]$, since multiplication by a positive scalar does not change the
  exponent vectors of the positive monomials considered in those systems.
  We now prove the two directions of the equivalence separately.
  \begin{itemize}
  \item
    \textbf{[\ref{thm:possol(sol)} $\Rightarrow$ \ref{thm:possol(sys)}]}
    Assume that the $n$-\MPI $\MFun(\unkVec[\top]) > \PFun(\unkVec[\top])$ has a
    Diophantine natural solution in $\SetN[+][n]$.
    If $\PFun[+](\unkVec[\top])$ is the zero polynomial, let $\alphaElm \defeq
    1$.
    Otherwise, let $\alphaElm \in \SetR[\geq 1]$ be a value such that the
    polynomial $\alphaElm \PFun[+](\unkVec[\top])$ has coefficients in
    $\SetR[\geq 1]$.
    By \cref{lem:solred}, the $n$-\MPI $\MFun(\unkVec[\top]) > \alphaElm
    \PFun[+](\unkVec[\top])$ has an admissible increasing solution in
    $\SetR[\geq 1][n]$.
    Thus, \cref{thm:incsol(sol)} of \cref{thm:incsol} is satisfied for
    $\MFun(\unkVec[\top]) > \alphaElm \PFun[+](\unkVec[\top])$ \wrt
    $\der{\PStr}$.
    Hence, by \cref{thm:incsol(sys)} of the same theorem, the $\tuple*
    {\MFun(\unkVec[\top]) \!>\! \alphaElm \PFun[+](\unkVec[\top])} {\der{\PStr}}
    [+]$-system admits a solution.
    By the initial observation, the $\tuple* {\MFun(\unkVec[\top]) \!>\!
    \PFun[+](\unkVec[\top])} {\der{\PStr}} [+]$-system admits a solution as
    well.
  \item
    \textbf{[\ref{thm:possol(sys)} $\Rightarrow$ \ref{thm:possol(dphsol)}]}
    Assume that the $\tuple* {\MFun(\unkVec[\top]) \!>\!
    \PFun[+](\unkVec[\top])} {\der{\PStr}} [+]$-system admits a solution.
    By \cref{lem:solcon}, there exists $\alphaElm \in \SetR[\geq 1]$ such that
    $\alphaElm \PFun[+](\unkVec[\top])$ has coefficients in $\SetR[\geq 1]$ and,
    for every $\lambdaVec \in \SetN[+][n]$, there exists $\kappaVec[\lambdaVec]
    \in \SetN[+][n]$ such that, if $\MFun(\unkVec[\top]) > \alphaElm
    \PFun[+](\unkVec[\top])$ has a Diophantine increasing solution $\solVec \in
    \SetN[+][n]$ with $\solVec \geq \kappaVec[\lambdaVec]$, then
    $\MFun(\unkVec[\top]) > \PFun(\unkVec[\top])$ has a Diophantine natural
    increasing solution $\solVec[\lambdaVec] \in \SetN[+][n]$ with
    $\solVec[\lambdaVec] \geq \lambdaVec$.
    By the initial observation, the $\tuple* {\MFun(\unkVec[\top]) \!>\!
    \alphaElm \PFun[+](\unkVec[\top])} {\der{\PStr}} [+]$-system admits a
    solution as well.
    Thus, \cref{thm:incsol(sys)} of \cref{thm:incsol} is satisfied for
    $\MFun(\unkVec[\top]) > \alphaElm \PFun[+](\unkVec[\top])$ \wrt
    $\der{\PStr}$.
    Now, let $\lambdaVec \in \SetN[+][n]$ be an arbitrary $n$-vector.
    By \cref{thm:incsol(dphsol)} of the same theorem, the $n$-\MPI
    $\MFun(\unkVec[\top]) > \alphaElm \PFun[+](\unkVec[\top])$ has a Diophantine
    increasing solution $\solVec[{\kappaVec[\lambdaVec]}] \in \SetN[+][n]$ with
    $\solVec[{\kappaVec[\lambdaVec]}] \geq \kappaVec[\lambdaVec]$.
    By the choice of $\kappaVec[\lambdaVec]$, the $n$-\MPI $\MFun(\unkVec[\top])
    > \PFun(\unkVec[\top])$ has a Diophantine natural increasing solution
    $\solVec[\lambdaVec] \in \SetN[+][n]$ with $\solVec[\lambdaVec] \geq
    \lambdaVec$.
    \qedhere
  \end{itemize}
\end{proof}

The previous theorem still assumes that every unknown is assigned a value at
least $1$, whereas \cref{prb:nmpi} also admits Diophantine natural solutions
containing zeros.
We therefore isolate the unknowns that are forced to remain positive.
Naturality and positiveness of the weights imply that the positive support of
any natural valuation is upward closed.
Moreover, for the strongly non-negative instances of \cref{prb:nmpi}, every
unknown occurring in $\MFun$, as well as every unknown above a positive offset
contribution, must belong to this support for any Diophantine natural solution.
This motivates a canonical upward-closed restriction on which the positive
solution theorem can be applied.
In contrast with the collapse construction used in~\cite{KM25}, this reduction
works directly on the original finite poset and does not require any
meet-semilattice constraint on the underlying structure.

For a multiplicity poset $\PStr = \tuple {\unkVec} {\preccurlyeq} {\wFun}
{\oVec}$, an upward-closed $n$-vector of unknowns $\unkVec[\top] \subseteq
\unkVec$, and an $n$-{\MPI} $\MFun(\unkVec[\top]) > \PFun(\unkVec[\top])$ where
$\MFun(\unkVec[\top]) = \unkVec[\top]^{\expVec[0]}$ and $\PFun(\unkVec[\top]) =
\sum_{i = 1}^{m} \alphaElm[i] \unkVec[\top]^{\expVec[i]}$, for some $\{
\alphaElm[i] \}_{i = 1}^{m} \subseteq \SetR[\neq 0]$ and $\{ \expVec[i] \}_{i =
0}^{m} \subseteq \SetN[][n]$, with $m \in \SetN$, we define the \emph{reduced
upward-closed vector of unknowns} $\unkVec[\uparrow]$ as the subvector of
$\unkVec[\top]$ such that $\unkElm \in \unkVec[\uparrow]$ \iff there exists an
unknown $\unkElm[][*] \in \unkVec[\top]$ such that $\unkElm[][*] \preccurlyeq
\unkElm$ and at least one of the following holds: $\expElm[0{\unkElm[][*]}] > 0$
or $\oElm[{\unkElm[][*]}] > 0$.
Intuitively, $\unkVec[\uparrow]$ contains exactly the unknowns of
$\unkVec[\top]$ that cannot be set to zero without annihilating the monomial
$\MFun$ or removing a positive contribution forced by the offset vector.
We then define the \emph{reduced polynomial}
$\PFun[][\uparrow](\unkVec[\uparrow])$ \wrt $\unkVec[\uparrow]$ as the
polynomial obtained from $\PFun(\unkVec[\top])$ by setting to $0$ all
unknowns in $\unkVec[\top] \setminus \unkVec[\uparrow]$, namely
$\PFun[][\uparrow](\unkVec[\uparrow]) \defeq \PFun(\unkVec[\top])[\unkElm \in
\unkVec[\top] \setminus \unkVec[\uparrow] \mapsto 0]$.
Equivalently, $\PFun[][\uparrow](\unkVec[\uparrow]) \defeq \sum_{i \in \ISet}
\alphaElm[i] \unkVec[\uparrow]^{\expVec[i] \rst[{\unkVec[\uparrow]}]}$, where
$\ISet \subseteq \num{m}$ is the set of indices $1 \leq i \leq m$ such that
$\expElm[i \unkElm] > 0$ implies $\unkElm \in \unkVec[\uparrow]$, for all
$\unkElm \in \unkVec[\top]$, and $\expVec[i] \rst[{\unkVec[\uparrow]}]$ is the
restriction of $\expVec[i]$ to the unknowns in $\unkVec[\uparrow]$.
Clearly, every valuation $\solVec[\uparrow]$ of $\unkVec[\uparrow]$ can be
extended to a valuation $\solVec$ of $\unkVec[\top]$ by setting
$\solElm[\unkElm] \defeq \solElm[\uparrow \unkElm]$, for every $\unkElm \in
\unkVec[\uparrow]$, and $\solElm[\unkElm] \defeq 0$, for every $\unkElm \in
\unkVec[\top] \setminus \unkVec[\uparrow]$.
By construction, it then holds that $\MFun(\solVec) = \MFun(\solVec[\uparrow])$
and $\PFun(\solVec) = \PFun[][\uparrow](\solVec[\uparrow])$.
Consequently, every solution of the reduced $k$-\MPI $\MFun(\unkVec[\uparrow]) >
\PFun[][\uparrow](\unkVec[\uparrow])$ induces, via zero-extension, a solution of
the original $n$-\MPI $\MFun(\unkVec[\top]) > \PFun(\unkVec[\top])$.

The reduced vector is upward closed and contains every unknown forced to be
positive by either $\MFun$ or the offset.
Consequently, the restriction $\PStr[\uparrow]$ is again a multiplicity poset,
and every Diophantine natural valuation over it extends by zero to a Diophantine
natural valuation over $\PStr$.
For the same reason, strong non-negativeness is inherited by the reduced
polynomial: evaluating the deleted unknowns as $0$ turns every Diophantine
natural valuation of the reduced problem into a Diophantine natural valuation of
the original one.
We can therefore apply the positive-solution characterisation to the reduced
problem and obtain the final result.

\begin{theorem}[General Solution]
\label{thm:gensol}
  For every multiplicity poset $\PStr = \tuple {\unkVec} {\preccurlyeq} {\wFun}
  {\oVec}$, upward-closed $n$-vector of unknowns $\unkVec[\top] \subseteq
  \unkVec$, and $n$-{\MPI} $\MFun(\unkVec[\top]) > \PFun(\unkVec[\top])$, with
  $\PFun(\unkVec[\top])$ strongly non-negative \wrt $\PStr$, the following
  statements are equivalent, where $\unkVec[\uparrow]$ is the reduced
  upward-closed $k$-vector of unknowns, $\PFun[][\uparrow](\unkVec[\uparrow])$
  the reduced polynomial \wrt $\unkVec[\uparrow]$, and $\PStr[\uparrow] \defeq
  \tuple {\unkVec[\uparrow]} {{\preccurlyeq} \cap \unkVec[\uparrow] \times
  \unkVec[\uparrow]} {\wFun \rst[{\unkVec[\uparrow] \times \unkVec[\uparrow]}]}
  {\oVec \rst[{\unkVec[\uparrow]}]}$:
  \begin{enumerate}[a)]
  \item\label{thm:gensol(sol)}
    the original $n$-\MPI $\MFun(\unkVec[\top]) > \PFun(\unkVec[\top])$ admits a
    Diophantine natural solution \wrt $\PStr$;
  \item\label{thm:gensol(sys)}
    the $\tuple* {\MFun(\unkVec[\uparrow]) \!>\!
    \PFun[+][\uparrow](\unkVec[\uparrow])} {\der{\PStr}[\uparrow]} [+]$-system
    admits a solution, where $\der{\PStr}[\uparrow] \defeq \tuple
    {\unkVec[\uparrow]} {{\preccurlyeq} \cap \unkVec[\uparrow] \times
    \unkVec[\uparrow]}$;
  \item\label{thm:gensol(redsol)}
    for all $\lambdaVec \!\in\! \SetN[+][k]$, the reduced $k$-\MPI
    $\MFun(\unkVec[\uparrow]) > \PFun[][\uparrow](\unkVec[\uparrow])$ admits a
    Diophantine natural increasing solution $\solVec[\lambdaVec] \geq
    \lambdaVec$ \wrt $\PStr[\uparrow]$.
  \end{enumerate}
\end{theorem}
\begin{proof}
  We prove the equivalence by focusing on the only non-immediate implication,
  namely \ref{thm:gensol(sol)} $\Rightarrow$ \ref{thm:gensol(sys)}.
  Indeed, \ref{thm:gensol(sys)} $\Rightarrow$ \ref{thm:gensol(redsol)} clearly
  follows from the implication \ref{thm:possol(sys)} $\Rightarrow$
  \ref{thm:possol(dphsol)} of \cref{thm:possol} applied to the reduced $k$-\MPI,
  since $\PFun[][\uparrow](\unkVec[\uparrow])$ is strongly non-negative \wrt
  $\PStr[\uparrow]$.
  The latter property follows from the strong non-negativeness of
  $\PFun(\unkVec[\top])$ \wrt $\PStr$, since every Diophantine natural valuation
  over $\PStr[\uparrow]$ zero-extends to a Diophantine natural valuation over
  $\PStr$, and the reduced polynomial is obtained precisely by evaluating all
  deleted unknowns as $0$.
  Moreover, \ref{thm:gensol(redsol)} $\Rightarrow$ \ref{thm:gensol(sol)} follows
  by considering $\lambdaVec = \vec*{1}$ and zero-extending the obtained
  solution from $\unkVec[\uparrow]$ to $\unkVec[\top]$.
  \begin{itemize}
  \item
    \textbf{[\ref{thm:gensol(sol)} $\Rightarrow$ \ref{thm:gensol(sys)}]}
    Assume that the original $n$-\MPI $\MFun(\unkVec[\top]) >
    \PFun(\unkVec[\top])$ admits a Diophantine natural solution $\solVec$ \wrt
    $\PStr$ and let $\unkVec[\neq 0] \subseteq \unkVec[\top]$ be the positive
    support of $\solVec$, \ie, $\unkElm \in \unkVec[\neq 0]$ \iff
    $\solElm[\unkElm] > 0$.
    We first observe that $\unkVec[\neq 0]$ is upward closed, as by the
    properties of the convolution action $\solVec$ is non-decreasing.
    Indeed, by naturality, every positive value $\solElm[\unkElm]$ is generated
    by a positive net-or-offset contribution below $\unkElm$; since the
    multiplicity weights are positive, the same contribution is propagated to
    every unknown above $\unkElm$.
    Moreover, $\unkVec[\uparrow] \subseteq \unkVec[\neq 0]$.
    First, every unknown occurring in the monomial $\MFun$ belongs to
    $\unkVec[\neq 0]$.
    Otherwise, we would have had $\MFun(\solVec) = 0$, whereas the strong
    non-negativeness of $\PFun$ gives $\PFun(\solVec) \geq 0$, contradicting
    $\MFun(\solVec) > \PFun(\solVec)$.
    Second, if $\unkElm[][*] \preccurlyeq \unkElm$ and $\oElm[{\unkElm[][*]}] >
    0$, then naturality and positiveness of the multiplicity weights give
    $\solElm[{\unkElm[][*]}] > 0$.
    Since $\solVec$ is non-decreasing, it follows that $\solElm[\unkElm] > 0$.
    Let $\PFun[][\neq 0](\unkVec[\neq 0])$ be the polynomial obtained from
    $\PFun(\unkVec[\top])$ by setting to $0$ all unknowns in $\unkVec[\top]
    \setminus \unkVec[\neq 0]$, and set $\PStr[\neq 0] \defeq \PStr
    \rst[{\unkVec[\neq 0]}]$.
    Since $\solVec$ is $0$ outside $\unkVec[\neq 0]$, its restriction
    $\solVec[\neq 0]$ to $\unkVec[\neq 0]$ satisfies $\MFun(\solVec[\neq 0]) >
    \PFun[][\neq 0](\solVec[\neq 0])$.
    It is also a Diophantine natural positive solution \wrt $\PStr[\neq 0]$.
    In addition, $\PFun[][\neq 0](\unkVec[\neq 0])$ is strongly non-negative
    \wrt $\PStr[\neq 0]$.
    Indeed, every Diophantine natural valuation over $\PStr[\neq 0]$
    zero-extends to a Diophantine natural valuation over $\PStr$, since
    $\unkVec[\neq 0]$ is upward closed and contains all unknowns of
    $\unkVec[\top]$ that are forced by the offset.
    Moreover, $\PFun[][\neq 0]$ is obtained from $\PFun$ by evaluating all
    deleted unknowns as $0$.
    Therefore, by \cref{thm:possol} applied to the \MPI $\MFun(\unkVec[\neq 0])
    > \PFun[][\neq 0](\unkVec[\neq 0])$, the $\tuple* {\MFun(\unkVec[\neq 0])
    \!>\! \PFun[+][\neq 0](\unkVec[\neq 0])} {\der{\PStr}[\neq 0]} [+]$-system
    admits a solution, say $\dsolVec[\neq 0]$.
    We now restrict $\dsolVec[\neq 0]$ to $\unkVec[\uparrow]$ obtaining
    $\dsolVec[\neq 0] \rst[{\unkVec[\uparrow]}]$.
    All order constraints of $\der{\PStr}[\uparrow]$ are order constraints of
    $\der{\PStr}[\neq 0]$, since $\unkVec[\uparrow] \subseteq \unkVec[\neq 0]$.
    Moreover, every positive monomial of the reduced polynomial
    $\PFun[][\uparrow](\unkVec[\uparrow])$ is also a positive monomial of
    $\PFun[][\neq 0](\unkVec[\neq 0])$ whose support is contained in
    $\unkVec[\uparrow]$.
    Hence every strict degree inequality required by $\MFun(\unkVec[\uparrow]) >
    \PFun[+][\uparrow](\unkVec[\uparrow])$ is already satisfied by
    $\dsolVec[\neq 0]$.
    Thus the restriction $\dsolVec[\neq 0] \rst[{\unkVec[\uparrow]}]$ is a
    solution of the $\tuple* {\MFun(\unkVec[\uparrow]) \!>\!
    \PFun[+][\uparrow](\unkVec[\uparrow])} {\der{\PStr}[\uparrow]}[+]$-system.
    \qedhere
  \end{itemize}
\end{proof}




\subsection{Decision Procedure for {$n$-\MPI}s}
\label{sec:dphprb;sub:decprcnmpi}

To effectively solve \cref{prb:nmpi}, we leverage the combination of
\cref{thm:gensol,lem:syssol} via a $\QEA$-alternating polynomial-time algorithm
that outputs $\Tt$ \iff the problem has an affirmative answer.
By \cref{thm:gensol}, it suffices to decide whether the $\tuple*
{\MFun(\unkVec[\uparrow]) \!>\! \PFun[+][\uparrow](\unkVec[\uparrow])}
{\der{\PStr}[\uparrow]} [+]$-system induced by the reduced $k$-\MPI admits a
solution.
By \cref{lem:syssol}, whenever this system is feasible, it admits a $(6 \phi
k^{3})$-bit-bounded Diophantine positive solution.
We can therefore adopt a standard guess-and-check strategy: given a multiplicity
poset $\PStr = \tuple {\unkVec} {\preccurlyeq} {\wFun} {\oVec}$, an
upward-closed $n$-vector of unknowns $\unkVec[\top] \subseteq \unkVec$, and an
$n$-\MPI $\MFun(\unkVec[\top]) > \PFun(\unkVec[\top])$, with
$\PFun(\unkVec[\top])$ strongly non-negative \wrt $\PStr$, we first guess a $(6
\phi k^{3})$-bit-bounded Diophantine positive solution $\dsolVec \in
\SetN[+][k]$ for the induced $\tuple* {\MFun(\unkVec[\uparrow]) \!>\!
\PFun[+][\uparrow](\unkVec[\uparrow])} {\der{\PStr}[\uparrow]} [+]$-system of
homogeneous linear inequalities and then universally check its degree and
strictness constraints.
If either family of constraints is empty, the corresponding universal check is
simply omitted.
As a consequence, we derive the following upper bound on the complexity of our
\MPI problem.

\begin{theorem}[Decision Procedure]
\label{thm:decprc}
  \cref{prb:nmpi} can be decided in $\QEA$-alternating polynomial-time with
  existential-guess space polynomial in the size of the encoding of the induced
  homogeneous linear system and universal-guess space logarithmic in the maximum
  between the number $n$ of unknowns and the number $m$ of monomials in the
  $n$-\MPI.
\end{theorem}
\begin{proof}
  The decision procedure we provide, whose pseudo code is reported in
  \cref{alg:nmpisol}, can be interpreted as a polynomial-time $\QEA$-alternating
  algorithm, or equivalently as a polynomial-time $\QEA$-alternating RAM
  machine.
  When applied to the reduced instance, the solver receives
  $\MFun(\unkVec[\uparrow])$, $\PFun[][\uparrow](\unkVec[\uparrow])$, and
  $\der{\PStr}[\uparrow]$ as input, and therefore operates on $k \leq n$
  unknowns.
  The algorithm is formed by a first phase of existential guesses of global
  space $\AOmicron{6 \phi k^{3}}$, hence $\AOmicron{6 \phi n^{3}}$ (see
  \cref{alg:nmpisol(exs:sol)}), followed by a phase of universal guesses of
  global space $\AOmicron{\log_{2} \max \{ m, n \}}$ (see
  \cref{alg:nmpisol(unv:sol),alg:nmpisol(unv:str)}), interleaved by
  polynomial-time operations.
  Recall that the facet complexity $\phi$ is polynomial in the encoding size of
  the linear system~\cite{Sch86}.

  We now prove that the algorithm returns the truth value $\Tt$ \iff there
  exists a Diophantine natural solution of the $n$-\MPI $\MFun(\unkVec[\top]) >
  \PFun(\unkVec[\top])$ under scrutiny.
  By \cref{thm:gensol}, this is equivalent to proving that the algorithm returns
  $\Tt$ \iff the $\tuple* {\MFun(\unkVec[\uparrow]) \!>\!
  \PFun[+][\uparrow](\unkVec[\uparrow])} {\der{\PStr}[\uparrow]} [+]$-system of
  homogeneous linear inequalities admits a solution.
  Assume first that the procedure outputs $\Tt$.
  Then, there necessarily exists a $(6 \phi k^{3})$-bit-bounded Diophantine
  positive solution $\dsolVec$ of the $\tuple* {\MFun(\unkVec[\uparrow]) \!>\!
  \PFun[+][\uparrow](\unkVec[\uparrow])} {\der{\PStr}[\uparrow]} [+]$-system of
  homogeneous linear inequalities.
  Hence, by the implication \ref{thm:gensol(sys)} $\Rightarrow$
  \ref{thm:gensol(sol)} of \cref{thm:gensol}, there exists a Diophantine natural
  solution of the $n$-\MPI $\MFun(\unkVec[\top]) > \PFun(\unkVec[\top])$.
  Conversely, assume that the $n$-\MPI $\MFun(\unkVec[\top]) >
  \PFun(\unkVec[\top])$ has a Diophantine natural solution.
  By the implication \ref{thm:gensol(sol)} $\Rightarrow$ \ref{thm:gensol(sys)}
  of \cref{thm:gensol}, the $\tuple* {\MFun(\unkVec[\uparrow]) \!>\!
  \PFun[+][\uparrow](\unkVec[\uparrow])} {\der{\PStr}[\uparrow]} [+]$-system of
  homogeneous linear inequalities admits a solution.
  Then, by \cref{lem:syssol}, there exists a $(6 \phi k^{3})$-bit-bounded
  Diophantine positive solution $\dsolVec$ for this system.
  Therefore, such a $k$-vector $\dsolVec$ can be used as existential guesses to
  force the algorithm to output $\Tt$, regardless of the universal guesses.
\end{proof}

\algnmpisol





\section{Decidability \& Complexity}
\label{sec:deccom}

We now turn the Diophantine characterisation of \cref{thm:polchr} into a
decision procedure for the containment problems associated with the classes of
containee queries introduced in \cref{sec:spccasbagcon;sub:talthrqry}, namely
\BPFQ, \BJoFQ, and \BJUQ.
The key point here is that non-containment of a \BJUQ by an arbitrary \BCQ is
equivalent to the existence of a solution of one of the monomial-polynomial
inequalities constructed in \cref{sec:commul;sub:polchr}, subject to the
naturality constraints induced by the corresponding multiplicity poset.
The general decision procedure developed in \cref{sec:dphprb} for
\cref{prb:nmpi} can then be applied to these inequalities.
This gives the algorithmic core of the solution of \cref{prb:conprb}, with
complexity \ULH[2][P], \CoNExpTime, or 2\CoNExpTime, depending on the considered
containee class.

Given a \BJUQ $\qryElm$ and a \BCQ $\pqryElm$, \cref{thm:polchr} states that
$\qryElm \not\bsinc \pqryElm$ holds \iff there exists a non-trivial ground
selection $\GSet$ for $\qryElm$ such that the associated $n$-\MPI
\[
  \MFun[\qryElm](\unkVec[\top]) > \PFun[\qryElm][\pqryElm](\GSet, \unkVec[\top])
\]
admits a Diophantine natural solution \wrt $\PStr[\qryElm, \GSet]$.
Here $\PStr[\qryElm, \GSet]$ is the multiplicity poset induced by $\qryElm$ and
$\GSet$ over the full minimal unification closure $\denot{\qryElm}$, while the
upward-closed $n$-vector of unknowns $\unkVec[\top]$ is indexed by the
non-ground part $\denot{\qryElm}[\top]$.
The weight function records the relevant homomorphism multiplicities inside the
closure, as discussed in \cref{lem:mulwgh}, and the offset vector accounts for
the ground contribution fixed by $\GSet$, as per \cref{def:grnsel}.
More precisely, the offset is the gross contribution obtained by the convolution
action of the multiplicity weight on the characteristic function of $\GSet$.
Thus, after fixing $\GSet$, the containment question is reduced to a single
instance of \cref{prb:nmpi}, which is then solved by exploiting
\cref{thm:decprc}.

The following lemma states that the objects produced by the construction of
\cref{sec:commul;sub:polchr} indeed satisfy the hypotheses required by
\cref{prb:nmpi}.

\begin{lemma}
\label{lem:conprbhyp}
  For all {\BJUQ}s $\qryElm$, {\BCQ}s $\pqryElm$, and non-trivial ground
  selections $\GSet$ for $\qryElm$, the following statements hold true:
  \begin{enumerate}[a)]
  \item\label{lem:conprbhyp(mulpos)}
    the structure $\PStr[\qryElm, \GSet] \defeq \tuple {\unkVec} {\preccurlyeq}
    {\wFun} {\oVec}$ defined as follows is a multiplicity poset:
    \begin{itemize}
    \item
      $\unkVec$ is a vector of unknowns, with an unknown $\unkElm[\rqryElm]$ per
      \BCQ $\rqryElm \in \denot{\qryElm}$; in addition, $\unkVec[\top] \subseteq
      \unkVec$ is the upward-closed sub-vector containing exactly the unknowns
      $\unkElm[\rqryElm]$ such that $\rqryElm \in \denot{\qryElm}[\top]$;
    \item
      $\unkElm[\sqryElm] \preccurlyeq \unkElm[\rqryElm]$ \iff $\sqryElm
      \preccurlyeq \rqryElm$;
    \item
      $\wFun(\unkElm[\sqryElm], \unkElm[\rqryElm]) \defeq \mFun(\sqryElm,
      \rqryElm)$, for all $\sqryElm \preccurlyeq \rqryElm$, with $\mFun$ the
      multiplicity weight of \cref{def:mulwgh};
    \item
      $\oVec \defeq (\mFun \star \chiFun[\GSet])$, where $\chiFun[\GSet]$ is the
      characteristic function of $\GSet$ over $\denot{\qryElm}$, \ie,
      $\chiFun[\GSet](\rqryElm) = 1$, if $\rqryElm \in \GSet$, and
      $\chiFun[\GSet](\rqryElm) = 0$, otherwise, for all $\rqryElm \in
      \denot{\qryElm}$.
    \end{itemize}
  \item\label{lem:conprbhyp(strnonneg)}
    the polynomial $\PFun[\qryElm][\pqryElm](\GSet, \unkVec[\top])$ is strongly
    non-negative \wrt $\PStr[\qryElm, \GSet]$.
  \end{enumerate}
\end{lemma}
\begin{proof}
  \textbf{[\cref{lem:conprbhyp(mulpos)}]}
  By \cref{thm:juqmucchr}, the minimal unification closure $\denot{\qryElm}$ is
  finite and the relation $\preccurlyeq$ induced by bag-set containment is a
  partial order on it.
  Hence, the same is true for the order induced on the unknowns in $\unkVec$.
  Moreover, the sub-vector $\unkVec[\top]$ is upward closed, since it contains
  exactly the unknowns $\unkElm[\rqryElm]$ with $\rqryElm \in
  \denot{\qryElm}[\top]$ and the latter set is clearly an upward closed subset
  of $\denot{\qryElm}$ since it contains all and only the non-ground queries.
  By \cref{lem:mulwgh}, the function $\wFun(\unkElm[\sqryElm],
  \unkElm[\rqryElm]) \defeq \mFun(\sqryElm, \rqryElm)$ is a multiplicity weight
  on this poset.
  In particular, it is defined on comparable pairs, takes values in $\SetN[+]$,
  and satisfies the required monotonicity condition: for all $\tqryElm
  \preccurlyeq \sqryElm \preccurlyeq \rqryElm$, one has
  $\wFun(\unkElm[\tqryElm], \unkElm[\sqryElm]) \leq \wFun(\unkElm[\tqryElm],
  \unkElm[\rqryElm])$.
  Finally, by definition of the offset vector, it holds that $(\wFun^{-1} \star
  \oVec)(\unkElm[\rqryElm]) = (\mFun^{-1} \star (\mFun \star
  \chiFun[\GSet]))(\rqryElm) = \chiFun[\GSet](\rqryElm) \in \SetN$, for every
  $\rqryElm \in \denot{\qryElm}$.
  Summing up, $\PStr[\qryElm, \GSet]$ is a multiplicity poset.

  \textbf{[\cref{lem:conprbhyp(strnonneg)}]}
  Let $\solVec$ be an arbitrary Diophantine natural valuation \wrt
  $\PStr[\qryElm, \GSet]$.
  We prove that $\PFun[\qryElm][\pqryElm](\GSet, \solVec)$ dominates the unitary
  form of each of its monomials with positive coefficient.
  View the $\oVec$-extension of $\solVec$ as a function $\der{\gFun}$ on
  $\denot{\qryElm}$ and put $\der{\nFun} \defeq \mFun^{-1} \star \der{\gFun}$.
  Since $\oVec = \mFun \star \chiFun[\GSet]$, it holds that $\der{\nFun} =
  \mFun^{-1} \star (\der{\gFun} - \oVec) + \chiFun[\GSet]$.
  By naturality, the first summand is natural on the non-ground elements.
  It vanishes on the ground ones, since $\der{\gFun} - \oVec$ is zero there and
  no non-ground element lies below a ground one.
  Hence, $\der{\nFun}$ is a net-image counting for $\qryElm$ and
  $\GSet[\der{\nFun}] = \GSet$.
  By~\cref{lem:netimgcnt2mulcan}, there exists a multicanonical instance
  $\dbElm[][*]$ realising this counting.
  Moreover, by~\cref{lem:nethomcnt}, for every $\rqryElm \in \denot{\qryElm}$,
  we have $\der{\gFun}(\rqryElm) = \card{\HomSet(\rqryElm, \dbElm[][*])}$.
  For every $\rqryElm \in \denot{\qryElm}[\top]$, the definition of
  $\DeltaFun[\rqryElm]$ gives $\DeltaFun[\rqryElm](\GSet, \solVec) =
  \mFun(\rqryElm, \rqryElm) \cdot \der{\nFun}(\rqryElm)$.
  Therefore, by the definition of $\beta_{\GSet, \imath}$, after cancelling the
  diagonal multiplicity factors, and by~\cref{lem:polrwt,thm:homcnt}, it follows
  that $\PFun[\qryElm][\pqryElm](\GSet, \solVec) =
  \NFun[\qryElm][\pqryElm](\GSet, \der{\nFun}) = \card{\HomSet(\pqryElm,
  \dbElm[][*])} \geq 0$.
  We now observe two consequences of this representation.
  First, $\der{\gFun}$ is non-decreasing along $\preccurlyeq$.
  Indeed, if $\sqryElm \preccurlyeq \rqryElm$, then, by~\cref{lem:mulwgh(mon)}
  of~\cref{lem:mulwgh}, it holds that
  \[
    \der{\gFun}(\sqryElm)
  =
    \sum_{\tqryElm \preccurlyeq \sqryElm} \mFun(\tqryElm, \sqryElm) \,
    \der{\nFun}(\tqryElm)
  \leq
    \sum_{\tqryElm \preccurlyeq \sqryElm} \mFun(\tqryElm, \rqryElm) \,
    \der{\nFun}(\tqryElm)
  \leq
    \der{\gFun}(\rqryElm).
  \]
  Second, fix a profile $\imath \in \Prof[\qryElm][\pqryElm]{\GSet}$ and an
  $\imath$-homomorphism $\homFun \in \HSet[\qryElm][\pqryElm](\GSet, \imath)$.
  We claim that
  \[
    \PFun[\qryElm][\pqryElm](\GSet, \solVec)
  \geq
    \prod_{\wqryElm \in \dom{\imath}} \solElm[\imath(\wqryElm)].
  \]
  To prove the claim, for every $\wqryElm \in \dom{\imath}$, choose an arbitrary
  homomorphism $\homFun[\wqryElm] \in \HomSet(\iotaFun(\imath(\wqryElm)),
  \dbElm[][*])$.
  By the definition of a profile, $\homFun(\wqryElm) \subseteq
  \iotaFun(\imath(\wqryElm))$.
  Thus, for every variable $\varElm \in \ter{\wqryElm}$ occurring in a
  non-ground $\homFun$-component $\wqryElm \in \dom{\imath}$, set
  $\trn{\homFun}(\varElm) \defeq \homFun[\wqryElm](\homFun(\varElm))$.
  For every ground $\homFun$-component $\wqryElm \in
  \comp[\bot][\homFun]{\pqryElm}$, instead set $\trn{\homFun}(\varElm) \defeq
  \homFun(\varElm)$, for every $\varElm \in \ter{\wqryElm}$.
  The maps defined on the different $\homFun$-components are compatible.
  Indeed, if a variable occurs in two distinct $\homFun$-components, its image
  under $\homFun$ is a constant of $\ConSet(\qryElm)$, as shown in the proof
  of~\cref{thm:lenpar}.
  Since every $\homFun[\wqryElm]$ fixes these constants and the ground branch is
  the identity, the two definitions of $\trn{\homFun}$ agree.
  Moreover, every atom is mapped to a fact of $\dbElm[][*]$: this follows from
  $\homFun[\wqryElm] \in \HomSet(\iotaFun(\imath(\wqryElm)), \dbElm[][*])$ on a
  non-ground component, while on a ground component $\homFun$ maps into the
  selected ground facts, which occur in $\dbElm[][*]$.
  Hence, the componentwise definitions determine a homomorphism $\trn{\homFun}
  \in \HomSet(\pqryElm, \dbElm[][*])$.
  This construction is injective in the family of choices
  $\{ \homFun[\wqryElm] \}_{\wqryElm \in \dom{\imath}}$.
  Indeed, suppose that two such families differ on a component $\wqryElm$ and
  put $\rqryElm \defeq \imath(\wqryElm)$.
  The corresponding homomorphisms from $\iotaFun(\rqryElm)$ to $\dbElm[][*]$
  must differ on a variable, since they both fix constants.
  By~\cref{thm:juqmucchr}, $\rqryElm$ is join-uniform and so is its isomorphic
  copy $\iotaFun(\rqryElm)$.
  Consequently, every variable of $\iotaFun(\rqryElm)$ occurs in every one of
  its atoms.
  Since $\homFun(\wqryElm)$ is a non-empty subquery of $\iotaFun(\rqryElm)$,
  every such variable is therefore the image under $\homFun$ of some variable
  occurring in $\wqryElm$.
  Hence, two different choices on $\wqryElm$ induce different values of
  $\trn{\homFun}$ on some variable of $\wqryElm$.
  We have thus obtained an injection from $\prod_{\wqryElm \in \dom{\imath}}
  \HomSet(\iotaFun(\imath(\wqryElm)), \dbElm[][*])$ into $\HomSet(\pqryElm,
  \dbElm[][*])$.
  Since $\iotaFun(\imath(\wqryElm)) \approx \imath(\wqryElm)$ and
  $\der{\gFun}(\imath(\wqryElm)) = \solElm[\imath(\wqryElm)]$, for every
  $\wqryElm \in \dom{\imath}$, it follows that
  \[
    \PFun[\qryElm][\pqryElm](\GSet, \solVec)
  =
    \card{\HomSet(\pqryElm, \dbElm[][*])}
  \geq
    \prod_{\wqryElm \in \dom{\imath}} \card{\HomSet(\iotaFun(\imath(\wqryElm)),
    \dbElm[][*])}
  =
    \prod_{\wqryElm \in \dom{\imath}} \solElm[\imath(\wqryElm)],
  \]
  which proves the claim.
  Now let $\solVec^{\expVec[i]}$ be the unitary form of a monomial of
  $\PFun[\qryElm][\pqryElm](\GSet, \unkVec[\top])$ having positive coefficient.
  Since its collected coefficient is positive, at least one term contributing to
  it in the full distributive expansion has positive coefficient.
  Hence, there exist a profile $\imath \in \Prof[\qryElm][\pqryElm]{\GSet}$ and
  a product of one summand from every factor $\DeltaFun[\imath(\wqryElm)]$ whose
  coefficient is positive and whose unitary form is
  $\unkVec[\top]^{\expVec[i]}$.
  Fix such a product.
  For each $\wqryElm \in \dom{\imath}$, put $\rqryElm \defeq \imath(\wqryElm)$.
  If the selected summand of $\DeltaFun[\rqryElm]$ contains an unknown
  $\unkElm[\sqryElm]$, then $\sqryElm \preccurlyeq \rqryElm$ and the
  monotonicity proved above gives $\solElm[\sqryElm] \leq \solElm[\rqryElm]$.
  Otherwise, the selected summand comes from some ground $\sqryElm \in
  \LambdaSet[\GSet]$ with $\sqryElm \preccurlyeq \rqryElm$ and contributes $1$
  to the unitary form.
  Since $\der{\gFun}$ is the $\oVec$-extension of $\solVec$, the
  ground-coordinate identity $\mFun \star \chiFun[\GSet] =
  \chiFun[{\LambdaSet[\GSet]}]$ established in the proof of~\cref{thm:polchr}
  gives $\der{\gFun}(\sqryElm) = 1$.
  Thus, monotonicity yields $1 \leq \solElm[\rqryElm]$ in this case as well.
  Multiplying these inequalities over all $\wqryElm \in \dom{\imath}$ gives
  \[
    \solVec^{\expVec[i]}
  \leq
    \prod_{\wqryElm \in \dom{\imath}} \solElm[\imath(\wqryElm)]
  \leq
    \PFun[\qryElm][\pqryElm](\GSet, \solVec).
  \]
  Hence, $\PFun[\qryElm][\pqryElm](\GSet, \solVec) - \solVec^{\expVec[i]} \geq
  0$ for every monomial with positive coefficient.
  Since $\solVec$ was arbitrary, the polynomial is strongly non-negative \wrt
  $\PStr[\qryElm, \GSet]$.
\end{proof}

As far as decidability is concerned, one could simply construct the full
multiplicity poset appearing in \cref{lem:conprbhyp} and the corresponding
$n$-\MPI explicitly.
Indeed, by \cref{thm:juqmucchr}, the minimal unification closure
$\denot{\qryElm}$ of a \BJUQ is finite and its homomorphic preorder is a finite
partial order.
Moreover, once a non-trivial ground selection $\GSet$ is fixed, the multiplicity
weights, the offset vector, and the homomorphism profiles defining
$\PFun[\qryElm][\pqryElm](\GSet, \unkVec[\top])$ are finite objects determined
exactly by $\qryElm$, $\pqryElm$, and $\GSet$.
Hence, by enumerating all non-trivial ground selections $\GSet$ and applying
\cref{thm:decprc} to the corresponding instance of \cref{prb:nmpi}, one obtains
a decision procedure for non-containment and, therefore, for containment by
complementation.

To obtain the declared upper bound, however, we need to estimate the size of the
objects involved in this construction.
The proof of \cref{thm:juqmucchr} shows that no element of $\denot{\qryElm}$
uses relation symbols or constants that do not already occur in $\qryElm$.
It also shows that every non-ground element of $\denot{\qryElm}$ is
join-uniform and, therefore, all its variables occur in every one of its atoms.
Thus, the number of variables of a non-ground closure element is bounded by the
maximum number of variables in any component of $\qryElm$, which itself is
bounded by the arity of a relation symbol occurring in the query.
After choosing one atom as a distinguished anchor, all variables of the closure
element can be identified through the variables occurring in that atom.
With this convention, every other atom of the same closure element is obtained
from an atom occurring in a component of $\qryElm$ by replacing each position
with either one of the anchor variables or one of the constants of $\qryElm$.
Hence, the number of possible atom-shapes that may occur in a closure element is
at most exponential in $\len{\qryElm}$.
This exponential bound on atom-shapes does not imply, however, that
$\card{\denot{\qryElm}}$ itself is exponential.
A closure element is a finite set of such atom-shapes, and different subsets may
give rise to non-isomorphic closure elements.
Thus, in the \BJUQ case, the powerset blow-up gives a double-exponential upper
bound on the size of the full minimal unification closure.
More precisely, every representative in $\denot{\qryElm}$ has at most
exponentially many atoms and can be represented with exponential size, while the
number of representatives in $\denot{\qryElm}$ is bounded by
$\pow{\pow{\AOmicron{\len{\qryElm} \log \len{\qryElm}}}}$.
The same double-exponential bound controls the remaining data.
The order relation $\preccurlyeq$ on $\denot{\qryElm}$ contains at most
double-exponentially many comparable pairs.
Ground selections are not enumerated as arbitrary subsets of
$\denot{\qryElm}[\bot]$.
Rather, they are determined by the finite ground set $\PhiSet[\GSet]$, which
is a subset of the possible ground atom-shapes.
Since the latter family is exponential in $\len{\qryElm}$, the possible choices
of $\PhiSet[\GSet]$, and therefore of $\GSet$, are bounded double-exponentially.
The maximality condition (see \cref{def:grnsel}) is then checked once
$\PhiSet[\GSet]$ is fixed.
Finally, the homomorphism profiles contributing to
$\PFun[\qryElm][\pqryElm](\GSet, \unkVec[\top])$ are obtained by mapping the
containing query $\pqryElm$ into the closure structure determined by $\qryElm$
(see \cref{def:prf}).
Since this closure has double-exponential size and each representative has
exponential size, the full polynomial profile has size double-exponential in the
size of the containee query $\qryElm$ and exponential in the size of the
containing query $\pqryElm$.
Consequently, in the \BJUQ case, for a fixed non-trivial ground selection
$\GSet$, the associated instance of \cref{prb:nmpi} has double-exponential
(\resp, exponential) size in the containee (\resp, containing) query.
For {\BJoFQ}s, the same construction yields only exponential-size Diophantine
instances, since the minimal unification closure consists of atom-level
representatives.
For {\PFQ}s under bag-bag semantics, the projection-free analysis of~\cite{KM19}
gives the sharper polynomial-hierarchy bound recalled below.
We can now collect the resulting upper bounds.

\begin{theorem}
\label{thm:conprbsol}
  \CP[\bs](\JUQ, \CQ) is solvable in 2\CoNExpTime.
  \CP[\bs](\JoFQ, \CQ) and \CP[\bb](\JoFQ, \CQ) are solvable in \CoNExpTime.
  \CP[\bb](\PFQ, \CQ) is solvable in \ULH[2][P].
\end{theorem}
\begin{proof}
  We describe the complement procedure for non-containment, which is obtained
  by preceding the polynomial-time \QEA-alternating procedure reported in
  \cref{alg:nmpisol} with the guesses needed to reduce the containment problem
  to an instance of \cref{prb:nmpi}.
  The stated complexity bounds then follow by complementation.

  We start with \CP[\bs](\JUQ, \CQ).
  Let $(\qryElm, \pqryElm) \in \JUQ \times \CQ$ be an input pair and put $L
  \defeq \len{\qryElm} + \len{\pqryElm}$.
  First, in the non-Boolean case, one guesses a prototypical tuple for the free
  variables.
  This guess has size linear in $L$, and can in fact be avoided by using the
  most-general prototypical tuple, as discussed in~\cite{KM19} for {\PFQ}s.
  Second, one guesses a non-trivial ground selection $\GSet$ for the Boolean
  version of the containee query.
  As explained above, such a guess is represented through the finite ground set
  $\PhiSet[\GSet]$, and therefore has at most exponential size in
  $\len{\qryElm}$.
  Once these guesses are fixed, \cref{thm:polchr} reduces non-containment to the
  existence of a Diophantine natural solution of the associated $n$-\MPI.
  By \cref{lem:conprbhyp}, this is an instance of \cref{prb:nmpi}.
  The size analysis above shows that, for {\BJUQ}s, the minimal unification
  closure has size $N \leq \pow{\pow{\AOmicron{\len{\qryElm} \log
  \len{\qryElm}}}}$.
  The number of monomials occurring in the expanded polynomial is bounded by
  $N^{\AOmicron{\len{\pqryElm}}}$.
  This is still double-exponential in $L$.
  Hence the encoding of the induced Diophantine instance is bounded by
  $\pow{\pow{L^{k}}}$, for some fixed $k \in \SetR[+]$.

  The algorithm of \cref{alg:nmpisol} first guesses an existential certificate
  whose size is polynomial in the encoding of the induced homogeneous linear
  system.
  This certificate therefore has double-exponential size in $L$ and can be
  guessed nondeterministically within double-exponential time.
  The subsequent universal choices range over positive monomials and cover
  pairs.
  Their number is bounded by the size of the constructed Diophantine instance,
  and each check is polynomial in that instance.
  Thus, after the existential phase, all universal checks can be enumerated
  deterministically within double-exponential time.
  Consequently, non-containment for {\BJUQ}s is decidable in 2\NExpTime, which
  implies that \CP[\bs](\JUQ, \CQ) is decidable in 2\CoNExpTime.

  The soundness of this procedure follows from \cref{thm:polchr}.
  If one of the constructed Diophantine instances has a Diophantine natural
  solution, then the corresponding ground selection witnesses $\qryElm
  \not\bsinc \pqryElm$.
  Conversely, if $\qryElm \not\bsinc \pqryElm$, then \cref{thm:polchr} provides
  a non-trivial ground selection $\GSet$ whose associated Diophantine instance
  has a Diophantine natural solution, and this instance is considered by the
  procedure.

  Let us now consider \CP[\bs](\JoFQ, \CQ), for which the same argument gives a
  sharper bound.
  If the containee query $\qryElm$ is a \BJoFQ, then every component of the
  query is a single atom.
  Consequently, the elements of the minimal unification closure are atom-level
  representatives, as in the construction of~\cite{KM25}.
  Equivalently, each closure element is obtained, up to isomorphism, by unifying
  a set of atoms of the original query.
  Thus, the powerset blow-up described above for {\BJUQ}s does not occur inside
  each closure representative.
  The number of representatives in $\denot{\qryElm}$ is only exponential in
  $\len{\qryElm}$, and every representative is a single atom, hence has size
  polynomial in $\len{\qryElm}$.
  It follows that the induced multiplicity posets, the admissible ground
  selections, and the polynomial profiles defining
  $\PFun[\qryElm][\pqryElm](\GSet, \unkVec[\top])$ are exponential-size objects.
  The existential certificate guessed by \cref{alg:nmpisol} is therefore
  exponential-size, while all universal checks can be enumerated in exponential
  time.
  Hence non-containment is decidable in \NExpTime, and containment is decidable
  in \CoNExpTime.
  The reduction from bag-bag to bag-set containment preserves the join-on-free
  property and has linear overhead.
  Therefore, both \CP[\bs](\JoFQ, \CQ) and \CP[\bb](\JoFQ, \CQ) are solvable in
  \CoNExpTime.

  Finally, for {\PFQ}s under bag-bag semantics, the sharper bound
  from~\cite{KM19} applies.
  Indeed, the join-on-free closure associated with the \BJoFQ rewriting of a
  \PFQ has exactly the atom-level structure handled in the projection-free
  analysis, where the minimal unification closure contains, up to isomorphism,
  only the atom-level representatives occurring in the rewritten query.
  Hence, the minimal unification closure has size $N \leq \len{\qryElm}$.
  The existential certificate is therefore polynomial-size.
  Although the number of monomials may be exponential, each universal monomial
  can be represented by a polynomial-size homomorphism profile and checked in
  polynomial time.
  Thus the complement problem lies in the second existential level of the
  polynomial hierarchy, \ie, is in \ELH[2][P].
  Therefore, \CP[\bb](\PFQ,\CQ) is solvable in \ULH[2][P].
\end{proof}

We finally recall that an \NPTime lower bound already holds even for the
Boolean restrictions of the fragments considered here.
Indeed, the projection-free bag-bag containment problem is \NPTimeH, and the
remaining hardness results follow by the inclusions between containee classes
and by the reductions mentioned above.

\begin{theorem}
\label{thm:conprbhrd}
  All four $\CP[\bs](\BJUQ, \BCQ)$, $\CP[\bs](\BJoFQ, \BCQ)$, $\CP[\bb](\BJoFQ,
  \BCQ)$, and $\CP[\bb](\BPFQ, \BCQ)$ are \NPTimeH.
\end{theorem}
\begin{proof}
  The proof generalises the classic way to reduce the $3$-colourability problem
  to the set-containment problem.
  Let $\GName = \tuple{\VSet}{\ERel}$ be an arbitrary undirected graph, fix an
  arbitrary orientation of its edges, and let $\qryElm[\GName]$ be the Boolean
  query having one variable $x_{v}$ for every vertex $v \in \VSet$ and one atom
  $\RRel(x_{u}, x_{v})$ for every oriented edge $(u, v)$.
  Consider now the ground query $\qryElm[\TName] \leftarrow \RRel(\aSym,\bSym),
  \RRel(\bSym,\aSym), \RRel(\aSym,\cSym), \RRel(\cSym,\aSym),
  \RRel(\bSym,\cSym), \RRel(\cSym,\bSym)$, associated with the complete graph on
  three vertices.
  Since $\qryElm[\TName]$ contains an atom for every ordered pair of distinct
  elements of $\{\aSym,\bSym,\cSym\}$, there exists a homomorphism from
  $\qryElm[\GName]$ to $\qryElm[\TName]$ iff $\GName$ is $3$-colourable.
  Equivalently, $\GName$ is $3$-colourable iff $\qryElm[\TName]$ is
  set-contained into $\qryElm[\GName]$.
  We now show that, under either bag-set or bag-bag semantics, $\qryElm[\TName]$
  is bag-contained into $\qryElm[\TName] \wedge \qryElm[\GName]$ iff $\GName$ is
  $3$-colourable.

  \textbf{[$\Rightarrow$]}
  If $\qryElm[\TName]$ is bag-contained into $\qryElm[\TName] \wedge
  \qryElm[\GName]$, it immediately follows that $\qryElm[\TName]$ is
  set-contained into $\qryElm[\GName]$ as well.
  Thus, by the above observation, $\GName$ is $3$-colourable.

  \textbf{[$\Leftarrow$]}
  Suppose that $\GName$ is $3$-colourable.
  Then, there exists a homomorphism from $\qryElm[\GName]$ to $\qryElm[\TName]$.
  Therefore, whenever $\qryElm[\TName]$ has positive multiplicity on an
  instance, so does $\qryElm[\GName]$.
  Moreover, the multiplicity of $\qryElm[\TName] \wedge \qryElm[\GName]$ is the
  product of the multiplicities of $\qryElm[\TName]$ and $\qryElm[\GName]$ in
  isolation, on any database instance, since $\qryElm[\TName]$ is ground.
  Hence, the multiplicity of $\qryElm[\TName] \wedge \qryElm[\GName]$ is at
  least the multiplicity of $\qryElm[\TName]$, and the containment follows.

  The construction is polynomial in the size of $\GName$.
  Finally, $\qryElm[\TName]$ is ground and therefore a \BPFQ, hence also a
  \BJoFQ and a \BJUQ.
  Since the argument applies under both bag-set and bag-bag semantics, it proves
  the \NPTime-hardness of all four problems in the statement.
\end{proof}




\section{Discussion}

We developed a unified framework for attacking bag containment for classes of
conjunctive queries through the internal unification structure of the containee
query under analysis.
The framework covers the two main decidable fragments previously obtained by
this line of work, namely \emph{projection-free queries} and \emph{join-on-free
queries}, and extends them to the broader class of \emph{join-uniform queries}.
In all these cases, the containing query is left unrestricted.

The adopted solution strategy can be broken down into two independent but
synergistic contributions.
On the one hand, we provide a \emph{finite arithmetic representation} of
multiplicities induced by the \emph{minimal unification closure} of the
containee query.
This is done through a characterisation of the containment problem in terms of
\emph{multicanonical instances}, which isolate the relevant witnesses for
non-containment and reduce the comparison between two queries to a comparison
between their respective \emph{homomorphism counts} over such instances.
In particular, these counts can be computed by means of a \emph{monomial
function}, for the containee query, and a \emph{polynomial function}, for the
containing one, the latter enjoying the crucial property of \emph{strong
non-negativeness}.
The central difference with the earlier projection-free and join-on-free
settings is that the parameters whose values describe a witness of
non-containment are no longer plain image cardinalities, but arbitrary
homomorphism cardinalities.
These parameters are linked through a \emph{weighted M\"obius inversion} over
the multiplicity partial order associated with the unification closure.
On the other hand, we solve a special case of the \emph{Diophantine inequality
problem}.
Although Diophantine inequalities are undecidable in general, the inequalities
arising here have a special \emph{monomial-polynomial} form and interact with
the multiplicity partial order in a controlled way, thanks to the strong
non-negativeness of the polynomials.
Consequently, the usual source of undecidability in bag containment is not
avoided, but constrained tightly enough to become algorithmically usable.
The main difference with the approach in the previous setting is that the
Diophantine analysis no longer relies on the underlying order being a
meet-semilattice.
In the join-on-free framework, this lattice structure was used to reduce the
problem to a strict version in which comparable unknowns take strictly ordered
values.
Here, instead, we prove directly that every admissible solution over the given
multiplicity poset can be transformed into an admissible strict solution of the
same monomial-polynomial inequality.
This step removes the need to restrict the original order to a meet-semilattice,
and allows the Diophantine argument to work over an arbitrary finite partial
order equipped with multiplicity weights.

This work leaves a few interesting directions for future investigation.
The first one concerns complexity.
While the containment problems for \BPFQ, \BJoFQ, and \BJUQ studied here are
placed in \ULH[2][P], \CoNExpTime, and 2\CoNExpTime, respectively, and are known
to be \NPTimeH, identifying tighter bounds remains an open challenge.
On the upper-bound side, it would be valuable to understand which part of the
exponential blow-up due to the unification closure is avoidable.
On the lower-bound side, instead, since the withdrawal of the
\ULH[2][P]-hardness claim by Chaudhuri and Vardi for the general bag containment
problem, progress in this area has stalled, making it an appealing topic for
further exploration.
A second direction is to identify the exact frontier beyond join-uniform
queries.
The proof suggests that join-uniformity is sufficient because it guarantees
finiteness of the unification closure, a well-behaved containment order, and
the weighted M\"obius properties needed for realisability.
It is not yet clear whether some of these requirements can be weakened, \eg,
finiteness of the closure, leading to a larger decidable class.
Finally, a third direction is the integration of the techniques discussed in
this work with those that obtain decidability by restricting the containing
query~\cite{KR11,KKNS20,KKNS21}.
Such an integration may be a necessary step towards a final solution of this
open problem.









  \bibliographystyle{alphaurl}
  \bibliography{References}

\end{document}